\documentclass[11pt,a4paper]{article}
\usepackage[a4paper]{geometry}
\usepackage{amssymb,latexsym,amsmath,amsfonts,amsthm}
\allowdisplaybreaks
\usepackage{graphicx}
\usepackage{epstopdf}
\usepackage{float,amsmath,amssymb,mathrsfs,bm,multirow,graphics}
\usepackage{tikz}
\usepackage{pgflibraryshapes}
\usetikzlibrary{arrows,decorations.markings}
\usepackage{subfig}
\usepackage[percent]{overpic}
\usepackage{fullpage}
\usepackage{comment,verbatim}
\usepackage{hyperref}
\usepackage{color}
\usepackage{mathrsfs}
\usepackage[title,titletoc]{appendix}
\usepackage{ulem}
\usepackage{enumitem}
\usepackage{mathabx}
\DeclareMathOperator{\diag}{diag}
\DeclareMathOperator*{\Res}{Res}

\newcommand{\C}{\mathbb{C}}

\renewcommand{\Re}{\mathrm{Re}\,}
\renewcommand{\Im}{\mathrm{Im}\,}
\newcommand{\Ai}{\mathrm{Ai}}
\newcommand{\PC}{\mathrm{PC}}

\newcommand{\ud}{\,\mathrm{d}}

\newcommand{\what}{\widehat}

\newcommand{\ii}{\mathrm{i}}
\newcommand{\I}{\mathcal I}
\newcommand{\J}{\mathcal J}

\newcommand{\Boh}{\mathcal{O}}

\newtheorem{theorem}{Theorem}[section]
\newtheorem{lemma}[theorem]{Lemma}
\newtheorem{remark}[theorem]{Remark}
\newtheorem{proposition}[theorem]{Proposition}

\newtheorem{rhp}[theorem]{RH problem}
\newtheorem{assumption}[theorem]{Assumption}
\theoremstyle{definition}

\theoremstyle{remark}

\numberwithin{equation}{section}

\begin{document}
	%\title{Asymptotics of the particular solution of noncommutative Painlev\'{e} II equation}
    \title{Asymptotics for the noncommutative Painlev\'{e} II equation}
	
	\author{Junwen Liu\footnotemark[1],\quad Luming Yao\footnotemark[2], \quad Lun Zhang\footnotemark[1]}
	
	\renewcommand{\thefootnote}{\fnsymbol{footnote}}
	\footnotetext[1]{School of Mathematical Sciences, Fudan University, Shanghai 200433, P. R. China. E-mail: \{junwenliu21\symbol{'100}m.fudan.edu.cn, lunzhang\symbol{'100}fudan.edu.cn.\}}
	\footnotetext[2]{Institute for Advanced Study, Shenzhen University, Shenzhen 518060, P. R. China. E-mail: lumingyao\symbol{'100}szu.edu.cn.}
	%\footnotetext[3]{Shanghai Key Laboratory for Contemporary Applied Mathematics, Fudan University, Shanghai 200433, P. R. China.}
	
	\date{\today}
	
	\maketitle
	\begin{abstract}
In this paper, we are concerned with the following noncommutative Painlev\'{e} II equation
\begin{equation*}
\mathbf{D}^2 \beta_1 = 4\mathbf{s} \beta_1 +4 \beta_1 \mathbf{s} +8 \beta_1^3,
\end{equation*}
where  $\beta_1=\beta_1(\vec{s})$ is an $n \times n$ matrix-valued function of $\vec{s}=(s_1,\ldots,s_n)$,  $\mathbf{s}=\diag(s_1,\ldots,s_n)$ and
$\mathbf{D}=\sum_{j=1}^n\frac{\partial}{\partial s_j}$. If $n=1$, it reduces to the classical Painlev\'{e} II equation up to a scaling. Given an  arbitrary $n \times n$ constant matrix $C=\left(c_{j k}\right)_{j, k=1}^n$, a remarkable result due to Bertola and Cafasso asserts that there exists a unique solution $\beta_1(\vec{s})=\beta_1(\vec{s};C)$ of the noncommutative PII equation  such that its $(k,l)$-th entry behaves like  $-c_{kl} \Ai (s_k+s_l)$ as $S= \frac{1}{n}\sum_{i=1}^n s_j\to+\infty$, where $\Ai$ stands for the standard Airy function. For a class of structured matrices $C$, we establish 
asymptotics of the associated solutions as $S \to -\infty$, which particularly include the so-called connection formulas. In the present setting, it comes out that the solution exhibits a hybrid behavior in the sense that 
each entry corresponds to either an extension of the Hastings–McLeod solution or an extension of the Ablowitz--Segur solution for the PII equation. It is worthwhile to emphasize the asymptotics of the $(k,l)$-th entry as $S \to -\infty$ cannot be deduced solely from its behavior as $S \to +\infty$ in general, which actually also depends on the positive infinity asymptotics of the $(l,k)$-th entry. This new and intriguing phenomenon disappears in the scalar case.

	% In this paper, we consider a class of matrix-valued, pole-free solutions to the noncommutative (or matrix) Painlev\'e II equation, a natural generalization of the classical Painlev\'e II equation arising in integrable systems and random matrix theory. These solutions are characterized by their connection to a matrix Airy-type convolution operator and are rigorously analyzed through a $2n \times 2n$ Riemann--Hilbert (RH) problem framework. Using the nonlinear steepest descent method of Deift and Zhou, we derive precise asymptotic expansions and connection formulas as the mean parameter $S \to -\infty$, complementing earlier results for $S \to +\infty$.

% Our analysis reveals that these matrix solutions exhibit a hybrid behavior, combining both features of Hastings--McLeod type and Ablowitz--Segur type solutions. A distinctive feature emerging in the matrix case is that the asymptotic behavior of individual matrix entries cannot be determined autonomously, but rather depends crucially on the asymptotic behavior of interrelated matrix components in the Ablowitz--Segur case.
	\end{abstract}

% \tableofcontents

\section{Introduction}\label{sec:intro}
\renewcommand{\thefootnote}{\arabic{footnote}}
% 	The  Painlevé II equation
%  \begin{align}\label{eq:PII}
%	q''(x)=2q(x)^3+xq(x)-\alpha, \qquad \alpha>-1/2,
%  \end{align}	
% is a cornerstone of integrable systems theory, renowned for its role in describing critical phenomena such as universality in random matrix theory, nonlinear wave modulation, and transition asymptotics near singularities \cite{SUL}. Among its two most celebrated solutions that emerge in the integral representation of the Airy Fredholm determinant, the Hastings--McLeod solution\cite{Fokas2006,HM1980} has become pivotal in characterizing the Tracy-Widom distribution for the largest eigenvalue of Gaussian unitary ensembles. Similarly, the Ablowitz--Segur solution \cite{SA} underpins connections to inverse scattering and soliton theory. These solutions exemplify the profound interplay between integrability, asymptotics, and physical applications.

The second Painlev\'e (PII) equation  
%is the following second-order nonlinear differential equation
\begin{align}\label{eq:PII}
	q''(s)=2q(s)^3+sq(s)-\alpha,  \quad \alpha\in\mathbb{C},
\end{align}	
%where $\alpha$ is a constant. This equation, 
together with five other nonlinear ordinary differential equations, was first studied by Painlev\'e and his colleagues at the turn of the twentieth century; cf. \cite{Ince:book} for the historical background and more information.  
%see more information about the Painlev\'e equation and the historical background in \cite{Ince:book}. 
In general, the Painlev\'e  equations are irreducible, in the sense that they cannot be solved in terms of closed forms of elementary functions or known classical special functions (like Airy, Bessel, hypergeometric functions, etc.). Thus, the solutions of Painlev\'e equations are called Painlev\'e transcendents. All Painlev\'e transcendents satisfy the so-called Painlev\'e property: the only movable singularities are poles. Here, ``movable" means that the locations of singularities depend on the constants of integration associated with the initial or boundary conditions; cf. \cite{DLMF}. Due to the extensive connections between Painlev\'e equations and many branches of mathematics and physics, the Painlev\'e transcendents are nowadays widely recognized as the nonlinear counterparts of the classical special functions.

For the homogeneous PII equation, i.e., $\alpha=0$ in \eqref{eq:PII}, by assuming that the solutions tend to zero exponentially fast as $s \to +\infty$, it is readily seen that the equation should be well-approximated by the classical Airy equation $y''(s)=sy(s)$. A result due to Hastings and McLeod \cite{HM1980} asserts that, for any $k$, there exists a unique solution to the homogeneous PII equation which behaves like $k\Ai(s)$ for large positive $s$, where $\Ai$ stands for the standard Airy function (cf. \cite[Chapter 9]{DLMF}). Depending on different values of $k$, we encounter several classes of well-known solutions of the PII equation, which we denote by $q(s;k)$ in what follows.
\begin{itemize}
  \item If $k=\pm 1$, the solutions are known as the Hastings--McLeod solutions \cite{HM1980}. They are characterized by the asymptotics
\begin{equation}\label{HM}
q(s;\pm 1)=\left\{
                   \begin{array}{ll}
                     \pm \sqrt{-\frac{s}{2}}+\Boh(|s|^{-5/2}), & \hbox{ $s \to -\infty$,} \\
                     \pm \Ai(s)+\Boh\left(\frac{e^{-(4/3)s^{3/2}}}{s^{1/4}}\right), & \hbox{ $s \to +\infty$;}
                   \end{array}
                 \right.
\end{equation}
see \cite{DZ95} for the detailed derivation of the above asymptotic formulas. The Hastings--McLeod solutions $q(s;\pm 1)$ are real and pole-free on the real axis.

\item If $-1<k<1$, the one-parameter family of solutions is known as the real Ablowitz--Segur solutions \cite{AS1,SA}. They are also real and pole-free on the real line \cite{AS3}, 
with the asymptotics
\begin{equation}\label{AS}
q(s;k)= \begin{cases}
\frac{\sqrt{-2\chi}}{(-s)^{1/4}}\cos\left(\frac{2}{3}(-s)^{3/2}+\chi\ln(8(-s)^{3/2})+\phi\right)  +\Boh\left(\frac{\ln |s|}{|s|^{5/4}}\right), &  s\to -\infty, \\
k\Ai(s)+\Boh\left(\frac{e^{-(4/3)s^{3/2}}}{s^{1/4}}\right), &  s\to +\infty.
\end{cases}
\end{equation}
Here,
\begin{equation}\label{chi}
\chi:=\frac{1}{2\pi}\ln(1-k^2), \quad \phi:=-\frac{\pi}{4}-\arg\Gamma(\ii\chi)-\arg(-k\ii),
\end{equation}
where $\Gamma(z)$ is the Gamma function. In the case $k=0$, we simply have $q(s;0)=0$.

\item One can extend Ablowitz--Segur solutions to more general $k$. Indeed, if $k\in\mathbb{C}\setminus((-\infty,-1]\cup[1,\infty))$, it follows from \cite{BCI2016} that $q(s;k)$ is pole-free on the real line with the asymptotics
% If There also exist complex Ablowitz--Segur solutions, which correspond to complex $k$; see \cite{BCI2016}. They have no poles at real values of $s$, with the asymptotic behavior
\begin{equation}\label{Complex}
q(s;k)= \begin{cases}
\frac{\sqrt{-2\chi}}{(-s)^{1/4}}\sin\left(\frac{2}{3}(-s)^{3/2}+\chi\ln(8(-s)^{3/2})+\tilde\phi\right)  +\Boh\left(\frac{1}{|s|^{{2-3|\Im \chi|}}}\right), &  s\to -\infty, \\
k\Ai(s)+\Boh\left(\frac{e^{-(4/3)s^{3/2}}}{s^{1/4}}\right), &   s\to +\infty.
\end{cases}
\end{equation}
Here, $\chi$ is defined in \eqref{chi} with $|\Im \chi|<\frac 12$ and
\begin{equation}\label{def:tildephi}
 \tilde\phi:=-\frac{\pi}{4}-\frac{\ii}{2}\ln\frac{\Gamma(-\ii \chi)}{\Gamma(\ii \chi)}^{\footnotemark}. 
\end{equation}
\footnotetext{It seems that there is a typo in \cite[Equation (1.11)]{BCI2016}, the coefficient of logarithmic term $\ln\frac{\Gamma(-\ii \chi)}{\Gamma(\ii \chi)}$ should be $-\ii/2$ instead of $-\ii$.}The asymptotic formula \eqref{Complex} reduces to \eqref{AS} if $-1<k<1$. 
\end{itemize}

The one-parameter family of solutions $q(s;k)$ plays an important role in a variety of different areas of mathematics, ranging from transient asymptotics of integrable differential equations (cf. \cite{Monvel,FN1980,XYZ}) to random matrix theory and related fields (cf. \cite{Duits,FW2015}).  In particular, we note that
% For instance, just to name a few, similarity reductions of the modified Korteweg-de Vries equation satisfy \eqref{eq:PII}; see \cite{FN1980}. The cumulative distribution function of the celebrated Tracy-Widom distribution \cite{TW1994}, which describes the limiting distribution of the largest eigenvalue for Gaussian unitary ensembles, can be represented as an integral involving the so-called Hastings--McLeod solution \cite{HM1980} of \eqref{eq:PII} with $\alpha=0$, $\gamma=1$
        \begin{align}\label{eq:TW}
          \det \left(I- \gamma\mathcal{K}^{\Ai}_{(s,+\infty)}\right)= \exp\left(-\int_{s}^{\infty} (x-s) q(x;\gamma)^2\ud x\right), \quad 0<\gamma\leq 1,
        \end{align}
 where $\mathcal{K}^{\Ai}_{(s,+\infty)}$ is the trace-class operator acting on $L^{2} (s,+\infty)$ with the Airy
kernel 
\begin{align}
    K^{\Ai}(x,y):=\int_0^\infty 
    \Ai(x+z)\Ai(y+z)dz=\frac{\Ai(x)\Ai'(y)-\Ai'(x)\Ai(y)}{x-y}.
\end{align}
If $\gamma=1$, the left-hand side of \eqref{eq:TW} is the celebrated Tracy-Widom distribution \cite{TW1994}, which describes the limiting distribution of the largest eigenvalue for a large class of unitary ensembles \cite{Deift1999}. If $\gamma \in (0,1)$, the Fredholm determinant can be interpreted as the limiting distribution of the largest observed eigenvalue after thinning; see \cite{Boh1,Boh04,Boh06} for more information.

In the context of infinite Toda system, Retakh and Rubtsov introduced a fully noncommutative analog of the PII equation \cite{RR2010}. It reads 
\begin{align}\label{noncommutative equation}
			\mathbf{D}^2 \beta_1 = 4\mathbf{s} \beta_1 +4 \beta_1 \mathbf{s} +8 \beta_1^3, \quad \mathbf{s} := \operatorname{diag}(s_1, \ldots, s_n), \quad
            \mathbf{D} := \sum_{j=1}^{n} \frac{\partial}{\partial s_j}, \quad n\in\mathbb{N}, 
		\end{align}
where $\beta_1=\beta_1(\vec{s})$ is an $n \times n$ matrix-valued function of $\vec{s}:=\left(s_1, \ldots, s_n\right)$. If $n=1$, one has 
       \begin{align}\label{eq:beta1}
           \beta_1(s_1) =\sqrt{2} q(2\sqrt{2}s_1),
       \end{align}
where $q$ satisfies the homogeneous PII equation. A special rational solution of \eqref{noncommutative equation} was also constructed in \cite{RR2010} by using the theory of quasideterminants \cite{GGRW}. We refer to \cite{IP2020,Rum1,Rum2} for the appearance of noncommutative PII equation in description of Tracy-Widom distribution function of the general $\beta$-ensembles with the even values of the parameter $\beta$; see also \cite{BCR2018} for its relation with systems of Calogero type. 

Let $C=\left(c_{j k}\right)_{j, k=1}^n$ be an arbitrary $n \times n$ constant matrix, and set
\begin{align}\label{def:S}
    S:= \frac{1}{n}\sum_{i=1}^n s_j, \qquad \delta_i:=s_i-S, \quad i=1,\ldots,n. 
\end{align} 
A remarkable result established by Bertola and Cafasso \cite{MM12} asserts that there exists a unique solution $\beta_1(\vec{s})=\beta_1(\vec{s};C)$ of the noncommutative PII equation \eqref{noncommutative equation} such that 
\begin{align}\label{+infty asympototics}
			(\beta_1)_{kl} = -c_{kl} \Ai (s_k+s_l)+\Boh \left(\sqrt{S} e^{-\frac{4}{3}(2S-2\epsilon S)^\frac{3}{2}}\right),\qquad S\to+\infty,
            %^{\footnotemark}.
            %^\ddag.
		\end{align} 
        %\setcounter{footnote}{3}
        %\footnotetext{In \cite{MM12}, $\epsilon S$ is replaced by the same constant $m$.}
with $\lvert \delta_j \rvert \leq \epsilon S$,\footnote{In \cite{MM12}, the result is stated for $|\delta_j|\leq m$ with the error bound in  \eqref{+infty asympototics} given by $\Boh \left(\sqrt{S} e^{-\frac{4}{3}(2S-2m)^\frac{3}{2}}\right)$. The same proof methodology therein, however, can be adapted to relax the constraint on $\delta_j$, which leads to a more general conclusion as stated in \eqref{+infty asympototics}.} where $\epsilon \in [0,1)$ is an arbitrary real number and $(\beta_1)_{kl}$ stands for the $(k,l)$-th entry of $\beta_1$. Clearly, this conclusion generalizes that in \cite{HM1980}, which corresponds to $n=1$. In addition, it was also shown in \cite{MM12} that if the singular values of $C$ lie in $[0, 1]$, then the associated solution is pole free for $\vec{s} \in \mathbb{R}^n$.  

In view of the strikingly different asymptotic behaviors of $q$ as shown in \eqref{HM}--\eqref{def:tildephi}, a natural and challenging problem is then to establish the asymptotics of $\beta_1(\vec{s};C)$ as $S\to -\infty$. It is the aim of the present work to resolve this problem for a class of structured matrices $C$, which particularly gives the so-called connection formulas for this family of special solutions. Our results are stated in the next section.

 \section{Main results}
 We start with our assumption on the $n\times n$ matrix $C$ in \eqref{+infty asympototics}, which will be used throughout this paper. 
 \begin{assumption}\label{assump}
It is assumed that $C= \Lambda P$, where $\Lambda= \diag (\mu_1,\ldots,\mu_n)$ with $\mu_i\in \mathbb{C}$, $\lvert \mu_i \rvert \leq 1, i=1,\ldots,n$, and $P$ is a permutation matrix such that $C^2$ is a diagonal matrix. 
 \end{assumption}
 
Under the above assumption, the matrix $C$ can be viewed as a block matrix, where some blocks are diagonal matrices and the others are off-diagonal matrices. Thus, the  diagonal matrices and the off-diagonal matrices satisfy Assumption \ref{assump}. If $n=3$, $C$ takes one of the following forms:
            \begin{align}
                \begin{pmatrix}
                    \mu_1 & 0& 0\\
                    0 & \mu_2 & 0\\
                    0 & 0&\mu_3
                \end{pmatrix}, \quad \begin{pmatrix}
                    0 & 0& \mu_1\\
                    0 & \mu_2 & 0\\
                    \mu_3 & 0&0
                \end{pmatrix}, \quad \begin{pmatrix}
                    0 & \mu_1& 0\\
                    \mu_2 & 0 & 0\\
                    0 & 0&\mu_3
                \end{pmatrix}, \quad\begin{pmatrix}
                    \mu_1 & 0& 0\\
                    0 & 0 & \mu_2 \\
                    0 & \mu_3 & 0
                \end{pmatrix}.
            \end{align}
% In , $C$ can be viewed as a block matrix, where some blocks are diagonal matrices and others are off-diagonal matrices.

The condition that $C^2$ is a diagonal matrix implies that $P^2$ is an identity matrix.  Together with $\lvert \mu_i \rvert \leq 1$, $i=1,\ldots,n$, it follows that the singular values of $C$ lie in $[0, 1]$, which shows the associated solution $\beta_1(\vec{s};C)$ is pole free for $\vec{s} \in \mathbb{R}^n$. 

% then $CC^*$ is also a diagonal matrix, so the singular values of $C$ lie in $[0, 1]$ for $\lvert \mu_i \rvert \leq 1, i=1, \cdots,n$. Thus, $\beta_1(\vec{s})$ is pole free for $\vec{s} \in \mathbb{R}^n$. 

We are now ready to state large negative $S$ asymptotics of $\beta_1$.

	%Here are our results:
	\begin{theorem}\label{-infty asympototics results}
		% Assume $C= \Lambda P$, where $\Lambda= \diag (\mu_1,\cdots,\mu_n)$, $\lvert \mu_i \rvert \leq 1, i=1, \cdots,n$; $P$ is a permutation matrix and $C^2$ is a diagonal matrix. 
        Let the $n\times n$ matrix-valued function $\beta_1(\vec{s})= \beta_1(\vec{s};C)$ be the unique solution of the noncommutative PII equation \eqref{noncommutative equation} characterized by the asymptotics \eqref{+infty asympototics}. With $S$ and $\delta_i$ defined in \eqref{def:S}, we have, under Assumption \ref{assump} on $C$,  
		\begin{align}\label{-infty asympototics}
			(\beta_1)_{kl}= \begin{cases}
				\sqrt{\frac{-s_k-s_l}{2}}c_{kl}+\Boh(S^{-1}), \quad &c_{kl}c_{lk}=1,\\
                \frac{ (-s_k-s_l)^{-\frac 14}}{\sqrt{\pi }} \cos \left(\ii \left(\what \theta_k\left(\sqrt{\frac{s_k+s_l}{S}}\right)+\what \theta_l\left(\sqrt{\frac{s_k+s_l}{S}}\right)\right)- \frac{\pi}{4} \right)c_{kl}+ \Boh (S^{-1}), \quad &c_{kl}c_{lk}=0,\\
				(-s_k - s_l)^{-\frac 14} \sqrt{\frac{-\ln \left(1-c_{kl}c_{lk}\right)}{\pi c_{kl}c_{lk}}}\cos \psi(s_k,s_l) c_{kl} + \Boh (S^{-1}), \quad &c_{kl}c_{lk}\neq 0,1,
			\end{cases}
		\end{align}
		if $S \to -\infty$ and $\delta_i = \epsilon_i S$ with $\epsilon_i \in (-1,1)$ being fixed, where the function $\psi(s_k,s_l)$ is related to the parameters $c_{kl}$ and $c_{lk}$ through the connection formula
        \begin{align}\label{def:psi}
           \psi(s_k,s_l) &:=  \ii \left(\what \theta_k\left(\sqrt{\frac{s_k+s_l}{S}}\right)+\what \theta_l\left(\sqrt{\frac{s_k+s_l}{S}}\right)\right) +\frac{3}{4\pi}\ln (1-c_{kl}c_{lk}) \ln \left(-4(s_k+s_l)\right)\nonumber\\
           &\quad +\frac{\ii}{2}\ln\frac{\Gamma\left(-\frac{\ln(1-c_{kl}c_{lk})}{2 \pi \ii}\right)}{\Gamma\left(\frac{\ln(1-c_{kl}c_{lk})}{2 \pi \ii}\right)}+ \frac{\pi}{4}
           %&\quad + \arg \Gamma \left(\frac{\ln (1- c_{kl}c_{lk})}{2 \pi \ii}\right) + \frac{\pi}{4}
        \end{align}
        with 
        \begin{equation}\label{def:what-theta}
            \what \theta_k (z) := \ii (-S)^{\frac 32} \left(\frac{z^3}{6} - \frac{s_k}{S}z\right).
        \end{equation}
	\end{theorem} 
For $n=1$, the formulas \eqref{-infty asympototics} are equivalent to those given in \eqref{HM}, \eqref{AS} and \eqref{Complex} by noting the relation \eqref{eq:beta1}. For general $n>1$, our result reveals special features of $\beta_1$ under the present setting. On the one hand, it comes out that $(\beta_1)_{kl}$ corresponds to either an extension of the Hastings–McLeod solution or an extension of the Ablowitz--Segur solution for the PII equation, depending on the value of the product $c_{kl}c_{lk}$. Consequently, $\beta_1$ in general exhibits a hybrid behavior, combining both features of Hastings--McLeod type and Ablowitz--Segur type solutions. On the other hand, it is worthwhile to emphasize the asymptotic behavior of $(\beta_1)_{kl}$ as $S \to -\infty$ cannot be deduced solely from its behavior as $S \to +\infty$ in the Ablowitz--Segur case, which actually also depends on the positive infinity asymptotics of $(\beta_1)_{lk}$. This new and intriguing phenomenon disappears in the scalar case.

It remains open to establish the large negative $S$ asymptotics of $\beta_1$ without Assumption \ref{assump} on $C$. Even for $n=2$, this question is highly nontrivial. If $C$ is a $2 \times 2$ Hermitian matrix with eigenvalues in $(-1,1)$, a recent work of Du, Xu and Zhao \cite{DXZ25} shows that the asymptotics of $\beta_1$ are expressed in terms of one-parameter family of special solutions of the Painlev\'{e} V equation in a different asymptotic regime from us. We believe the asymptotics of $\beta_1$ will exhibit rich asymptotic behaviors for general matrices $C$, and leave this problem for future investigations.

\paragraph{Strategy of proof}
To prove Theorem \ref{-infty asympototics results}, our starting point is a profound representation of $\beta_1$ in terms of the Fredholm determinant of a matrix version of the Airy-convolution operator $\mathcal{A}\boldsymbol{\mathit{i}}
_{\vec{s}}$ acting on $L^2\left(\mathbb{R}_{+}, \mathbb{C}^n\right)$. More precisely, this operator reads 
\begin{align}\label{def:Ai}
    \left(\mathcal{A}\boldsymbol{\mathit{i}}_{\vec{s}} \vec{f}\right)(x) & :=\int_{\mathbb{R}_{+}} \mathbf{A i}(x+y ; \vec{s}) \vec{f}(y) d y
\end{align}
with $\vec{f}:=\left(f_1, \cdots, f_n\right)^T$ and
\begin{align}
    \mathbf{A i}(x ; \vec{s}) & :=\int_{\gamma_{+}} \mathrm{e}^{\theta(\mu)} C \mathrm{e}^{\theta(\mu)} \mathrm{e}^{\ii x \mu} \frac{d \mu}{2 \pi}=\left(c_{j k} \operatorname{Ai}\left(x+s_j+s_k\right)\right)_{j, k=1}^n,
\end{align}
where 
\begin{align}\label{def:theta}
    \theta(\mu) & :=\ii\operatorname{diag}\left(\frac{\mu^3}{6}+s_1\mu, \frac{\mu^3}{6}+s_2\mu, \ldots, \frac{\mu^3}{6}+s_n\mu\right)
\end{align}
and $\gamma_{+}$ is an infinite contour in the upper half plane that is asymptotic to the straight lines with arguments  $\frac{\pi}{2} \pm \frac{\pi}{3}$ at infinity.  It was shown in \cite{MM12} that
\begin{align}\label{det:Ai}
\operatorname{det}\left(I-\mathcal{A}\boldsymbol{\mathit{i}}
_{\vec{s}}^2\right)=\exp \left(-4 \int_s^{\infty}(t-s) \operatorname{Tr}\left(\beta_1(t+\vec{\delta})^2\right) d t\right),
\end{align}
where $\beta_1$ is the unique solution of the noncommutative PII equation \eqref{noncommutative equation} characterized by the asymptotics \eqref{+infty asympototics},
$\vec{\delta}:=\left(\delta_1, \ldots, \delta_n\right)$
with $\delta_i$ defined in \eqref{def:S}, and $t+\vec{\delta}:=(t+\delta_1,\ldots,t+\delta_n)$. This relation provides a noncommutative (matrix) version of the Tracy–Widom distribution \eqref{eq:TW}. 

By extending the general theory of integrable operators of Its-Izergin-Korepin-Slavnov \cite{IIKS} to operators of Hankel form, Bertola and Cafasso in \cite{MM12} constructed a $2n \times 2n$ Riemann–Hilbert (RH) problem to compute the resolvent operator of the matrix convolution operator $\mathcal{A}\boldsymbol{\mathit{i}}_{\vec{s}}$. This in turn yields an RH characterization of $\beta_1$; see \eqref{def:beta1} below. While it is straightforward to show \eqref{+infty asympototics} from this relation, one encounters significant obstacles when performing Deift-Zhou steepest descent analysis \cite{Deift1999,Deift1993} to the associated RH problem for large negative $S$ and a general constant matrix $C$. Indeed, even in the scalar case, the analytical techniques to deal with the Hastings--McLeod and Ablowitz--Segur solutions are quite different. The hybrid feature of $\beta_1$ and large size of the RH problem thus leads to the difficulty of the present work. To address the technical challenges, a key observation here is that  Assumption \ref{assump} on $C$ allows us to decompose the original RH problem into two RH problems by introducing proper index sets associated with the permutation matrix $P$; see \eqref{IJ} below. The advantage of this decomposition is that it decouples the different features of the solution, which are ready to analyze separately.

The rest of this paper is organized as follows. In Section \ref{sec:3}, we recall the $2n \times 2n$ RH problem formulated in \cite{MM12}, which characterizes $\beta_1$. After the decomposition of the rescaled RH problem, we split the entire analysis into two parts: the asymptotic analysis corresponding to a higher-dimension extension of the Hastings–McLeod solution in Section \ref{sec:I} and that for a higher-dimension extension of the Ablowitz--Segur solution in Section \ref{sec:II}. The proof of Theorem \ref{-infty asympototics} is presented Section \ref{proof}, as outcomes of RH analysis.

    %Although there are some subtle differences in the Riemann–Hilbert analysis between these two parts, the idea of handling higher-dimension RH problems remains essentially the same.
    
    \paragraph{Notations} Throughout this paper, the following notations are frequently used.
	\begin{itemize}
		\item If $M$ is a matrix, then $(M)_{ij}$ stands for its $(i,j)$-th entry. A $2n \times 2n$ matrix $M$ can also be viewed as a $2 \times 2$ block matrix and each block is an $n\times n$ matrix. The $(i,j)$-th block is denoted by $[M]_{ij}$ for $i,j = 1,2$.	
        
        \item 
        We use $I$ to denote an identity matrix, and the size might differ in different contexts. To emphasize a $k\times k$ identity matrix, we also use the notation $I_k$.
		\item  It is notationally convenient to denote by $E_{jk}$ the $n \times n$ elementary matrix
		whose entries are all $0$, except for the $(j,k)$-entry, which is $1$, that is,
		\begin{equation}\label{def:Eij}
			E_{jk}=\left( \delta_{l,j}\delta_{k,m} \right)_{l,m=1}^n,
		\end{equation}
		where $\delta_{j,k}$ is the Kronecker delta.

		\item We denote by $D(z_0, r)$ the open disc centered at $z_0$ with radius $r > 0$, i.e.,
		\begin{equation}\label{def:dz0r}
			D(z_0, r) := \{ z\in \mathbb{C} \mid |z-z_0|<r \},
		\end{equation}
		and by $\partial D(z_0, r)$ its boundary. The orientation of $\partial D(z_0, r)$ is taken in a clockwise manner.
        % \item A $2n \times 2n$ matrix $M$ can be viewed as a $2 \times 2$ block matrix, where the $(i,j)$-th block is denoted by $[M]_{ij}$ for $i,j = 1,2$.		
        \item As usual, the three Pauli matrices $\{\sigma_j\}_{j=1}^3$ are defined by
		\begin{equation}\label{def:Pauli}
			\sigma_1=\begin{pmatrix}
				0 & 1 \\
				1 & 0
			\end{pmatrix},
			\qquad
			\sigma_2=\begin{pmatrix}
				0 & -\ii \\
				\ii & 0
			\end{pmatrix},
			\qquad
			\sigma_3=
			\begin{pmatrix}
				1 & 0 \\
				0 & -1
			\end{pmatrix}.
		\end{equation}
        \item  For an arbitrary $n \times n$ matrix $M$, we denote the tensor product
        \begin{align}\label{def:tensor}
            M \otimes \sigma_3 := \begin{pmatrix}
                M & 0\\
                0 & -M
            \end{pmatrix}.
        \end{align}
	\end{itemize}
	
	%\section{Statement of results}

\section{The RH problem setup}\label{sec:3}
\subsection{An RH characterization of $\beta_1$}
Let $\beta_1(\vec{s})= \beta_1(\vec{s};C)$, $\vec{s}\in\mathbb{R}^n$, be the unique solution of the noncommutative PII equation \eqref{noncommutative equation} characterized by the asymptotics \eqref{+infty asympototics}. By \cite[Theorem 5.3]{MM12}, it follows that 
\begin{align}\label{def:beta1}
			\beta_1(\vec s) = -\ii \lim_{\lambda\to \infty} \lambda \left[\Xi\right]_{12}(\lambda;\vec{s}),
		\end{align}
where $\Xi$ is a $2n\times 2n$ matrix-valued function solving the following RH problem and $\left[\Xi\right]_{12}$ stands for its $(1,2)$-th block.
	\begin{rhp}\label{rhp:pole-free}
		\hfill
		\begin{itemize}
			\item[\rm (a)] $\Xi(\lambda):= \Xi(\lambda;\vec {s},C)$ is defined and analytic in $ \mathbb{C} \setminus (\gamma_+ \cup \gamma_-)$, 
			where the contours $\gamma_+$ and $\gamma_-$ are illustrated in Figure \ref{fig:pole-free}. 
            %$n$ real parameters introduced earlier, see Figure \ref{fig:pole-free} for an illustration of the contour $\gamma_+ \cup \gamma_-$ and the orientation.
		
			\item[\rm (b)] For  $\lambda \in\gamma_+ \cup \gamma_-$, we have
			\begin{align}
				\Xi_+(\lambda) = \Xi_-(\lambda) \begin{pmatrix}
					I_n & e^{\theta(\lambda)}Ce^{\theta(\lambda)}\, \chi_{\gamma_+} \\
					e^{-\theta(\lambda)}Ce^{-\theta(\lambda)}\,\chi_{\gamma_-} & I_n
				\end{pmatrix}, 
			\end{align}
			% \begin{align}
			% 	M(\lambda):= \begin{pmatrix}
			% 		I_n & e^{\theta(\lambda)}Ce^{\theta(\lambda)}\, \chi_{\gamma_+} \\
			% 		e^{-\theta(\lambda)}Ce^{-\theta(\lambda)}\,\chi_{\gamma_-} & I_n
			% 	\end{pmatrix},
			% \end{align}
			where $\chi_{\scalebox{0.6} X}$ denotes the indicator function of set $X$ and $\theta$ is defined in \eqref{def:theta}.  
   %          matrix $C$ is a constant $n\times n$ matrix with
			% \begin{align}
			% 	\theta(\lambda):&=\frac{\ii \lambda^3}{6} I_n+\ii \mathbf{s} \lambda, 
			% \end{align}
   %          where $\mathbf{s}$ is defined in \eqref{noncommutative equation}.
			\item[\rm (c)] As $\lambda \to \infty$ with $\lambda \in \mathbb{C} \setminus (\gamma_+ \cup \gamma_-)$, we have
			\begin{align}\label{asy:Xi}
				\Xi(\lambda) = I_{2n}+ \frac{\Xi_1}{\lambda}+\Boh(\lambda^{-2}),
			\end{align}
			where $\Xi_1$ is independent of $\lambda$.
		\end{itemize}
	\end{rhp}
	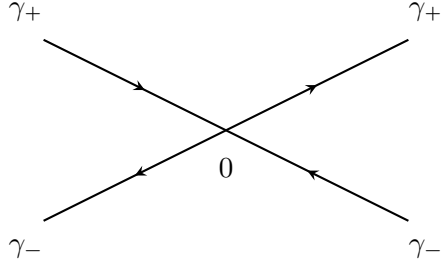
\begin{figure}[t]
				\begin{center}
					\begin{tikzpicture}[scale=1.2, >=stealth]
						
						% Draw the crossing lines with arrows in the middle
						\draw[thick, postaction={decorate, decoration={markings, mark=at position 0.5 with {\arrow{>}}}}] 
						(0,0) -- (2,1);
						\draw[thick, postaction={decorate, decoration={markings, mark=at position 0.5 with {\arrow{>}}}}] 
						(0,0) -- (-2,-1);
						\draw[thick, postaction={decorate, decoration={markings, mark=at position 0.5 with {\arrow{<}}}}] 
						(0,0) -- (-2,1);
						\draw[thick, postaction={decorate, decoration={markings, mark=at position 0.5 with {\arrow{<}}}}] 
						(0,0) -- (2,-1);
						
						% Add label for the origin
						\node[below] at (0, -0.2) {0};
						
						% Add labels \gamma_+ and \gamma_- symmetrically
						\node[above right=10pt] at (1.6, 0.8) {$\gamma_+$};
						\node[above left=10pt] at (-1.6, 0.8) {$\gamma_+$};
						\node[below right=10pt] at (1.6, -0.8) {$\gamma_-$};
						\node[below left=10pt] at (-1.6, -0.8) {$\gamma_-$};
						
					\end{tikzpicture}
					\caption{The jump contours $\gamma_+$ and $  \gamma_-$  in the RH problem for $\Xi$.}
					\label{fig:pole-free}
					
				\end{center}
			\end{figure}
            
	% In \cite[Theorem 5.3]{MM12}, Bertola and Cafasso established the connection between a family of specific solutions to the noncommutative PII equation and the aforementioned Riemann–Hilbert problem. 
 %    More preciseley, let $\Xi(\lambda)= \Xi(\lambda;\vec {s})$ be the solution of RH problem \ref{rhp:pole-free}, then we have 	\begin{align}\label{def:beta1}
	% 		\beta_1(\vec s) = -\ii \lim_{\lambda\to \infty} \lambda \left[\Xi\right]_{12}(\lambda;\vec{s}),
	% 	\end{align}
	% 	where $\beta_1(\vec s)$ is defined in \eqref{noncommutative equation}.

\subsection{Rescaling and decomposition of $\Xi$}
\label{sec:XitoPsi}
	We start with a rescaling of the $\mathrm{RH}$ problem for $\Xi$, which is defined by
	\begin{equation}\label{def:XitoPsi}
		\Psi(z)=\Xi(\sqrt{-S}z)e^{\what \theta(z)\otimes \sigma_3},
	\end{equation}
	where the tensor product $\otimes$ is defined in \eqref{def:tensor}, 
	\begin{align}\label{def:thetahat}
		\what \theta(z) :=\ii(-S)^{\frac{3}{2}} \left[\frac{1}{6}z^3 I_n-\frac{z}{S}\mathbf{s}\right]
	\end{align}
	with $\mathbf{s}$ defined in \eqref{noncommutative equation} and $\sigma_3$ is defined in \eqref{def:Pauli}.  In view of $\mathrm{RH}$ problem for $\Xi$, it is readily seen that $\Psi$ satisfies the following $\mathrm{RH}$ problem. 
	\begin{rhp}\label{rhp:Psi}
		\hfill
		\begin{itemize}
			\item[\rm (a)] $\Psi(z)$ is defined and analytic in $ \mathbb{C} \setminus (\gamma_+ \cup \gamma_-)$.
            %, see Figure \ref{fig:pole-free} for an illustration of the contour $\gamma_+ \cup \gamma_-$ and the orientation.
			\item[\rm (b)] For  $z \in\gamma_+ \cup \gamma_-$, we have
			\begin{align}\label{eq:Psijump}
				\Psi_+(z) = \Psi_-(z)\begin{pmatrix}
					I_n & C\, \chi_{\gamma_+} \\
					C\, \chi_{\gamma_-} & I_n
				\end{pmatrix}.
			\end{align}
			% with jump matrix
			% \begin{align}\label{def:hatM}
			% 	\what M:= \begin{pmatrix}
			% 		I_n & C\, \chi_{\gamma_+} \\
			% 		C\, \chi_{\gamma_-} & I_n
			% 	\end{pmatrix}.
			% \end{align}
			\item[\rm (c)] As $z \to \infty$ with $z \in \mathbb{C} \setminus (\gamma_+ \cup \gamma_-)$, we have
			\begin{align}\label{eq:asyPsi}
				\Psi(z) = \left(I_{2n}+ \frac{\Xi_1}{\sqrt{-S}z}+\Boh\left(z^{-2}\right)\right) e^{\what \theta(z)\otimes \sigma_3},
			\end{align}
			where $\Xi_1$ is defined in \eqref{asy:Xi} and $\what \theta(z)$ is defined in \eqref{def:thetahat}.
		\end{itemize}
	\end{rhp}

We next show that, under Assumption \ref{assump}, $\Psi$ can be written as a sum of two functions, and each function is related to a specific RH problem. This decomposition is crucial in our further analysis. To proceed, recall the canonical one-to-one correspondence between the permutation matrices of order $n$ and the symmetric group $S_n$, we denote by $\sigma\in S_n$ the permutation associated with the matrix $P$ in Assumption \ref{assump}, that is, 
        \begin{align}\label{def:sigma}
        P = \sum_{i=1}^{n} E_{i\sigma(i)},
        \end{align}
where $E_{ij}$ is defined in \eqref{def:Eij}. Note that $P^2=I_n$, it then follows $\sigma^2 = \mathrm{id}$, i.e., $\sigma$ is an involution. Since $C^2$ is a diagonal matrix, the condition $c_{kl} \neq 0$ implies $l = \sigma(k)$, thus, 
  \begin{align} \label{expansion}
      C= \sum_{i=1}^{n} c_{i\sigma(i)} E_{i\sigma(i)}. 
  \end{align}
Based on \eqref{expansion}, we distinguish between the cases $c_{kl}c_{lk}=1$ and $c_{kl}c_{lk}\neq 1$ by introducing the index sets:
\begin{align}\label{def:calIJ}
      \mathcal{I}:=\{i:c_{i\sigma(i)}c_{\sigma(i)i}=1\}, \qquad  \mathcal{J}=\{1,\ldots,n\} \setminus \mathcal{I}.
\end{align}
  
\begin{remark}
      When $\mathcal{I} = \emptyset$, this is equivalent to the condition $\det(I_n - C^2) \neq 0$; whereas when $\mathcal{J} = \emptyset$, it corresponds to $C^2 = I_n$.
\end{remark}

With the aid of the index sets $\I$ and $\J$ defined in \eqref{def:calIJ}, it is readily seen that $\Psi$ admits the following decomposition:
\begin{align}\label{IJ}
        \Psi (z) = \Psi_{\mathcal I}(z) \begin{pmatrix}
            \sum_{i \in \mathcal{I}} E_{ii} & 0\\
            0 & \sum_{i \in \mathcal{I}} E_{ii}
        \end{pmatrix} + \Psi_{\mathcal J}(z) \begin{pmatrix}
            \sum_{j \in \mathcal{J}} E_{jj} & 0\\
            0 & \sum_{j \in \mathcal{J}} E_{jj}
        \end{pmatrix},
    \end{align}
where
    \begin{align}
        \Psi_{\mathcal I}(z) &:= \Psi (z)\begin{pmatrix}
            \sum_{i \in \mathcal{I}} E_{ii} & 0\\
            0 & \sum_{i \in \mathcal{I}} E_{ii}
        \end{pmatrix} + \begin{pmatrix}
            \sum_{j \in \mathcal{J}} E_{jj} & 0\\
            0 & \sum_{j \in \mathcal{J}} E_{jj}
        \end{pmatrix},\label{def-I}\\
        \Psi_{\mathcal J}(z) &:=  \Psi (z)\begin{pmatrix}
            \sum_{j \in \mathcal{J}} E_{jj} & 0\\
            0 & \sum_{j \in \mathcal{J}} E_{jj}
              \end{pmatrix}+\begin{pmatrix}
            \sum_{i \in \mathcal{I}} E_{ii} & 0\\
            0 & \sum_{i \in \mathcal{I}} E_{ii}
        \end{pmatrix}.\label{def-J}
    \end{align}

A key observation now is the sets $\I$ and $\J$ are invariant under the involution  $\sigma$, that is, if $i \in \I$ or $j\in \J$, then $\sigma(i) \in \I$ or $\sigma(j) \in \J$, respectively. This invariance ensures the  commutation relations
  \begin{align}\label{IC}
      \sum_{i \in \mathcal{I}} E_{ii}C=C\sum_{i \in \mathcal{I}} E_{ii},
  \qquad
      \sum_{j \in \mathcal{J}} E_{jj}C=C\sum_{j \in \mathcal{J}} E_{jj},
  \end{align}
hold. 
This pivotal commutation property allows us to establish RH problems for $\Psi_{\mathcal I}$ and $\Psi_{\mathcal J}$, respectively.   
   %Accordingly, we decompose $\Psi(z)$ into the sum of two 

The function $\Psi_{\mathcal I}$ defined in \eqref{def-I} satisfies the following RH problem.
    \begin{rhp}\label{rhp:PsiI}
		\hfill
		\begin{itemize}
			\item[\rm (a)] $\Psi_{\mathcal I}(z)$ is defined and analytic in $ \mathbb{C} \setminus (\gamma_+ \cup \gamma_-)$.
            %, see Figure \ref{fig:pole-free} for an illustration of the contour $\gamma_+ \cup \gamma_-$ and the orientation.
			\item[\rm (b)] For  $z \in\gamma_+ \cup \gamma_-$, we have
			\begin{align}\label{jump:PsiI}
				\Psi_{\mathcal I, +}(z) = \Psi_{\mathcal I,-}(z) \begin{pmatrix}
					I_n & \sum_{i \in \mathcal{I}} E_{ii}C\, \chi_{\gamma_+} \\
					\sum_{i \in \mathcal{I}} E_{ii}C\, \chi_{\gamma_-} & I_n
				\end{pmatrix}.
			\end{align}
			\item[\rm (c)] As $z \to \infty$ with $z \in \mathbb{C} \setminus (\gamma_+ \cup \gamma_-)$, we have
			\begin{align}\label{asy:PsiI}
				\Psi_{\mathcal I}(z) = \left(I_{2n}+ \frac{\Psi_{\mathcal I,1}}{\sqrt{-S}z}+\Boh\left(z^{-2}\right)\right) e^{\sum_{i \in \mathcal{I}} E_{ii} \what \theta(z)\otimes \sigma_3},
			\end{align}
			where 
            \begin{align}\label{I1}
                \Psi_{\mathcal I,1} = \Xi_1 \begin{pmatrix}
            \sum_{i \in \mathcal{I}} E_{ii} & 0\\
            0 & \sum_{i \in \mathcal{I}} E_{ii}
        \end{pmatrix}
            \end{align}
            with $\Xi_1$ defined in \eqref{asy:Xi}.
		\end{itemize}
	\end{rhp}
    \begin{proof}
    It is easily seen that $\Psi_{\I}$ is analytic in $\mathbb{C}\setminus (\gamma_+\cup\gamma_-)$. To show the jump condition, we see from \eqref{eq:Psijump} and \eqref{def-I}   that for $z\in\gamma_+$,
         % Here, we show that $\Psi_{\I}$ defined in \eqref{def-I} indeed solves RH problem \ref{rhp:PsiI}. We start with the jump condition \eqref{jump:PsiI}. For $z \in \gamma_+$, by applying \eqref{IC} to the fourth equation, we obtain
        \begin{align}\label{com:PsiI}
            &\Psi_{\mathcal I,-}(z) \begin{pmatrix}
					I_n & \sum_{i \in \mathcal{I}} E_{ii}C \\
					0 & I_n
				\end{pmatrix}\nonumber\\
                &=\left[\Psi_- (z)\begin{pmatrix}
            \sum_{i \in \mathcal{I}} E_{ii} & 0\\
            0 & \sum_{i \in \mathcal{I}} E_{ii}
        \end{pmatrix} + \begin{pmatrix}
            \sum_{j \in \mathcal{J}} E_{jj} & 0\\
            0 & \sum_{j \in \mathcal{J}} E_{jj}
        \end{pmatrix}\right]\begin{pmatrix}
					I_n & \sum_{i \in \mathcal{I}} E_{ii}C \\
					0 & I_n
				\end{pmatrix}\nonumber\\
                &=\Psi_- (z)\begin{pmatrix}
            \sum_{i \in \mathcal{I}} E_{ii} & \sum_{i \in \mathcal{I}} E_{ii}C\\
            0 & \sum_{i \in \mathcal{I}} E_{ii}
        \end{pmatrix} + \begin{pmatrix}
            \sum_{j \in \mathcal{J}} E_{jj} & 0\\
            0 & \sum_{j \in \mathcal{J}} E_{jj}
        \end{pmatrix}\nonumber\\
        &=\Psi_+ (z)\begin{pmatrix}
					I_n & -C \\
					 0& I_n
				\end{pmatrix}\begin{pmatrix}
            \sum_{i \in \mathcal{I}} E_{ii} & \sum_{i \in \mathcal{I}} E_{ii}C\\
            0 & \sum_{i \in \mathcal{I}} E_{ii}
        \end{pmatrix} + \begin{pmatrix}
            \sum_{j \in \mathcal{J}} E_{jj} & 0\\
            0 & \sum_{j \in \mathcal{J}} E_{jj}
        \end{pmatrix}\nonumber\\
        &=\Psi_+ (z)\begin{pmatrix}
            \sum_{i \in \mathcal{I}} E_{ii} & \sum_{i \in \mathcal{I}} E_{ii}C-C\sum_{i \in \mathcal{I}} E_{ii}\\
            0 & \sum_{i \in \mathcal{I}} E_{ii}
        \end{pmatrix} + \begin{pmatrix}
            \sum_{j \in \mathcal{J}} E_{jj} & 0\\
            0 & \sum_{j \in \mathcal{J}} E_{jj}
        \end{pmatrix}\nonumber\\
        &=\Psi_+ (z)\begin{pmatrix}
            \sum_{i \in \mathcal{I}} E_{ii} & 0\\
            0 & \sum_{i \in \mathcal{I}} E_{ii}
        \end{pmatrix} + \begin{pmatrix}
            \sum_{j \in \mathcal{J}} E_{jj} & 0\\
            0 & \sum_{j \in \mathcal{J}} E_{jj}
        \end{pmatrix} =\Psi_{\mathcal I, +}(z),
        \end{align}
        where we have made use of the first equality of \eqref{IC} in the fourth equality. For $z \in \gamma_-$, the computation proceeds in a similar way and we omit the details here.
    
    The large-$z$ behavior of $\Psi$ in \eqref{asy:PsiI} follows from a direct calculation. More precisely, as $z \to \infty$ with $z \in \mathbb{C} \setminus (\gamma_+ \cup \gamma_-)$, we have from \eqref{eq:asyPsi}, \eqref{def-I} 
and the facts 
$$ \sum_{i \in \mathcal{I}} E_{ii}+\sum_{j \in \mathcal{J}} E_{jj}=I_n, 
\quad \left(\sum_{i \in \mathcal{I}} E_{ii}\right)\cdot\left(\sum_{j \in \mathcal{J}} E_{jj}\right)=0_{n\times n},$$
that
\begin{align}
        \Psi_{\I}(z)&= \Psi (z)\begin{pmatrix}
            \sum_{i \in \mathcal{I}} E_{ii} & 0\\
            0 & \sum_{i \in \mathcal{I}} E_{ii}
        \end{pmatrix} + \begin{pmatrix}
            \sum_{j \in \mathcal{J}} E_{jj} & 0\\
            0 & \sum_{j \in \mathcal{J}} E_{jj}
        \end{pmatrix} \nonumber\\
        &= \left(I_{2n}+ \frac{\Xi_1}{\sqrt{-S}z}+\Boh\left(z^{-2}\right)\right) e^{\what \theta(z)\otimes \sigma_3}\begin{pmatrix}
            \sum_{i \in \mathcal{I}} E_{ii} & 0\\
            0 & \sum_{i \in \mathcal{I}} E_{ii}
        \end{pmatrix} + \begin{pmatrix}
            \sum_{j \in \mathcal{J}} E_{jj} & 0\\
            0 & \sum_{j \in \mathcal{J}} E_{jj}
        \end{pmatrix} \nonumber\\
        &= \left(I_{2n}+ \frac{\Xi_1}{\sqrt{-S}z}+\Boh\left(z^{-2}\right)\right) 
        \begin{pmatrix}
            \sum_{i \in \mathcal{I}} E_{ii} & 0\\
            0 & \sum_{i \in \mathcal{I}} E_{ii}
        \end{pmatrix}
        e^{\what \theta(z)\otimes \sigma_3}\begin{pmatrix}
            \sum_{i \in \mathcal{I}} E_{ii} & 0\\
            0 & \sum_{i \in \mathcal{I}} E_{ii}
        \end{pmatrix} \nonumber\\
       &\quad + \begin{pmatrix}
            \sum_{j \in \mathcal{J}} E_{jj} & 0\\
            0 & \sum_{j \in \mathcal{J}} E_{jj}
        \end{pmatrix}e^{\sum_{i \in \mathcal{I}} E_{ii} \what \theta(z)\otimes \sigma_3} \nonumber\\
        &=\left(I_{2n}+ \frac{\Psi_{\mathcal I,1}}{\sqrt{-S}z}+\Boh\left(z^{-2}\right)\right) e^{\sum_{i \in \mathcal{I}} E_{ii} \what \theta(z)\otimes \sigma_3},
    \end{align}
    where $\Psi_{\mathcal I,1}$ is defined in \eqref{I1}.
    \end{proof}
           
Similarly, by using the second equality of \eqref{IC}, one has that $\Psi_{\mathcal J}$ in \eqref{def-J} satisfies the following RH problem.
    \begin{rhp}\label{rhp:PsiJ}
		\hfill
		\begin{itemize}
			\item[\rm (a)] $\Psi_{\mathcal J}(z)$ is defined and analytic for $z \in \mathbb{C} \setminus (\gamma_+ \cup \gamma_-)$. 
			\item[\rm (b)] For  $z \in\gamma_+ \cup \gamma_-$, we have
			\begin{align}
				\Psi_{\mathcal J, +}(z) = \Psi_{\mathcal J,-}(z) \begin{pmatrix}
					I_n & \sum_{j \in \mathcal{J}} E_{jj}C\, \chi_{\gamma_+} \\
					\sum_{j \in \mathcal{J}} E_{jj}C\, \chi_{\gamma_-} & I_n
				\end{pmatrix}.
			\end{align}
			\item[\rm (c)] As $z \to \infty$ with $z \in \mathbb{C} \setminus (\gamma_+ \cup \gamma_-)$, we have
			\begin{align}
				\Psi_{\mathcal J}(z) = \left(I_{2n}+ \frac{\Psi_{\mathcal J,1}}{\sqrt{-S}z}+\Boh\left(z^{-2}\right)\right) e^{\sum_{j \in \mathcal{J}} E_{jj} \what \theta(z)\otimes \sigma_3},
			\end{align}
			where 
            \begin{align}\label{J1}
                \Psi_{\mathcal J,1} = \Xi_1 \begin{pmatrix}
            \sum_{j \in \mathcal{J}} E_{jj} & 0\\
            0 & \sum_{j \in \mathcal{J}} E_{jj}
        \end{pmatrix}
            \end{align}
            with $\Xi_1$ defined in \eqref{asy:Xi}.
		\end{itemize}
	\end{rhp}

The decomposition \eqref{IJ} facilitates a separate treatment of the cases $c_{kl}c_{lk}=1$ and $c_{kl}c_{lk}\neq 1$, which will be presented in the next two sections.  

%as presented in the following two sections, thereby simplifying the whole analysis.
    
\section{Asymptotic analysis of the RH problem for $\Psi_\mathcal{I}$} \label{sec:I}
In this section, we analyze RH problem \ref{rhp:PsiI} for $\Psi_\mathcal{I}$ as $S \to -\infty$.
\subsection{First transformation: $\Psi_\mathcal{I} \to T$}\label{sec:PsiI-T}
In this transformation we normalize RH problem \ref{rhp:PsiI} for $\Psi_{\I}$ at infinity. For this purpose, we introduce $g$-functions
	\begin{align}\label{def:gi}
		g_i(z):=\frac{1}{6}\left(z^2-4r_i\right)^{\frac{3}{2}}, \qquad z \in \mathbb{C} \backslash[-2\sqrt{r_i},2\sqrt{r_i}],\quad i \in \mathcal{I},
	\end{align}
where
	\begin{align}\label{def:ri}
		r_i:=\frac{s_{i}+s_{\sigma(i)}}{2S},\quad i \in \mathcal{I},
	\end{align}
and $\sigma$ is defined in \eqref{def:sigma}. Recall that $\delta_i = s_i-S=\epsilon_i S$, $-1<\epsilon_i<1$, one has $r_i = r_{\sigma(i)} = \frac{\epsilon_i+\epsilon_{\sigma(i)}}{2}+1>0$.
%with $\epsilon_i$ defined in Theorem \ref{-infty asympototics results}.
As $z\to \infty$, it is readily seen that
	\begin{align}\label{asy:g_i}
		g_i(z)=\frac{1}{6}(z^2-4r_i)^{\frac{3}{2}}=\frac{1}{6}z^3-r_i z+\Boh(z^{-1}).
	\end{align}

It is possible that $r_i$ may coincide with each other.	For clarity, we select distinct elements of $2\sqrt{r_i}$, arrange them in ascending order, and denote them as	
	\begin{align}
		\widetilde{r}_1< \widetilde{r}_2<\cdots< \widetilde{r}_{m_1}, \quad 1\le m_1 \le \lvert \mathcal{I} \rvert,
	\end{align}
where $\lvert \mathcal{I} \rvert$ denotes the cardinality of $\mathcal{I}$. For later use, we also set 
\begin{align}\label{def:J_k}
		J_k:=\{i: 2\sqrt{r_i}=\widetilde r_k\}, \,  1\leq k \leq m_1,  \quad J_0 := \{1,\ldots,n\} \setminus \bigcup_{k=1}^{m_1} J_k= \mathcal{J}, 
\end{align}
and define
\begin{align}\label{def:C_k}
		C_{k} := \sum_{i\in \bigcup_{l=k+1}^{m_1}J_l}E_{ii} C, \quad 0\leq k \leq m_1-1, \qquad C_{m_1} := 0_{n\times n}. 
	\end{align}
\begin{remark}
	For any $i\in J_k$, the involution  $\sigma$ defined in \eqref{def:sigma} preserves the subset invariance, i.e., $\sigma(i) \in J_k$. This invariance property implies the commutation relations
\begin{align}
	\sum_{i \in J_k} E_{ii}C=\sum_{i \in J_k}c_{i\sigma(i)}E_{i\sigma(i)} = \sum_{i \in J_k}c_{\sigma(i)i} E_{\sigma(i)i}= C\sum_{i \in J_k} E_{ii}.
\end{align}
Moreover, we have 
\begin{align}\label{commutation}
    e^{\sum_{i\in \mathcal{I}} g_i(z) E_{ii}} C = C e^{\sum_{i\in \mathcal{I}} g_i(z) E_{ii}}.
\end{align}
	\end{remark}

    % \begin{remark}
    %     It is worth noting that, apart from the requirement $1\le m_1 \le \lvert \mathcal{I} \rvert$, where $\lvert \mathcal{I} \rvert$ denotes the cardinality of $\mathcal{I}$, there are no additional constraints on $m_1$. Moreover, the specific value of $m_1$ has no impact on the final asymptotic result.
    % \end{remark}
The first transformation is then defined by
	\begin{align}\label{def:PsitoT}
		T(z)=\Psi_{\I}(z)Q(z) e^{-G(z)\otimes \sigma_3},
	\end{align}
where
    \begin{align}\label{def:G}
		G(z):=\ii(-S)^{\frac{3}{2}}\sum_{i \in \mathcal{I}} \left(g_i(z)-t_i z\right)E_{ii},
	\end{align}
    with
	\begin{align}\label{def:ti}
		 t_i:=\frac{s_{i}-s_{\sigma(i)}}{2S},\quad i\in \mathcal{I},
	\end{align}
	and 
	\begin{align}\label{def:Q}
		Q(z):= \begin{cases}
			\begin{pmatrix}
				I_n & C_{k-1} \\
				0 & I_n
			\end{pmatrix}, & \quad z \in \Omega_{2k-1} \cup \Omega_{2m_1+2k-1},\, 1\leq k \leq m_1, \\
			\begin{pmatrix}
				I_n & 0 \\
				C_{k-1} & I_n
			\end{pmatrix}, & \quad z \in \Omega_{2k} \cup \Omega_{2m_1+2k}, \, 1\leq k \leq m_1,\\
			I_{2n}, & \quad \textrm{elsewhere},
		\end{cases} 
	\end{align}
 where the regions $\Omega_j$, $j=1, \ldots, 4m_1$, are shown in Figure \ref{fig:T} and the matrix $C_k$ is defined in \eqref{def:C_k}. We have that $T$ satisfies the following RH problem.
	% with
	% \begin{align}\label{def:J_k}
	% 	J_k:=\{i: 2\sqrt{r_i}=\widetilde r_k\}, \,  1\leq k \leq m_1,  \quad J_0 := \{1,\cdots,n\} \setminus \bigcup_{k=1}^{m_1} J_k= \mathcal{J}. 
	% \end{align}
	\begin{figure}[t]
\begin{center}
\tikzset{every picture/.style={line width=0.75pt}} 
        \begin{tikzpicture}[x=0.75pt,y=0.75pt,yscale=-1,xscale=0.9]
        \tikzstyle{every node}=[scale=0.9]

%Straight Lines [id:da8489016544296536] 
\draw    (229.17,163.95) -- (287.84,163.95) ;
\draw [shift={(261.1,163.95)}, rotate = 180] [fill={rgb, 255:red, 0; green, 0; blue, 0 }  ][line width=0.08]  [draw opacity=0] (7.14,-3.43) -- (0,0) -- (7.14,3.43) -- (4.74,0) -- cycle    ;
%Straight Lines [id:da6796983681136035] 
\draw    (275.84,163.95) -- (336.24,163.95) ;
%Straight Lines [id:da613185432630844] 
\draw    (336.24,163.95) -- (384.65,163.95) ;
%Straight Lines [id:da7504332571185158] 
\draw    (381.17,163.95) -- (439.84,163.95) ;
\draw [shift={(413.11,163.95)}, rotate = 180] [fill={rgb, 255:red, 0; green, 0; blue, 0 }  ][line width=0.08]  [draw opacity=0] (7.14,-3.43) -- (0,0) -- (7.14,3.43) -- (4.74,0) -- cycle    ;
%Straight Lines [id:da4143333786505953] 
\draw    (82.19,113.17) -- (135.17,163.95) ;
\draw [shift={(110.56,140.36)}, rotate = 223.79] [fill={rgb, 255:red, 0; green, 0; blue, 0 }  ][line width=0.08]  [draw opacity=0] (7.14,-3.43) -- (0,0) -- (7.14,3.43) -- (4.74,0) -- cycle    ;
%Straight Lines [id:da05382685170363366] 
\draw    (234.86,113.17) -- (287.84,163.95) ;
\draw [shift={(263.23,140.36)}, rotate = 223.79] [fill={rgb, 255:red, 0; green, 0; blue, 0 }  ][line width=0.08]  [draw opacity=0] (7.14,-3.43) -- (0,0) -- (7.14,3.43) -- (4.74,0) -- cycle    ;
%Straight Lines [id:da8321812744651109] 
\draw    (381.17,163.95) -- (434.15,214.74) ;
\draw [shift={(409.54,191.15)}, rotate = 223.79] [fill={rgb, 255:red, 0; green, 0; blue, 0 }  ][line width=0.08]  [draw opacity=0] (7.14,-3.43) -- (0,0) -- (7.14,3.43) -- (4.74,0) -- cycle    ;
%Straight Lines [id:da6699101040648437] 
\draw    (439.84,163.95) -- (492.81,214.74) ;
\draw [shift={(468.2,191.15)}, rotate = 223.79] [fill={rgb, 255:red, 0; green, 0; blue, 0 }  ][line width=0.08]  [draw opacity=0] (7.14,-3.43) -- (0,0) -- (7.14,3.43) -- (4.74,0) -- cycle    ;
%Straight Lines [id:da4246408662572354] 
\draw    (434.67,113.17) -- (381.17,163.95) ;
\draw [shift={(410.89,135.74)}, rotate = 136.49] [fill={rgb, 255:red, 0; green, 0; blue, 0 }  ][line width=0.08]  [draw opacity=0] (7.14,-3.43) -- (0,0) -- (7.14,3.43) -- (4.74,0) -- cycle    ;
%Straight Lines [id:da45611561092034236] 
\draw    (493.33,113.17) -- (439.84,163.95) ;
\draw [shift={(469.56,135.74)}, rotate = 136.49] [fill={rgb, 255:red, 0; green, 0; blue, 0 }  ][line width=0.08]  [draw opacity=0] (7.14,-3.43) -- (0,0) -- (7.14,3.43) -- (4.74,0) -- cycle    ;
%Straight Lines [id:da788272488071564] 
\draw    (135.17,163.95) -- (81.68,214.74) ;
\draw [shift={(111.4,186.52)}, rotate = 136.49] [fill={rgb, 255:red, 0; green, 0; blue, 0 }  ][line width=0.08]  [draw opacity=0] (7.14,-3.43) -- (0,0) -- (7.14,3.43) -- (4.74,0) -- cycle    ;
%Straight Lines [id:da23407888485874495] 
\draw    (287.84,163.95) -- (234.34,214.74) ;
\draw [shift={(264.06,186.52)}, rotate = 136.49] [fill={rgb, 255:red, 0; green, 0; blue, 0 }  ][line width=0.08]  [draw opacity=0] (7.14,-3.43) -- (0,0) -- (7.14,3.43) -- (4.74,0) -- cycle    ;
%Straight Lines [id:da3316896512135067] 
\draw    (535.84,163.95) -- (594.51,163.95) ;
\draw [shift={(567.77,163.95)}, rotate = 180] [fill={rgb, 255:red, 0; green, 0; blue, 0 }  ][line width=0.08]  [draw opacity=0] (7.14,-3.43) -- (0,0) -- (7.14,3.43) -- (4.74,0) -- cycle    ;
%Straight Lines [id:da8239056950451331] 
\draw    (594.51,163.95) -- (647.48,214.74) ;
\draw [shift={(622.87,191.15)}, rotate = 223.79] [fill={rgb, 255:red, 0; green, 0; blue, 0 }  ][line width=0.08]  [draw opacity=0] (7.14,-3.43) -- (0,0) -- (7.14,3.43) -- (4.74,0) -- cycle    ;
%Straight Lines [id:da06292186634830166] 
\draw    (648,113.17) -- (594.51,163.95) ;
\draw [shift={(624.23,135.74)}, rotate = 136.49] [fill={rgb, 255:red, 0; green, 0; blue, 0 }  ][line width=0.08]  [draw opacity=0] (7.14,-3.43) -- (0,0) -- (7.14,3.43) -- (4.74,0) -- cycle    ;
%Straight Lines [id:da7100279824601728] 
\draw    (75.78,163.95) -- (134.45,163.95) ;
\draw [shift={(107.71,163.95)}, rotate = 180] [fill={rgb, 255:red, 0; green, 0; blue, 0 }  ][line width=0.08]  [draw opacity=0] (7.14,-3.43) -- (0,0) -- (7.14,3.43) -- (4.74,0) -- cycle    ;
%Straight Lines [id:da02099420361186699] 
\draw    (22.8,113.17) -- (75.78,163.95) ;
\draw [shift={(51.17,140.36)}, rotate = 223.79] [fill={rgb, 255:red, 0; green, 0; blue, 0 }  ][line width=0.08]  [draw opacity=0] (7.14,-3.43) -- (0,0) -- (7.14,3.43) -- (4.74,0) -- cycle    ;
%Straight Lines [id:da0006363063981379424] 
\draw    (75.78,163.95) -- (22.28,214.74) ;
\draw [shift={(52,186.52)}, rotate = 136.49] [fill={rgb, 255:red, 0; green, 0; blue, 0 }  ][line width=0.08]  [draw opacity=0] (7.14,-3.43) -- (0,0) -- (7.14,3.43) -- (4.74,0) -- cycle    ;
%Straight Lines [id:da21306869660108763] 
\draw  [dash pattern={on 0.84pt off 2.51pt}]  (153.17,163.95) -- (211.84,163.95) ;
%Straight Lines [id:da7270609534987692] 
\draw    (175.86,113.17) -- (228.84,163.95) ;
\draw [shift={(204.23,140.36)}, rotate = 223.79] [fill={rgb, 255:red, 0; green, 0; blue, 0 }  ][line width=0.08]  [draw opacity=0] (7.14,-3.43) -- (0,0) -- (7.14,3.43) -- (4.74,0) -- cycle    ;
%Straight Lines [id:da10193383976883319] 
\draw    (228.84,163.95) -- (175.34,214.74) ;
\draw [shift={(205.06,186.52)}, rotate = 136.49] [fill={rgb, 255:red, 0; green, 0; blue, 0 }  ][line width=0.08]  [draw opacity=0] (7.14,-3.43) -- (0,0) -- (7.14,3.43) -- (4.74,0) -- cycle    ;
%Straight Lines [id:da5569065677997533] 
\draw  [dash pattern={on 0.84pt off 2.51pt}]  (456.17,163.95) -- (514.84,163.95) ;
%Straight Lines [id:da16872551302732608] 
\draw    (534.84,163.95) -- (587.81,214.74) ;
\draw [shift={(563.2,191.15)}, rotate = 223.79] [fill={rgb, 255:red, 0; green, 0; blue, 0 }  ][line width=0.08]  [draw opacity=0] (7.14,-3.43) -- (0,0) -- (7.14,3.43) -- (4.74,0) -- cycle    ;
%Straight Lines [id:da44917248327018044] 
\draw    (588.33,113.17) -- (534.84,163.95) ;
\draw [shift={(564.56,135.74)}, rotate = 136.49] [fill={rgb, 255:red, 0; green, 0; blue, 0 }  ][line width=0.08]  [draw opacity=0] (7.14,-3.43) -- (0,0) -- (7.14,3.43) -- (4.74,0) -- cycle    ;
%Straight Lines [id:da8184641288475145] 
\draw    (135.17,163.95) -- (151.84,163.95) ;
%Straight Lines [id:da969878565615637] 
\draw    (212.17,163.95) -- (228.84,163.95) ;
%Straight Lines [id:da8915694131977353] 
\draw    (439.84,163.95) -- (456.51,163.95) ;
%Straight Lines [id:da1971574127165706] 
\draw    (518.17,163.95) -- (534.84,163.95) ;
%Straight Lines [id:da5240366992467177] 
\draw  [dash pattern={on 4.5pt off 4.5pt}]  (389.74,113.17) -- (336.24,163.95) ;
%Straight Lines [id:da529353822773218] 
\draw  [dash pattern={on 4.5pt off 4.5pt}]  (336.24,163.95) -- (282.75,214.74) ;
%Straight Lines [id:da14886372622578725] 
\draw  [dash pattern={on 4.5pt off 4.5pt}]  (283.27,113.17) -- (336.24,163.95) ;
%Straight Lines [id:da8128449260499224] 
\draw  [dash pattern={on 4.5pt off 4.5pt}]  (336.24,163.95) -- (389.22,214.74) ;

\draw  [dash pattern={on 1.2pt off 3pt}]  (140,135) -- (166,135) ;
\draw  [dash pattern={on 1.2pt off 3pt}]  (140,200) -- (166,200) ;

\draw  [dash pattern={on 1.2pt off 3pt}]  (500,135) -- (526,135) ;
\draw  [dash pattern={on 1.2pt off 3pt}]  (500,200) -- (526,200) ;
% Text Node
\draw (219.96,94.95) node [anchor=north west][inner sep=0.75pt]  [font=\scriptsize] [align=left] {$\displaystyle {\Gamma }_{2m_1+1}$};
% Text Node
\draw (64.33,167.65) node [anchor=north west][inner sep=0.75pt]  [font=\scriptsize] [align=left] {$\displaystyle -\widetilde{r}_{m_1}$};
% Text Node
\draw (156.96,93.95) node [anchor=north west][inner sep=0.75pt]  [font=\scriptsize] [align=left] {$\displaystyle {\Gamma }_{2m_1+3}$};
% Text Node
\draw (64.96,94.95) node [anchor=north west][inner sep=0.75pt]  [font=\scriptsize] [align=left] {$\displaystyle {\Gamma }_{4m_1-3}$};
% Text Node
\draw (7.96,93.95) node [anchor=north west][inner sep=0.75pt]  [font=\scriptsize] [align=left] {$\displaystyle {\Gamma }_{4m_1-1}$};
% Text Node
\draw (425.96,96.95) node [anchor=north west][inner sep=0.75pt]  [font=\scriptsize] [align=left] {$\displaystyle {\Gamma }_{1}$};
% Text Node
\draw (487.96,95.95) node [anchor=north west][inner sep=0.75pt]  [font=\scriptsize] [align=left] {$\displaystyle {\Gamma }_{3}$};
% Text Node
\draw (571.96,93.95) node [anchor=north west][inner sep=0.75pt]  [font=\scriptsize] [align=left] {$\displaystyle {\Gamma }_{2m_1-3}$};
% Text Node
\draw (634.96,93.95) node [anchor=north west][inner sep=0.75pt]  [font=\scriptsize] [align=left] {$\displaystyle {\Gamma }_{2m_1-1}$};
% Text Node
\draw (636.96,211.95) node [anchor=north west][inner sep=0.75pt]  [font=\scriptsize] [align=left] {$\displaystyle {\Gamma }_{2m_1}$};
% Text Node
\draw (572.96,210.95) node [anchor=north west][inner sep=0.75pt]  [font=\scriptsize] [align=left] {$\displaystyle {\Gamma }_{2m_1-2}$};
% Text Node
\draw (483.96,212.95) node [anchor=north west][inner sep=0.75pt]  [font=\scriptsize] [align=left] {$\displaystyle {\Gamma }_{4}$};
% Text Node
\draw (423.96,212.95) node [anchor=north west][inner sep=0.75pt]  [font=\scriptsize] [align=left] {$\displaystyle {\Gamma }_{2}$};
% Text Node
\draw (220.96,210.95) node [anchor=north west][inner sep=0.75pt]  [font=\scriptsize] [align=left] {$\displaystyle {\Gamma }_{2m_1+2}$};
% Text Node
\draw (156.96,211.95) node [anchor=north west][inner sep=0.75pt]  [font=\scriptsize] [align=left] {$\displaystyle {\Gamma }_{2m_1+4}$};
% Text Node
\draw (64.96,211.95) node [anchor=north west][inner sep=0.75pt]  [font=\scriptsize] [align=left] {$\displaystyle {\Gamma }_{4m_1-2}$};
% Text Node
\draw (7.96,213.95) node [anchor=north west][inner sep=0.75pt]  [font=\scriptsize] [align=left] {$\displaystyle {\Gamma }_{4m_1}$};
% Text Node
\draw (125.78,166.95) node [anchor=north west][inner sep=0.75pt]  [font=\scriptsize] [align=left] {$\displaystyle -\widetilde{r}_{m_1-1}$};
% Text Node
\draw (216.84,166.95) node [anchor=north west][inner sep=0.75pt]  [font=\scriptsize] [align=left] {$\displaystyle -\widetilde{r}_{2}$};
% Text Node
\draw (280.17,167.95) node [anchor=north west][inner sep=0.75pt]  [font=\scriptsize] [align=left] {$\displaystyle -\widetilde{r}_{1}$};
% Text Node
\draw (378.17,167.95) node [anchor=north west][inner sep=0.75pt]  [font=\scriptsize] [align=left] {$\displaystyle \widetilde{r}_{1}$};
% Text Node
\draw (436.17,167.95) node [anchor=north west][inner sep=0.75pt]  [font=\scriptsize] [align=left] {$\displaystyle \widetilde{r}_{2}$};
% Text Node
\draw (507,165.95) node [anchor=north west][inner sep=0.75pt]  [font=\scriptsize] [align=left] {$\displaystyle \widetilde{r}_{m_1-1}$};
% Text Node
\draw (582,167.65) node [anchor=north west][inner sep=0.75pt]  [font=\scriptsize] [align=left] {$\displaystyle \widetilde{r}_{m_1}$};
% Text Node
\draw (333,168.07) node [anchor=north west][inner sep=0.75pt]  [font=\scriptsize]  {$0$};
% Text Node
\draw (373,132.73) node [anchor=north west][inner sep=0.75pt]  [font=\scriptsize]  {${\Omega }_{1}$};
% Text Node
\draw (430,131.4) node [anchor=north west][inner sep=0.75pt]  [font=\scriptsize]  {${\Omega }_{3}$};
% Text Node
\draw (375.33,183.4) node [anchor=north west][inner sep=0.75pt]  [font=\scriptsize]  {${\Omega }_{2}$};
% Text Node
\draw (428,184.07) node [anchor=north west][inner sep=0.75pt]  [font=\scriptsize]  {${\Omega }_{4}$};
% Text Node
\draw (581.33,182.73) node [anchor=north west][inner sep=0.75pt]  [font=\scriptsize]  {${\Omega }_{2m_1}$};
% Text Node
\draw (575.33,134.07) node [anchor=north west][inner sep=0.75pt]  [font=\scriptsize]  {${\Omega }_{2m_1-1}$};
% Text Node
\draw (269.33,128.73) node [anchor=north west][inner sep=0.75pt]  [font=\scriptsize]  {${\Omega }_{2m_1+1}$};
% Text Node
\draw (264.67,184.73) node [anchor=north west][inner sep=0.75pt]  [font=\scriptsize]  {${\Omega }_{2m_1+2}$};
% Text Node
\draw (214.33,128.07) node [anchor=north west][inner sep=0.75pt]  [font=\scriptsize]  {${\Omega }_{2m_1+3}$};
% Text Node
\draw (212,184.73) node [anchor=north west][inner sep=0.75pt]  [font=\scriptsize]  {${\Omega }_{2m_1+4}$};
% Text Node
\draw (64,131.4) node [anchor=north west][inner sep=0.75pt]  [font=\scriptsize]  {${\Omega }_{4m_1-1}$};
% Text Node
\draw (69.33,185.4) node [anchor=north west][inner sep=0.75pt]  [font=\scriptsize]  {${\Omega }_{4m_1}$};

\end{tikzpicture}
			\caption{Regions $\Omega_i$ and the jump contours $\Gamma_i$, $i=1,\ldots,4m_1$, in the RH problem for $T$.
		The dashed lines represent the contour $\gamma_+ \cup \gamma_-$.}
			\label{fig:T}
			
		\end{center}
	\end{figure}
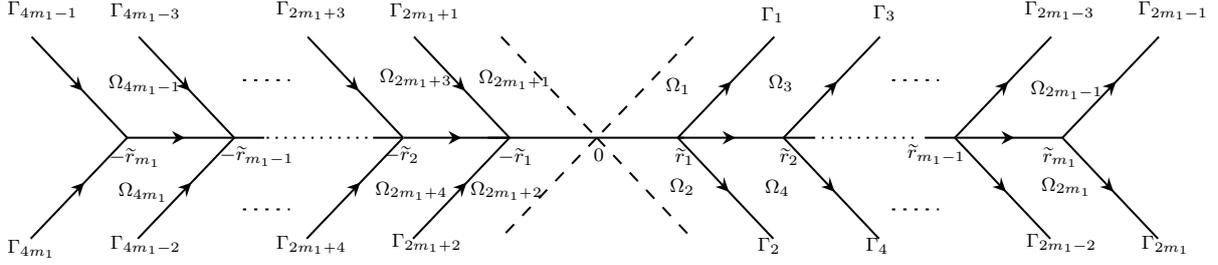

	\begin{rhp}\label{rhp:T}
		\hfill
		\begin{enumerate}
			\item[\rm (a)] $T(z)$ is defined and analytic in $\mathbb{C} \setminus \Gamma_T$, where 
			\begin{equation}\label{def:gammaT}
				\Gamma_T:=\cup_{i=1}^{4m_1} \Gamma_i \cup[-\widetilde r_{m_1}, \widetilde r_{m_1}];
			\end{equation}
see Figure \ref{fig:T} for the orientations of $\Gamma_T$.
			\item[\rm (b)] For $z \in \Gamma_T$, we have
            \begin{align}
                T_{+}(z)=T_{-}(z) J_T(z) 
            \end{align} 
            where
			\begin{equation}\label{def:JT}
				J_T(z):= \begin{cases}
					\begin{pmatrix}
						I_n & \sum_{i\in J_{k}} e^{2\ii (-S)^\frac{3}{2}g_i(z)}E_{ii}C \\
						0 & I_n
					\end{pmatrix}, & \quad z \in \Gamma_{2k-1} \cup \Gamma_{2m_1+2k-1},  \\
					\begin{pmatrix}
						I_n & 0 \\
						-\sum_{i\in J_{k}} e^{-2\ii (-S)^\frac{3}{2}g_i(z)}E_{ii}C & I_n
					\end{pmatrix}, & \quad z \in \Gamma_{2k} \cup \Gamma_{2m_1+2k},  \\
					\begin{pmatrix}
				e^{G_-(z) - G_+(z)} & \sum_{i\in \bigcup_{l=k}^{m_1}J_l}E_{ii} C\\
				  -\sum_{i\in \bigcup_{l=k}^{m_1}J_l}E_{ii} C  & \sum_{i\in \bigcup_{l=0}^{k-1}J_l}E_{ii}
			\end{pmatrix}
            , & \quad z \in (-\widetilde r_k, -\widetilde r_{k-1}) \cup  (\widetilde r_{k-1}, \widetilde r_k),
				\end{cases} 
			\end{equation}
			for $1 \leq k \leq m_1$ with $\widetilde r_0:=0$.
			\item[\rm (c)]As $z \to \infty$ with $z \in \mathbb{C} \setminus  \Gamma_T$, we have
			\begin{align}\label{eq:asyT}
				T(z) = I_{2n}+ \frac{\Psi_{\I,1}}{\sqrt{-S}z}+\Boh\left(z^{-2}\right) ,
			\end{align}
			where $\Psi_{\I,1}$ is defined in \eqref{I1}. 		
		\end{enumerate}
	\end{rhp}
	\begin{proof}
All the items follow directly from \eqref{def:PsitoT}--\eqref{def:Q} and RH problem \ref{rhp:PsiI} for $\Psi_{\I}$.  For the convenience of the reader, we check the jump condition of $T$ on $\Gamma_{2k-1}$ and $(\widetilde r_{k-1}, \widetilde r_{k})$, and the verification of the jump conditions on other contours is similar. For $z \in \Gamma_{2k-1}$, we see from \eqref{def:PsitoT} that 
		\begin{align}\label{check-JT}
			& T_{-}(z)^{-1}T_+(z)=e^{G(z)\otimes \sigma_3}Q_-(z)^{-1}Q_+(z)e^{-G(z)\otimes \sigma_3}
			= e^{G(z)\otimes \sigma_3} 	\begin{pmatrix}
				I_n & C_{k-1}-C_{k} \\
				0 & I_n
			\end{pmatrix}
			e^{-G(z)\otimes \sigma_3} \nonumber\\
		&=	\begin{pmatrix}
				I_n & e^{G(z)}\sum_{i\in J_{k}}E_{ii}C e^{G(z)} \\
				0 & I_n
			\end{pmatrix} 
            = \begin{pmatrix}
				I_n &\sum_{i\in J_{k}} e^{2\ii (-S)^\frac{3}{2}g_i(z)}E_{ii}C \\
				0 & I_n
			\end{pmatrix},
		\end{align}
where we have made use of \eqref{commutation} and the fact $t_i(z)=-t_{\sigma(i)}(z)$ in the last equality. For $z \in (\widetilde r_{k-1}, \widetilde r_{k})$, from \eqref{commutation} and the relation $g_{i,+}(z) = \begin{cases}
		    g_{i,-}(z), &\quad i \in \bigcup_{l=1}^{k-1}J_l \\
                -g_{i,-}(z), &\quad i \in \bigcup_{l=k}^{m_1}J_l 
		\end{cases}$, it follows that
		\begin{align}
			&T_{-}(z)^{-1}T_+(z)=e^{G_{-}(z)\otimes \sigma_3}Q_-(z)^{-1}Q_+(z)e^{-G_{+}(z)\otimes \sigma_3}\nonumber\\
			&= e^{G_{-}(z)\otimes \sigma_3} 	
			\begin{pmatrix}
				I_n & C_{k-1} \\
				-C_{k-1} & I_n-C_{k-1}^2
			\end{pmatrix}
			e^{-G_{+}(z)\otimes \sigma_3} \nonumber\\
			 &=\begin{pmatrix}
				e^{G_-(z) - G_+(z)} & e^{G_-(z)}C_{k-1}e^{G_+(z)}\\
				-e^{-G_-(z)}C_{k-1}e^{-G_+(z)} & e^{-G_-(z)}(I_n-C_{k-1}^2)e^{G_+(z)}
			\end{pmatrix}\nonumber\\
            &= \begin{pmatrix}
				e^{G_-(z) - G_+(z)} & e^{G_-(z)} \sum_{i\in \bigcup_{l=k}^{m_1}J_l}E_{ii} Ce^{G_+(z)}\\
				-e^{-G_-(z)} \sum_{i\in \bigcup_{l=k}^{m_1}J_l}E_{ii} C e^{-G_+(z)} & e^{-G_-(z)}\sum_{i\in \bigcup_{l=0}^{k-1}J_l}E_{ii}e^{G_+(z)}
			\end{pmatrix}\nonumber\\
            &=\begin{pmatrix}
				e^{G_-(z) - G_+(z)} & \sum_{i\in \bigcup_{l=k}^{m_1}J_l}E_{ii}e^{2\ii (-S)^\frac{3}{2}\left(g_{i,-(z)}+g_{i,+}(z)\right)} C\\
				  -\sum_{i\in \bigcup_{l=k}^{m_1}J_l}E_{ii} e^{-2\ii (-S)^\frac{3}{2}\left(g_{i,-(z)}+g_{i,+}(z)\right)}C  & \sum_{i\in \bigcup_{l=0}^{k-1}J_l}E_{ii}e^{2\ii (-S)^\frac{3}{2}\left(g_{i,+(z)}-g_{i,-}(z)\right)}
			\end{pmatrix}\nonumber\\
            &=\begin{pmatrix}
				e^{G_-(z) - G_+(z)} & \sum_{i\in \bigcup_{l=k}^{m_1}J_l}E_{ii} C\\
				  -\sum_{i\in \bigcup_{l=k}^{m_1}J_l}E_{ii} C  & \sum_{i\in \bigcup_{l=0}^{k-1}J_l}E_{ii}
			\end{pmatrix}.
		\end{align}
Finally, the large-$z$ behavior of $T$ is easy to check by combining \eqref{asy:PsiI}, \eqref{asy:g_i} and \eqref{def:PsitoT}, and we omit the details here.
	\end{proof}

	\subsection{Global parametrix}
In view of the signature of $\Re \ii g_i(z)$, $i \in \mathcal{I}$, illustrated in Figure \ref{real-g}, one has, as $S \to -\infty$, the jump matrix $J_T(z)$ of $T$ in \eqref{def:JT} tends to the identity matrix exponentially fast for $z$ bounded away from the interval $[-\widetilde r_{m_1}, \widetilde r_{m_1}]$.

\begin{figure}[t]
		\begin{center}
			\includegraphics[scale=0.6]{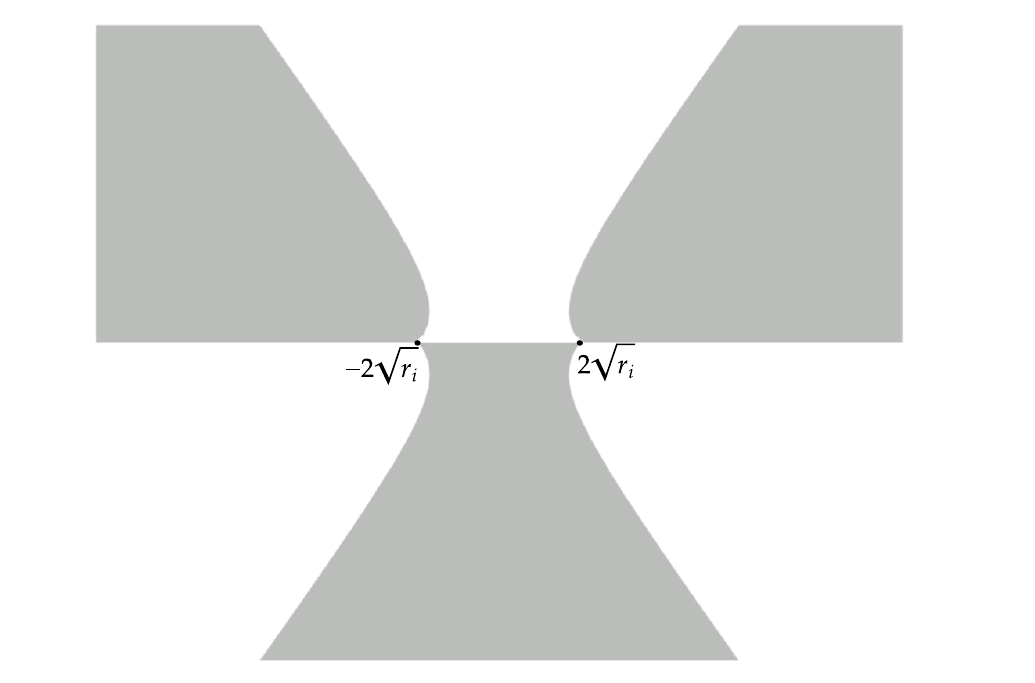}
			\caption{The signature of $\Re \ii g_i(z)$, $i \in \mathcal{I}$. The shaded areas indicate the regions where $\Re \ii g_i(z)< 0$, while the white areas indicate the regions where $\Re \ii g_i(z)> 0$.}
			\label{real-g}
		\end{center}
	\end{figure}
    
For $z \in (-\widetilde r_{k}, -\widetilde r_{k-1}) \cup  (\widetilde r_{k-1}, \widetilde r_{k})$, $1\le k \le m_1$, we note that
    \begin{align}
        J_T(z)& = \begin{pmatrix}
				e^{G_-(z) - G_+(z)} & \sum_{i\in \bigcup_{l=k}^{m_1}J_l}E_{ii} C\\
				  -\sum_{i\in \bigcup_{l=k}^{m_1}J_l}E_{ii} C  & \sum_{i\in \bigcup_{l=0}^{k-1}J_l}E_{ii}
			\end{pmatrix}\nonumber\\
            &=\begin{pmatrix}
				\sum_{i\in \bigcup_{l=0}^{k-1}J_l}E_{ii} +\sum_{i\in \bigcup_{l=k}^{m_1}J_l}e^{2\ii (-S)^\frac{3}{2}g_{i,-}(z)}E_{ii} & \sum_{i\in \bigcup_{l=k}^{m_1}J_l}E_{ii}                   
                                                           C\\
				  -\sum_{i\in \bigcup_{l=k}^{m_1}J_l}E_{ii} C  & \sum_{i\in \bigcup_{l=0}^{k-1}J_l}E_{ii}
			\end{pmatrix},
    \end{align}
    where
    \begin{align}
        g_{i,-}(z) = \frac{\ii}{6} \left(\widetilde r_l^2-z^2\right)^{\frac{3}{2}}, \quad i \in J_l, \quad k \le l \le m_1,
    \end{align}
    for $z\in (-\widetilde r_l, \widetilde r_l)$. Thus, by taking $S \to -\infty$, we are led to consider the following global parametrix.
	\begin{rhp}\label{def:gloN}
		\hfill
		\begin{enumerate}
			\item[\rm (a)] $N(z)$ is defined and analytic in $\mathbb{C} \setminus [-\widetilde r_{m_1}, \widetilde r_{m_1}]$.
			\item[\rm (b)] 
			 For $x\in (-\widetilde r_{m_1}, \widetilde r_{m_1})$, we have
            \begin{align}\label{eq:Njump}
                N_{+}(x)=N_{-}(x) J_N(x)
            \end{align}
			where 
			\begin{align}
				J_N(x):= \begin{pmatrix}
					\sum_{i\in \bigcup_{l=0}^{k-1}J_l}E_{ii} & \sum_{i\in \bigcup_{l=k}^{{m_1}}J_l}E_{ii} C\\
					-\sum_{i\in \bigcup_{l=k}^{m_1}J_l}E_{ii} C & \sum_{i\in \bigcup_{l=0}^{k-1}J_l}E_{ii}
				\end{pmatrix}, & \quad z \in  (-\widetilde r_{k}, -\widetilde r_{k-1}) \cup  (\widetilde r_{k-1}, \widetilde r_{k}),
			\end{align}
			with $1\leq k \leq m_1$.
			\item[\rm (c)] As $z \to \infty$ with $z \in \mathbb{C} \setminus  [-\widetilde r_{m_1}, \widetilde r_{m_1}]$, we have
			\begin{align}\label{eq:asyN}
				N(z)&= I_{2n}+ \frac{N_1}{z}+\Boh\left(z^{-2}\right),
			\end{align}
			where $N_1$ is independent of $z$. 
		\end{enumerate}
	\end{rhp}
    To solve the above RH problem, we set
    \begin{align}\label{def:gamma}
		\gamma_i(z):=\left(\frac{z-\widetilde r_{k}}{z+\widetilde r_{k}}\right)^{\frac{1}{4}}, \quad z \in \mathbb{C} \backslash[-\widetilde r_{k},\widetilde r_{k}],\quad 
        i\in J_k,    
	\end{align}
    and define
     \begin{align}\label{def:Ni}
        N_1=N_4:= \sum_{i\in \mathcal{I}} \frac{\gamma_i(z)+\gamma_i(z)^{-1}}{2} E_{ii} + \sum_{j\in \mathcal{J}} E_{jj}, \\   
        N_2=-N_3:=-\ii \sum_{i\in \mathcal{I}} \frac{\gamma_i(z)-\gamma_i(z)^{-1}}{2} E_{ii} C. \label{def:N23}
    \end{align}.

    \begin{lemma}
     A solution of RH problem \ref{def:gloN} is given by 
    \begin{align}\label{def:N}
		N(z)= \begin{pmatrix}
			N_1(z) & N_2(z)\\
			N_3(z) & N_4(z)
		\end{pmatrix},
	\end{align}
where the matrix-valued functions $N_i$, $i=1,\ldots,4$, are defined in \eqref{def:Ni} and \eqref{def:N23}. 
    \end{lemma}
	
	% \begin{align}
	% 	N_1=N_4:=\diag\left(\frac{\gamma_1(z)+\gamma_1(z)^{-1}}{2}, \frac{\gamma_2(z)+\gamma_2(z)^{-1}}{2},\cdots,\frac{\gamma_n(z)+\gamma_n(z)^{-1}}{2}\right);\\
	% 	N_2=-N_3:=-\ii\diag\left(\frac{\gamma_1(z)-\gamma_1(z)^{-1}}{2}, \frac{\gamma_2(z)-\gamma_2(z)^{-1}}{2},\cdots,\frac{\gamma_n(z)-\gamma_n(z)^{-1}}{2}\right)C 
	% \end{align}\
	\begin{proof}
It is clear that $N$ is defined and analytic in $\mathbb{C} \setminus [-\widetilde r_{m_1}, \widetilde r_{m_1}]$. To show the jump condition \eqref{eq:Njump}, we restrict our discussion for $x\in (\widetilde r_{k-1}, \widetilde r_{k})$, as the proof for $x\in (-\widetilde r_{k}, -\widetilde r_{k-1})$ follows in a similar manner. If $x\in (\widetilde r_{k-1}, \widetilde r_{k})$, we have
		\begin{align}\label{eq:Njump1}
			&\begin{pmatrix}
				N_{1,-}(z) & N_{2,-}(z)\\
				N_{3,-}(z) & N_{4,-}(z)
			\end{pmatrix} \begin{pmatrix}
					\sum_{i\in \bigcup_{l=0}^{k-1}J_l}E_{ii} & \sum_{i\in \bigcup_{l=k}^{m_1}J_l}E_{ii} C\\
					-\sum_{i\in \bigcup_{l=k}^{m_1}J_l}E_{ii} C & \sum_{i\in \bigcup_{l=0}^{k-1}J_l}E_{ii}
				\end{pmatrix} \nonumber\\
			&=\begin{small}
            \begin{pmatrix}
				N_{1,-}(z)\sum_{i\in \bigcup_{l=0}^{k-1}J_l}E_{ii}-N_{2,-}(z)\sum_{i\in \bigcup_{l=k}^{m_1}J_l}E_{ii} C & N_{1,-}(z)\sum_{i\in \bigcup_{l=k}^{m_1}J_l}E_{ii} C+N_{2,-}\sum_{i\in \bigcup_{l=0}^{k-1}J_l}E_{ii} \\
				N_{3,-}(z)\sum_{i\in \bigcup_{l=0}^{k-1}J_l}E_{ii}-N_{4,-}(z)\sum_{i\in \bigcup_{l=k}^{m_1}J_l}E_{ii} C & N_{3,-}(z)\sum_{i\in \bigcup_{l=k}^{m_1}J_l}E_{ii} C+N_{4,-}(z)\sum_{i\in \bigcup_{l=0}^{k-1}J_l} E_{ii}
			\end{pmatrix}.
            \end{small}
   %          \nonumber\\
			% &=\begin{pmatrix}
			% 	N_{1,+}(z) & N_{2,+}(z)\\
			% 	N_{3,+}(z) & N_{4,+}(z)
			% \end{pmatrix},
		\end{align}
% while the final equality derived from the facts that $\gamma_{i,+}(z)=\ii \gamma_{i,-}(z), \, i \in \bigcup_{l=k}^{m}J_l$, and recalling that $J_0 = \J$. 
For the $(1,1)$-block of the matrix on the right-hand side of the above equation, one has
% As an illustrative example, we compute the $(1,1)$-block explicitly (the other blocks follow by similar computations): 
        \begin{align*}
             & N_{1,-}(z)\sum_{i\in \bigcup_{l=0}^{k-1}J_l}E_{ii}-N_{2,-}(z)\sum_{i\in \bigcup_{l=k}^{m_1}J_l}E_{ii} C \nonumber\\
            &= \sum_{i\in \bigcup_{l=1}^{k-1}J_l}  \frac{\gamma_i(z)+\gamma_i(z)^{-1}}{2}E_{ii} +\sum_{j\in \J} E_{jj} +\ii  \sum_{i\in \bigcup_{l=k}^{m_1}J_l}  \frac{\gamma_{i,-}(z)-\gamma_{i,-}(z)^{-1}}{2}E_{ii}  \nonumber\\
            &=\sum_{i\in \bigcup_{l=1}^{k-1}J_l}  \frac{\gamma_i(z)+\gamma_i(z)^{-1}}{2}E_{ii} +\sum_{j\in \J} E_{jj} +  \sum_{i\in \bigcup_{l=k}^{m_1}J_l}  \frac{\gamma_{i,+}(z)+\gamma_{i,+}(z)^{-1}}{2}E_{ii} \nonumber\\
            &=N_{1,+}(z)\sum_{i\in \bigcup_{l=1}^{k-1}J_l}E_{ii}+\sum_{j\in \J} E_{jj}+N_{1,+}(z)\sum_{i\in \bigcup_{l=k}^{m_1}J_l}E_{ii} = N_{1,+}.
        \end{align*}
By similar calculations, one can show that the right-hand side of \eqref{eq:Njump1} is indeed equal to $\begin{pmatrix}
			 	N_{1,+}(z) & N_{2,+}(z)\\
			 	N_{3,+}(z) & N_{4,+}(z)
			 \end{pmatrix} $, as required.

Finally, the large-$z$ behavior of $N$ can be readily verified from the asymptotic expansions of $\gamma_i(z)$ for $i \in \I$, and we therefore omit the details.
	\end{proof}
    
\subsection{Local parametrix near $-\widetilde{r}_k$}
In a small disc centered at $ -\widetilde{r}_k$, $1\le k \le m_1$, we need to construct a local parametrix as an approximations of $T$, which reads as follows. 
	\begin{rhp}\label{rhp:P-k}
		\hfill
		\begin{enumerate}
			\item[\rm (a)] $ P^{(-k)}(z)$ is defined and analytic in $D(-\widetilde{r}_k, \varepsilon)\setminus \Gamma_T$, where $D(-\widetilde{r}_k, \varepsilon)$ with $\varepsilon$ being a small positive constant and $\Gamma_T$ are defined in \eqref{def:dz0r} and \eqref{def:gammaT}, respectively.
			
			\item[\rm (b)] For $z \in D(-\widetilde{r}_k, \varepsilon) \cap \Gamma_T$, we have 
			\begin{equation}
				P^{(-k)}_+(z)=P^{(-k)}_-(z)	J_{P^{(-k)}}(z),
			\end{equation}
			where 
			\begin{equation}\label{def:JP-k}
				J_{P^{(-k)}}(z):= \begin{cases}
					\begin{pmatrix}
						I_n & \sum_{i\in J_k} e^{2\ii (-S)^\frac{3}{2}g_i(z)}E_{ii}C \\
						0 & I_n
					\end{pmatrix}, & \quad z \in D(-\widetilde{r}_k, \varepsilon)\cap \Gamma_{2m_1+2k-1},  \\
					\begin{pmatrix}
						I_n & 0 \\
						-\sum_{i\in J_k} e^{-2\ii (-S)^\frac{3}{2}g_i(z)}E_{ii}C & I_n
					\end{pmatrix}, & \quad z \in  D(-\widetilde{r}_k, \varepsilon)\cap  \Gamma_{2m_1+2k},  \\
					\begin{pmatrix}
						e^{G_-(z) - G_+(z)} & \sum_{i\in \bigcup_{l=k+1}^{m_1}J_l}E_{ii} C\\
						-\sum_{i\in \bigcup_{l=k+1}^{m_1}J_l}E_{ii} C & \sum_{i \in \bigcup_{l=0}^{k}J_l}E_{ii}
					\end{pmatrix}, & \quad z \in  D(-\widetilde{r}_k, \varepsilon)\cap  (-\widetilde r_{k+1}, -\widetilde r_k),\\
					\begin{pmatrix}
						e^{G_-(z) - G_+(z)} & \sum_{i\in \bigcup_{l=k}^{m_1}J_l}E_{ii} C\\
						-\sum_{i\in \bigcup_{l=k}^{m_1}J_l}E_{ii} C & \sum_{i \in \bigcup_{l=0}^{k-1}J_l}E_{ii}
					\end{pmatrix}, & \quad z \in  D(-\widetilde{r}_k, \varepsilon)\cap  (-\widetilde r_k, -\widetilde r_{k-1}),
				\end{cases} 
			\end{equation}
			with $1\leq k \leq m_1$ and $-\widetilde r_{m_1+1}:=-\infty$\footnote{Note that $P^{(-m_1)}$ is analytic on the interval $ D(-\widetilde{r}_k, \varepsilon)\cap (-\infty, -\widetilde r_{m_1})$. Thus, the definition of $-\widetilde r_{m_1+1}$ is justified and is introduced solely for the sake of formal consistency.}. Here, we adopt the convention that when the subscript surpasses the superscript, the resulting set is empty and the summation is a zero matrix.
			
			\item[\rm (c)]As $S \to -\infty$,  we have the matching condition
			\begin{equation}\label{eq:matchhatP-k}
				P^{(-k)}(z)=\left( I+ \Boh\left((-S)^{-\frac{3}{2}}\right) \right)  N(z),\quad z \in \partial D(-\widetilde{r}_k, \varepsilon),
			\end{equation}
			where $N(z)$ is given in \eqref{def:N}.
		\end{enumerate}
	\end{rhp}
  %     \begin{remark}
  % 	For the union operator, we adopt the convention that when the subscript surpasses the superscript, the resulting set is empty and the summation is zero matrix.
  % \end{remark}
	The above RH problem can be solved by using the Airy parametrix $\Phi^{({\Ai})}$ introduced in Appendix \ref{airy}. To this end, we introduce the function 
	\begin{align}\label{def:fk}
		f^{(-k)}(z)&:=\begin{cases}
			\left(\frac{3\ii}{2}g_i(z)\right)^{\frac{2}{3}}, \quad & \Im z >0,\\
			\left(-\frac{3\ii}{2}g_i(z)\right)^{\frac{2}{3}}, \quad & \Im z <0,
		\end{cases}\nonumber\\
		&=  2^{-\frac 13} \widetilde{r}_k (z+\widetilde{r}_k) +\Boh\left((z+\widetilde{r}_k)^2\right), \quad \textrm{ $z \to -\widetilde{r}_k$},
	\end{align}
	where $i\in J_k$ and $g_i(z)$ is defined in \eqref{def:gi}. Clearly, $f^{(-k)}(z)$ is analytic in $D(-\widetilde{r}_k, \varepsilon)$ and is a conformal mapping.

%      And we define
% 	\begin{align}
% 		\widetilde{g}_i(z)=\begin{cases}
% 			-g_i(z), \quad & \Im z >0,\\
% 			g_i(z), \quad & \Im z <0,
% 		\end{cases}
% 	\end{align}
% where $i\in \bigcup_{l=k+1}^{m_1} J_l$, then each $\widetilde{g}_i(z)$ is analytic in $D(-\widetilde{r}_k, \varepsilon)$.
	
We now define 
	\begin{align}\label{def:P-k}
			P^{(-k)}(z)=E^{(-k)}(z)\Phi^{(-k)}(z) U^{(-k)}(z) e^{-\sum_{i\in \bigcup_{l=k}^{m_1} J_l}\ii (-S)^{\frac{3}{2}}g_i(z)E_{ii}\otimes \sigma_3},
	\end{align}
	where \begin{align}\label{def:E-k}
		E^{(-k)}(z)&:=N(z)U^{(-k)}(z)^{-1}\times\nonumber\\
		&\left[\begin{pmatrix}
			\sum_{i\in \bigcup_{l=0}^{k-1} J_l}E_{ii}+\sum_{i\in \bigcup_{l=k+1}^{m_1} J_l}e^{-\ii (-S)^{\frac{3}{2}}\widetilde{g}_i(z)}E_{ii} &0\\
				0& \sum_{i\in \bigcup_{l=0}^{k-1} J_l} E_{ii}+\sum_{i\in \bigcup_{l=k+1}^{m_1} J_l}e^{\ii (-S)^{\frac{3}{2}}\widetilde{g}_i(z)}E_{ii}
		\end{pmatrix}\right.\nonumber\\
	&\left.+\frac{1}{\sqrt{2}}\begin{pmatrix}
		\left(-S f^{(-k)}(z)\right)^{\frac{1}{4}}\sum_{i \in J_k}E_{ii} & -\ii \left(-S f^{(-k)}(z)\right)^{-\frac{1}{4}}\sum_{i \in J_k}E_{ii}\\
		-\ii \left(-S f^{(-k)}(z)\right)^{\frac{1}{4}}\sum_{i \in J_k}E_{ii} & \left(-S f^{(-k)}(z)\right)^{-\frac{1}{4}}\sum_{i \in J_k}E_{ii}
	\end{pmatrix}
	\right],
	\end{align}  
    with 
    \begin{align}
		\widetilde{g}_i(z):=\begin{cases}
			-g_i(z), \quad & \Im z >0,\\
			g_i(z), \quad & \Im z <0,
		\end{cases}\qquad i\in \bigcup_{l=k+1}^{m_1} J_l,
	\end{align}
	\begin{align}\label{def:Phi-k}
		\Phi^{(-k)}(z):=\begin{pmatrix}
			 \sum_{i\notin J_k}E_{ii}+\sum_{i\in J_k}\Phi_{11}^{({\Ai})}(-Sf^{(-k)}(z)) E_{ii} & \sum_{i\in J_k}\Phi_{12}^{({\Ai})}(-Sf^{(-k)}(z))E_{ii} \\
			\sum_{i\in J_k}\Phi_{21}^{({\Ai})}(-Sf^{(-k)}(z))E_{ii} & \sum_{i\notin J_k}E_{ii}+\sum_{i\in J_k}\Phi_{22}^{({\Ai})}(-Sf^{(-k)}(z)) E_{ii}
		\end{pmatrix},
	\end{align}
and
	\begin{align}
		U^{(-k)}(z):=\begin{cases}
			\begin{pmatrix}
				\sum_{i\notin J_k}E_{ii} & \sum_{i\in J_k}E_{ii}C \\
				\sum_{i\in J_k}E_{ii} & \sum_{i\notin J_k}E_{ii}
			\end{pmatrix}, &\quad \Im z>0, \\
			\begin{pmatrix}
				I_n-C_k^2 & -C_k \\
				C_k & \sum_{i\notin J_k}E_{ii}-C_k^2-\sum_{i\in J_k}E_{ii}C
			\end{pmatrix}, &\quad \Im z<0.
		\end{cases}
	\end{align}
For later use, we note that  
    \begin{align} \label{eq:Ukinverse}
        U^{(-k)}(z)^{-1}=\begin{cases}
           \begin{pmatrix}
				\sum_{i\notin J_k}E_{ii} & \sum_{i\in J_k}E_{ii}C \\
				\sum_{i\in J_k}E_{ii} & \sum_{i\notin J_k}E_{ii}
			\end{pmatrix}, &\quad \Im z>0, \\
            \begin{pmatrix}
				I_n-C_k^2 & C_k \\
				-C_k & \sum_{i\notin J_k}E_{ii}-C_k^2-\sum_{i\in J_k}E_{ii}C
			\end{pmatrix}, &\quad \Im z<0,
        \end{cases}
     \end{align}
where $i\notin J_k$ means that $i \in \{1,\cdots,n\} \setminus J_k$.
    
	\begin{proposition}\label{pro:P-k}
		The function $P^{(-k)}(z)$ defined in \eqref{def:P-k} solves RH problem \ref{rhp:P-k}.
	\end{proposition}
\begin{proof}
We start with showing that $E^{(-k)}(z)$ is analytic in $D(-\widetilde{r}_k, \varepsilon)$. From \eqref{def:E-k}, the only possible jump is on $(-\widetilde{r}_k-\varepsilon, -\widetilde{r}_k+\varepsilon)$. For $z\in (-\widetilde{r}_k-\varepsilon, -\widetilde{r}_k)$, it follows from   \eqref{def:E-k} that
	\begin{align*}
		&E^{(-k)}_-(z)^{-1}E^{(-k)}_+(z) \nonumber \\
		&=\begin{pmatrix}
			\sum_{i\in \bigcup_{l=0}^{k-1} J_l}E_{ii} & 0\\
			0 & \sum_{i\in \bigcup_{l=0}^{k-1} J_l}E_{ii}
		\end{pmatrix}+\frac{1}{2}\begin{pmatrix}
	\left(-S\right)^{-\frac{1}{4}}	f^{(-k)}_-(z)^{-\frac{1}{4}}\sum_{i \in J_k}E_{ii} &  \ii \left(-S\right)^{-\frac{1}{4}} f^{(-k)}_-(z)^{-\frac{1}{4}}\sum_{i \in J_k}E_{ii}\\
		\ii \left(-S\right)^{\frac{1}{4}} f^{(-k)}_-(z)^{\frac{1}{4}}\sum_{i \in J_k}E_{ii} & \left(-S\right)^{\frac{1}{4}} f^{(-k)}_-(z)^{\frac{1}{4}}\sum_{i \in J_k}E_{ii}
	\end{pmatrix}\nonumber \\
&\quad\times\begin{pmatrix}
	\sum_{i \in J_k}E_{ii} & 0\\
	0 & -\sum_{i \in J_k}E_{ii}C
\end{pmatrix} \begin{pmatrix}
	0 & \sum_{i\in J_k}E_{ii}\\
	\sum_{i\in J_k}E_{ii}C &0
\end{pmatrix} \nonumber\\
 &\quad\times\begin{pmatrix}
	\left(-S\right)^{\frac{1}{4}} f^{(-k)}_+(z)^{\frac{1}{4}}\sum_{i \in J_k}E_{ii} & -\ii \left(-S\right)^{-\frac{1}{4}} f^{(-k)}_+(z)^{-\frac{1}{4}}\sum_{i \in J_k}E_{ii}\\
	-\ii \left(-S\right)^{\frac{1}{4}} f^{(-k)}_+(z)^{\frac{1}{4}}\sum_{i \in J_k}E_{ii} & \left(-S\right)^{-\frac{1}{4}} f^{(-k)}_+(z)^{-\frac{1}{4}}\sum_{i \in J_k}E_{ii}
\end{pmatrix}\nonumber\\
&\quad+\begin{pmatrix}
	\sum_{i\in \bigcup_{l=k+1}^{m_1} J_l}e^{\ii (-S)^{\frac{3}{2}}\widetilde{g}_{i}(z)}E_{ii} & 0\\
	0 & \sum_{i\in \bigcup_{l=k+1}^{m_1} J_l}e^{-\ii (-S)^{\frac{3}{2}}\widetilde{g}_{i}(z)}E_{ii}
\end{pmatrix}\begin{pmatrix}
0 & -C_k \\
C_k & 0
\end{pmatrix}\begin{pmatrix}
0 & C_k \\
-C_k & 0
\end{pmatrix}\nonumber\\
&\quad\times \begin{pmatrix}
	\sum_{i\in \bigcup_{l=k+1}^{m_1} J_l}e^{-\ii (-S)^{\frac{3}{2}}\widetilde{g}_{i}(z)}E_{ii} & 0\\
	0 & \sum_{i\in \bigcup_{l=k+1}^{m_1} J_l}e^{\ii (-S)^{\frac{3}{2}}\widetilde{g}_{i}(z)}E_{ii}
\end{pmatrix}\nonumber\\
&= \begin{pmatrix}
			\sum_{i\in \bigcup_{l=0}^{k-1} J_l}E_{ii} & 0\\
			0 & \sum_{i\in \bigcup_{l=0}^{k-1} J_l}E_{ii}
		\end{pmatrix}+\frac{1}{2}\begin{pmatrix}
		\left(-S\right)^{-\frac{1}{4}} f^{(-k)}_-(z)^{-\frac{1}{4}}\sum_{i \in J_k}E_{ii} & \ii \left(-S\right)^{-\frac{1}{4}} f^{(-k)}_-(z)^{-\frac{1}{4}}\sum_{i \in J_k}E_{ii}\\
		\ii \left(-S\right)^{\frac{1}{4}} f^{(-k)}_-(z)^{\frac{1}{4}}\sum_{i \in J_k}E_{ii} & \left(-S\right)^{\frac{1}{4}} f^{(-k)}_-(z)^{\frac{1}{4}}\sum_{i \in J_k}E_{ii}
	\end{pmatrix}\nonumber \\
&\quad\times \begin{pmatrix}
	0 & \sum_{i\in J_k}E_{ii}\\
	-\sum_{i\in J_k}E_{ii} &0
\end{pmatrix} \begin{pmatrix}
	\left(-S\right)^{\frac{1}{4}} f^{(-k)}_+(z)^{\frac{1}{4}}\sum_{i \in J_k}E_{ii} & -\ii \left(-S\right)^{-\frac{1}{4}} f^{(-k)}_+(z)^{-\frac{1}{4}}\sum_{i \in J_k}E_{ii}\\
	-\ii \left(-S\right)^{\frac{1}{4}} f^{(-k)}_+(z)^{\frac{1}{4}}\sum_{i \in J_k}E_{ii} & \left(-S\right)^{-\frac{1}{4}} f^{(-k)}_+(z)^{-\frac{1}{4}}\sum_{i \in J_k}E_{ii}
\end{pmatrix}\nonumber\\ 
&\quad +\begin{pmatrix}
			\sum_{i\in \bigcup_{l=k+1}^{m_1} J_l}E_{ii} & 0\\
			0 & \sum_{i\in \bigcup_{l=k+1}^{m_1} J_l}E_{ii}
		\end{pmatrix} \nonumber\\
&= \begin{pmatrix}
			\sum_{i\in \bigcup_{l=0}^{k-1} J_l}E_{ii} & 0\\
			0 & \sum_{i\in \bigcup_{l=0}^{k-1} J_l}E_{ii}
		\end{pmatrix} +\begin{pmatrix}
			\sum_{i\in  J_k}E_{ii} & 0\\
			0 & \sum_{i\in J_k}E_{ii}
		\end{pmatrix} +\begin{pmatrix}
			\sum_{i\in \bigcup_{l=k+1}^{m_1} J_l}E_{ii} & 0\\
			0 & \sum_{i\in \bigcup_{l=k+1}^{m_1} J_l}E_{ii}
		\end{pmatrix}\nonumber\\
        \nonumber\\
&=I_{2n},
\end{align*}
where we have made use of the fact $f^{(-k)}_+(z)^{\frac{1}{4}} = \ii f^{(-k)}_-(z)^{\frac{1}{4}}$, $z\in (-\widetilde{r}_k-\varepsilon, -\widetilde{r}_k)$, in the third equality.

For $z\in (-\widetilde{r}_k, -\widetilde{r}_k+\varepsilon)$, one has
\begin{align*}
	&E^{(-k)}_-(z)^{-1}E^{(-k)}_+(z) \nonumber \\
		&=\begin{pmatrix}
		\sum_{i\in \bigcup_{l=0}^{k-1} J_l}E_{ii} & 0\\
		0 & \sum_{i\in \bigcup_{l=0}^{k-1} J_l}E_{ii}
	\end{pmatrix}+\frac{1}{2}\begin{pmatrix}
		\left(-S\right)^{-\frac{1}{4}}f^{(-k)}(z)^{-\frac{1}{4}}\sum_{i \in J_k}E_{ii} & \ii \left(-S\right)^{-\frac{1}{4}} f^{(-k)}(z)^{-\frac{1}{4}}\sum_{i \in J_k}E_{ii}\\
		\ii \left(-S\right)^{\frac{1}{4}} f^{(-k)}(z)^{\frac{1}{4}}\sum_{i \in J_k}E_{ii} & \left(-S\right)^{\frac{1}{4}} f^{(-k)}(z)^{\frac{1}{4}}\sum_{i \in J_k}E_{ii}
	\end{pmatrix}\nonumber\\
&\quad \times\begin{pmatrix}
	\sum_{i \in J_k}E_{ii} & 0\\
	0 & -\sum_{i \in J_k}E_{ii}C
\end{pmatrix}
\begin{pmatrix}
	0& \sum_{i \in J_k}E_{ii}C\\
	-\sum_{i \in J_k}E_{ii}C& 0
\end{pmatrix}
\begin{pmatrix}
	0 & \sum_{i\in J_k}E_{ii}\\
	\sum_{i\in J_k}E_{ii}C &0
\end{pmatrix}\nonumber\\
&\quad \times \begin{pmatrix}
	\left(-S\right)^{\frac{1}{4}} f^{(-k)}(z)^{\frac{1}{4}}\sum_{i \in J_k}E_{ii} & -\ii \left(-S\right)^{-\frac{1}{4}} f^{(-k)}(z)^{-\frac{1}{4}}\sum_{i \in J_k}E_{ii}\\
	-\ii \left(-S\right)^{\frac{1}{4}} f^{(-k)}(z)^{\frac{1}{4}}\sum_{i \in J_k}E_{ii} & \left(-S\right)^{-\frac{1}{4}} f^{(-k)}(z)^{-\frac{1}{4}}\sum_{i \in J_k}E_{ii}
\end{pmatrix}\nonumber\\
&\quad +\begin{pmatrix}
	\sum_{i\in \bigcup_{l=k+1}^{m_1} J_l}e^{\ii (-S)^{\frac{3}{2}}\widetilde{g}_{i}(z)}E_{ii} & 0\\
	0 & \sum_{i\in \bigcup_{l=k+1}^{m_1} J_l}e^{-\ii (-S)^{\frac{3}{2}}\widetilde{g}_{i}(z)}E_{ii}
\end{pmatrix}\begin{pmatrix}
	0 & -C_k \\
	C_k & 0
\end{pmatrix}\begin{pmatrix}
	0 & C_k \\
	-C_k & 0
\end{pmatrix}\nonumber\\
&\quad \times \begin{pmatrix}
	\sum_{i\in \bigcup_{l=k+1}^{m_1} J_l}e^{-\ii (-S)^{\frac{3}{2}}\widetilde{g}_{i}(z)}E_{ii} & 0\\
	0 & \sum_{i\in \bigcup_{l=k+1}^{m_1} J_l}e^{\ii (-S)^{\frac{3}{2}}\widetilde{g}_{i}(z)}E_{ii}
\end{pmatrix}\nonumber\\
&= \begin{pmatrix}
			\sum_{i\in \bigcup_{l=0}^{k-1} J_l}E_{ii} & 0\\
			0 & \sum_{i\in \bigcup_{l=0}^{k-1} J_l}E_{ii}
		\end{pmatrix} + \frac{1}{2}\begin{pmatrix}
		\left(-S\right)^{-\frac{1}{4}} f^{(-k)}(z)^{-\frac{1}{4}}\sum_{i \in J_k}E_{ii} & \ii \left(-S\right)^{-\frac{1}{4}} f^{(-k)}(z)^{-\frac{1}{4}}\sum_{i \in J_k}E_{ii}\\
		\ii \left(-S\right)^{\frac{1}{4}} f^{(-k)}(z)^{\frac{1}{4}}\sum_{i \in J_k}E_{ii} & \left(-S\right)^{\frac{1}{4}} f^{(-k)}(z)^{\frac{1}{4}}\sum_{i \in J_k}E_{ii}
	\end{pmatrix}\nonumber\\
&\quad \times \begin{pmatrix}
	\left(-S\right)^{\frac{1}{4}} f^{(-k)}(z)^{\frac{1}{4}}\sum_{i \in J_k}E_{ii} & -\ii \left(-S\right)^{-\frac{1}{4}} f^{(-k)}(z)^{-\frac{1}{4}}\sum_{i \in J_k}E_{ii}\\
	-\ii \left(-S\right)^{\frac{1}{4}} f^{(-k)}(z)^{\frac{1}{4}}\sum_{i \in J_k}E_{ii} & \left(-S\right)^{-\frac{1}{4}} f^{(-k)}(z)^{-\frac{1}{4}}\sum_{i \in J_k}E_{ii}
\end{pmatrix}+ \begin{pmatrix}
			\sum_{i\in \bigcup_{l=k+1}^{m_1} J_l}E_{ii} & 0\\
			0 & \sum_{i\in \bigcup_{l=k+1}^{m_1} J_l}E_{ii}
		\end{pmatrix}\nonumber\\
&= \begin{pmatrix}
			\sum_{i\in \bigcup_{l=0}^{k-1} J_l}E_{ii} & 0\\
			0 & \sum_{i\in \bigcup_{l=0}^{k-1} J_l}E_{ii}
		\end{pmatrix} +\begin{pmatrix}
			\sum_{i\in  J_k}E_{ii} & 0\\
			0 & \sum_{i\in J_k}E_{ii}
		\end{pmatrix} +\begin{pmatrix}
			\sum_{i\in \bigcup_{l=k+1}^{m_1} J_l}E_{ii} & 0\\
			0 & \sum_{i\in \bigcup_{l=k+1}^{m_1} J_l}E_{ii}
		\end{pmatrix}\nonumber\\
&=I_{2n}.
\end{align*}
Thus,  $E^{(-k)}(z)$ is analytic in $D(-\widetilde{r}_k, \varepsilon)\setminus \{-\widetilde{r}_k\}$. Moreover, from \eqref{def:N} and \eqref{def:fk}, we have 
\begin{align}
    E^{(-k)}(-\widetilde{r}_k)&=\begin{pmatrix}
        \sum_{j \in \J} E_{jj} & 0\\
        0 & \sum_{j \in \J} E_{jj} &
    \end{pmatrix} +\frac{(-S)^{\frac{1}{4}}\widetilde{r}_k^{\frac{1}{2}}}{2^{\frac{1}{3}}}\begin{pmatrix}
        \sum_{i\in J_k}e^{-\frac{\pi \ii}{4}} E_{ii} &0 \\
            \sum_{i\in J_k}e^{\frac{\pi \ii}{4}} E_{ii} &0
    \end{pmatrix} \nonumber\\
    &\quad +\begin{pmatrix}
        \sum_{i\notin J_k\cup \J}\frac{\gamma_i(-\widetilde{r}_k)+\gamma_i(-\widetilde{r}_k)^{-1}}{2}E_{ii} & -\ii \sum_{i\notin J_k\cup \J}\frac{\gamma_i(-\widetilde{r}_k)+\gamma_i(-\widetilde{r}_k)^{-1}}{2}E_{ii} \\
        \ii \sum_{i\notin J_k\cup \J}\frac{\gamma_i(-\widetilde{r}_k)+\gamma_i(-\widetilde{r}_k)^{-1}}{2}E_{ii} & \sum_{i\notin J_k\cup \J}\frac{\gamma_i(-\widetilde{r}_k)+\gamma_i(-\widetilde{r}_k)^{-1}}{2}E_{ii}
    \end{pmatrix},
\end{align}
where $\gamma_i(z)$ is defined in \eqref{def:gamma}. As a consequence,  $E^{(-k)}(z)$ is analytic in $D(-\widetilde{r}_k, \varepsilon)$. 

The jump condition of $P^{(-k)}$ in \eqref{def:JP-k} can be verified from the analyticity of $E^{(-k)}(z)$ and the jump condition of $\Phi^{(\Ai)}$ in \eqref{jump:Airy}. Indeed, one can check
\begin{align*}
	&\Phi^{(-k)}_-(z)^{-1}\Phi^{(-k)}_+(z)
    \\
    &=
    \begin{cases}
		\begin{pmatrix}
			I_n & 0 \\
			\sum_{i\in J_k} E_{ii} & I_n
		\end{pmatrix}, & \quad z \in D(-\widetilde{r}_k, \varepsilon)\cap  \Gamma_{2m_1+2k-1}\cup \Gamma_{2m_1+2k},  \\
		\begin{pmatrix}
			\sum_{i\notin J_k}E_{ii} & \sum_{i\in J_k}E_{ii}\\
			-\sum_{i\in J_k}E_{ii} &\sum_{i\notin J_k}E_{ii}
		\end{pmatrix}, & \quad z \in D(-\widetilde{r}_k, \varepsilon)\cap  (-\widetilde r_{k+1}, -\widetilde r_k),\\
		\begin{pmatrix}
			I_n & \sum_{i\in J_k}E_{ii}\\
			0 & I_n
		\end{pmatrix}, & \quad z \in D(-\widetilde{r}_k, \varepsilon)\cap  (-\widetilde r_k, -\widetilde r_{k-1}).
	\end{cases} 
\end{align*}
This, together with \eqref{eq:Ukinverse}, implies that
\begin{align*}
	& U^{(-k)}_-(z)^{-1}	\Phi^{(-k)}_-(z)^{-1}\Phi^{(-k)}_+(z)U^{(-k)}_+(z) \\
    & =\begin{cases}
		\begin{pmatrix}
			I_n & \sum_{i\in J_k} E_{ii}C \\
			0 & I_n
		\end{pmatrix}, & \quad z \in  D(-\widetilde{r}_k, \varepsilon)\cap  \Gamma_{2m_1+2k-1},  \\
		\begin{pmatrix}
			I_n & 0 \\
			-\sum_{i\in J_k} E_{ii}C & I_n
		\end{pmatrix}, & \quad z \in D(-\widetilde{r}_k, \varepsilon)\cap  \Gamma_{2m_1+2k},  \\
		\begin{pmatrix}
			I_n & C_k\\
			-C_k & I_n-C_k^2
		\end{pmatrix}, & \quad z \in D(-\widetilde{r}_k, \varepsilon)\cap (-\widetilde r_{k+1}, -\widetilde r_k),\\
		\begin{pmatrix}
			I_n & C_{k-1}\\
			-C_{k-1} & I_n-C_{k-1}^2
		\end{pmatrix}, & \quad z \in D(-\widetilde{r}_k, \varepsilon)\cap  (-\widetilde r_k, -\widetilde r_{k-1}).
	\end{cases} 
\end{align*} 
The jump condition now follows directly from the definition of $P^{(-k)}$ in \eqref{def:P-k}.

Finally, from \eqref{def:N} and \eqref{def:fk}, the matching condition \eqref{eq:matchhatP-k} can be verified directly using asymptotic behavior of the Airy parametrix at infinity in \eqref{infty:Ai}. We only check \eqref{def:fk} for $\Im z>0$ in what follows. If $S\to -\infty$, one has 
\begin{align}
	&P^{(-k)}(z) N(z)^{-1}= E^{(-k)}(z)\Phi^{(-k)}(z) U^{(-k)}(z) e^{-\sum_{i\in \bigcup_{l=k}^{m_1} J_l}\ii (-S)^{\frac{3}{2}}g_i(z)E_{ii}\otimes \sigma_3} N(z)^{-1}\nonumber\\
	&=N(z)U^{(-k)}(z)^{-1}\nonumber\\
	&\quad \times\left[\begin{pmatrix}
		\sum_{i\in \bigcup_{l=0}^{k-1} J_l}E_{ii}+\sum_{i\in \bigcup_{l=k+1}^{m_1} J_l}e^{-\ii (-S)^{\frac{3}{2}}\widetilde{g}_i(z)}E_{ii} &0\\
		0& \sum_{i\in \bigcup_{l=0}^{k-1} J_l} E_{ii}+\sum_{i\in \bigcup_{l=k+1}^{m_1} J_l}e^{\ii (-S)^{\frac{3}{2}}\widetilde{g}_i(z)}E_{ii}
	\end{pmatrix}\right.\nonumber\\
	&\left. \quad +\frac{1}{\sqrt{2}}\begin{pmatrix}
		\left(-S\right)^{\frac{1}{4}} f^{(-k)}(z)^{\frac{1}{4}}\sum_{i \in J_k}E_{ii} & -\ii \left(-S\right)^{\frac{1}{4}} f^{(-k)}(z)^{-\frac{1}{4}}\sum_{i \in J_k}E_{ii}\\
		-\ii \left(-S\right)^{\frac{1}{4}} f^{(-k)}(z)^{\frac{1}{4}}\sum_{i \in J_k}E_{ii} & \left(-S\right)^{\frac{1}{4}} f^{(-k)}(z)^{-\frac{1}{4}}\sum_{i \in J_k}E_{ii}
	\end{pmatrix}
	\right]\nonumber\\
	& \quad \times \begin{pmatrix}
		\sum_{i \notin J_k}E_{ii}+\sum_{i\in J_k}\Phi_{11}^{({\Ai})}(-S f^{(-k)}(z)) E_{ii} & \sum_{i\in J_k}\Phi_{12}^{({\Ai})}(-S f^{(-k)}(z))E_{ii} \\
		\sum_{i\in J_k}\Phi_{21}^{({\Ai})}(-S f^{(-k)}(z))E_{ii} & \sum_{i \notin J_k}E_{ii}+\sum_{i\in J_k}\Phi_{22}^{({\Ai})}(-S f^{(-k)}(z)) E_{ii}
	\end{pmatrix} U^{(-k)}(z) \nonumber\\
&\quad \times e^{\left(-\sum_{i\in \bigcup_{l=k}^{m_1} J_l}\ii (-S)^{\frac{3}{2}}g_i(z)E_{ii}\right)\otimes \sigma_3} N(z)^{-1}\nonumber\\
&=N(z)U^{(-k)}(z)^{-1} \nonumber\\
&\quad \times \left[\begin{pmatrix}
	\sum_{i\in \bigcup_{l=0}^{k-1} J_l}E_{ii}+\sum_{i\in \bigcup_{l=k+1}^{m_1} J_l}e^{-\ii (-S)^{\frac{3}{2}}\widetilde{g}_i(z)}E_{ii} &0\\
	0& \sum_{i\in \bigcup_{l=0}^{k-1} J_l} E_{ii}+\sum_{i\in \bigcup_{l=k+1}^{m_1} J_l}e^{\ii (-S)^{\frac{3}{2}}\widetilde{g}_i(z)}E_{ii}
\end{pmatrix}\right.\nonumber\\
&\left. \quad +\left(I_{2n}+\Boh(S^{-\frac{3}{2}})\right)\begin{pmatrix}
	e^{-\frac{2}{3} \left(-S f^{(-k)}(z)\right)^{3/2}}\sum_{i \in J_k}E_{ii} & 0\\
0 & e^{\frac{2}{3} \left(-S f^{(-k)}(z)\right)^{3/2}}\sum_{i \in J_k}E_{ii}
\end{pmatrix}
\right] U^{(-k)}(z) \nonumber\\
&\quad \times e^{\left(-\sum_{i\in \bigcup_{l=k}^{m_1} J_l}\ii (-S)^{\frac{3}{2}}g_i(z)E_{ii}\right)\otimes \sigma_3} N(z)^{-1}\nonumber\\
&=N(z) \left[\begin{pmatrix}
    \sum_{i \notin J_k} E_{ii} & 0\\
    0 & \sum_{i \notin J_k} E_{ii}
\end{pmatrix} + \begin{pmatrix}
    0 & \sum_{i \in J_k} E_{ii}\\
    \sum_{i \in J_k}E_{ii}C & 0
\end{pmatrix} \right.\nonumber\\
&\left. \quad \times  \left( I_{2n}+ \Boh(S^{-\frac{3}{2}}) \right)
\begin{pmatrix}
	e^{-\frac{2}{3} \left(-S f^{(-k)}(z)\right)^{3/2}}\sum_{i \in J_k}E_{ii} & 0\\
0 & e^{\frac{2}{3} \left(-S f^{(-k)}(z)\right)^{3/2}}\sum_{i \in J_k}E_{ii}
\end{pmatrix} \right.\nonumber\\
&\left. \quad \times \begin{pmatrix}
    0 & \sum_{i \in J_k} E_{ii}C\\
    \sum_{i \in J_k}E_{ii} & 0
\end{pmatrix} e^{\left(-\sum_{i\in J_k}\ii (-S)^{\frac{3}{2}}g_i(z)E_{ii}\right)\otimes \sigma_3} \right]N(z)^{-1}\nonumber\\
&=N(z) \left( I_{2n}+ \Boh\left((-S)^{-\frac{3}{2}}\right) \right) N(z)^{-1} =I_{2n}+ \Boh\left((-S)^{-\frac{3}{2}}\right),
\end{align}
as required.
% where the last equation follows from the definition of $N$ in \eqref{def:N}.

This completes the proof of Proposition \ref{pro:P-k}.
\end{proof}

\subsection{Local parametrix near $\widetilde{r}_k$}
Near each $\widetilde{r}_k$, $1\le k \le m_1$, we need to construct the following local parametrix. 
	% We then move to the construction of the local parametrix near $\widetilde{r}_k$, $1\le k \le m_1$.
	\begin{rhp}\label{rhp:Pk}
		\hfill
		\begin{enumerate}
			\item[\rm (a)] $ P^{(k)}(z)$ is defined and analytic in $D(\widetilde{r}_k, \varepsilon)\setminus \Gamma_T$, where $\Gamma_T$ is defined in \eqref{def:gammaT}.
			\item[\rm (b)] For $z \in D(\widetilde{r}_k, \varepsilon) \cap \Gamma_T$, we have 
			\begin{equation}
				P^{(k)}_+(z)=P^{(k)}_-(z)J_{P^{(k)}}(z),
			\end{equation}
			where 
			\begin{equation*}\label{def:JPk}
				J_{P^{(k)}}(z):= \begin{cases}
					\begin{pmatrix}
						I_n & \sum_{i\in J_k} e^{2\ii (-S)^\frac{3}{2}g_i(z)}E_{ii}C \\
						0 & I_n
					\end{pmatrix}, & \quad z \in D(\widetilde{r}_k, \varepsilon) \cap \Gamma_{2k-1},  \\
					\begin{pmatrix}
						I_n & 0 \\
						-\sum_{i\in J_k} e^{-2\ii (-S)^\frac{3}{2}g_i(z)}E_{ii}C & I_n
					\end{pmatrix}, & \quad z \in D(\widetilde{r}_k, \varepsilon) \cap \Gamma_{2k},  \\
					\begin{pmatrix}
						e^{G_-(z) - G_+(z)} & e^{G_-(z)}C_ke^{G_+(z)}\\
						-e^{-G_-(z)}C_ke^{-G_+(z)} & e^{-G_-(z)+G_+(z)}(I_n-C_k^2)
					\end{pmatrix}, & \quad z \in D(\widetilde{r}_k, \varepsilon) \cap (\widetilde r_k, \widetilde r_{k+1}), \\
					\begin{pmatrix}
						e^{G_-(z) - G_+(z)} & e^{G_-(z)}C_{k-1}e^{G_+(z)}\\
						-e^{-G_-(z)}C_{k-1}e^{-G_+(z)} & e^{-G_-(z)+G_+(z)}(I_n-C_{k-1}^2)
					\end{pmatrix}, & \quad z \in D(\widetilde{r}_k, \varepsilon) \cap (\widetilde r_{k-1}, \widetilde r_k),
				\end{cases} 
			\end{equation*}
			with $1\leq k \leq m_1$ and $\widetilde r_{m_1+1}:=\infty$. 
			\item[\rm (c)]As $S \to -\infty$,  we have the matching condition
			\begin{equation}\label{eq:matchPk}
				P^{(k)}(z)=\left( I+ \Boh\left((-S)^{-\frac{3}{2}}\right) \right)    N(z),\quad z \in \partial D(\widetilde{r}_k, \varepsilon),
			\end{equation}
			where $N(z)$ is given in \eqref{def:N}.
		\end{enumerate}
	\end{rhp}
	One can directly verify that 
\begin{equation}\label{rk}
		P^{(k)}(z) : =\what \sigma_1 P^{(-k)}(-z)\what \sigma_1,
	\end{equation}
where 
\begin{equation}\label{def:hatsigma1}
    \what \sigma_1:=\begin{pmatrix}
		    0 & I_n \\
                I_n & 0
		\end{pmatrix}
\end{equation}
and the function
$P^{(k)}(z)$  defined in \eqref{def:P-k} solves RH problem \ref{rhp:Pk}, we omit the details here.  
    
	\subsection{Final transformation}
	The final transformation is defined by
	\begin{align}\label{def:R}
		R(z) = \begin{cases}
			T(z) P^{(-k)}(z)^{-1}, & \quad z \in D(-\widetilde{r}_k, \varepsilon),\\
			T(z) P^{(k)}(z)^{-1}, & \quad z \in D(\widetilde{r}_k, \varepsilon),\\
			T(z) N(z)^{-1}, & \quad \textrm{elsewhere}.
		\end{cases}
	\end{align}
	Based on the RH problems for $T$, $N$, and $P^{(\pm k)}$, it can be verified  that $\pm \widetilde{r}_k$ are actually removable singularities of $R$. Consequently, $R$ satisfies the following RH problem.
	
	\begin{rhp}\label{rhp:R}
		\hfill
		\begin{itemize}
			\item [\rm{(a)}] $R(z)$ is defined and analytic in $\mathbb{C} \setminus \Gamma_{R}$, where
			\begin{equation}
				\Gamma_{R}:=\Gamma_T \cup\bigcup_{k=1}^{m_1}\left( \partial D(-\widetilde{r}_k,\varepsilon) \cup \partial D(\widetilde{r}_k,\varepsilon)\right) \setminus \{ \bigcup_{k=1}^{m_1}\left(D(-\widetilde{r}_k,\varepsilon) \cup D(\widetilde{r}_k,\varepsilon)\right) \}.		\end{equation}
			%and the orientation of $\partial D(\pm \widetilde{r}_k,\varepsilon)$ are counterclockwise.
			\item [\rm{(b)}] For $z \in \Gamma_{R}$, we have
			\begin{equation}\label{eq:Rjump}
				R_+(z) = R_-(z) J_R (z),
			\end{equation}
			where 
			\begin{equation}\label{def:JR}
				J_R(z) := \begin{cases}
					P^{(\pm k)}(z) N(z)^{-1}, & \quad z \in \partial D(\pm\widetilde{r}_k, \varepsilon), \quad k=1,\ldots,m_1, \\
					% P^{(k)}(z) N(z)^{-1}, & \quad z \in \partial D(\widetilde{r}_k, \varepsilon),\\
					N_{-}(z) J_T(z) N_+(z)^{-1}, & \quad z \in \Gamma_{R} \setminus \{\bigcup_{k=1}^{m_1}\left(\partial D(-\widetilde{r}_k, \varepsilon) \cup \partial D(\widetilde{r}_k, \varepsilon)\right)\},
				\end{cases}
			\end{equation}
			 with $J_T$ defined in \eqref{def:JT}.
			\item [\rm{(c)}] As $z \to \infty$, we have
			\begin{equation}\label{eq:asyR}
				R(z) = I + \frac{R_1}{z} + \Boh (z^{-2}),
			\end{equation}
			where $R_1$ is independent of $z$.
		\end{itemize}
	\end{rhp}
	
	As $S\to -\infty$, we have the following estimate of $J_R(z)$ in \eqref{def:JR}. For $z \in \Gamma_R \setminus \left\{\bigcup_{k=1}^{m_1}\left(\partial D(-\widetilde{r}_k, \varepsilon) \cup  \partial D(\widetilde{r}_k, \varepsilon)\right) \right\}$, it is readily seen from \eqref{def:JT} and \eqref{def:N} that there exists a positive constant $c$ such that
	\begin{equation}
		J_R(z) = I + \Boh\left(e^{-c (-S)^{3/2}}\right).
	\end{equation}
	For $z \in \partial D(-\widetilde{r}_k, \varepsilon) \cup  \partial D(\widetilde{r}_k, \varepsilon)$, it follows from \eqref{eq:matchhatP-k} and \eqref{rk} that
	\begin{equation}
		J_R(z) = I + \Boh \left((-S)^{-\frac 32}\right).
	\end{equation}
	By a standard argument \cite{Deift1999, Deift1993}, we conclude that
	\begin{equation}\label{def:RR}
		R(z) = I + \Boh \left((-S)^{-\frac 32}\right), \qquad S \to -\infty,
	\end{equation}
	uniformly for $z\in \mathbb{C}\setminus \Gamma_{R}$.

\section{Asymptotic analysis of the RH problem for $\Psi_\mathcal{J}$}\label{sec:II}    
In this case, one has $c_{j\sigma(j)}c_{\sigma(j)j}\neq 1$, $j\in \mathcal{J}$. The transformation \eqref{def:PsitoT} with $\Psi_\mathcal{I}$ replaced by $\Psi_\mathcal{J}$, however, will cause a difficulty in the sense that the $(2,2)$-block of the jump matrix $J_T$ on $(-\widetilde r_k, -\widetilde r_{k-1}) \cup (\widetilde r_{k-1},\widetilde r_k)$ cannot be reduced to the form $\sum_{i \in \bigcup_{l=0}^{k-1}J_l}E_{ii}$. This leads to a completely different asymptotic analysis of the RH problem for $\Psi_\mathcal{J}$, which will be performed in this section.
    
\subsection{First transformation: $\Psi_{\mathcal{J}} \to \what T$}
    In order to normalize the RH problem \ref{rhp:PsiJ} for $\Psi_{\mathcal{J}}$ at infinity, we introduce the  $\tilde \theta$-functions
    \begin{equation}\label{def:tilde-theta}
        \tilde \theta_j(z): =\frac{\what \theta_j(z) + \what \theta_{\sigma (j)}(z)}{2} = \ii (-S)^{\frac 32} \left(\frac 16 z^3 -  r_j z\right), \quad j \in \mathcal{J}
    \end{equation}
    with $\what \theta_j(z)$ defined in \eqref{def:what-theta} and
    \begin{align}\label{def:rj}
		r_j:=\frac{s_{j}+s_{\sigma(j)}}{2S}>0,\quad j \in \mathcal{J},
	\end{align}
where $\sigma$ is defined in \eqref{def:sigma}.
%and $r_j = r_{\sigma(j)} = \frac{\epsilon_j+\epsilon_{\sigma(j)}}{2}+1$ with $\epsilon_j$ defined in Theorem \ref{-infty asympototics results}.

Similar to the strategy used in Section \ref{sec:PsiI-T}, we choose distinct elements from $\sqrt{2r_j}$, arrange them in ascending order and denote them as 
    \begin{equation}
        \what r_1 < \what r_2 < \cdots < \what r_{m_2}, \quad 1\le m_2\le |\mathcal{J}|. 
    \end{equation}
By setting
\begin{align}\label{def:what-J_k}
		\what J_k:=\{j: \sqrt{2r_j}=\widehat r_k\}, \,  1\leq k \leq m_2,  \quad \what J_0 :=\{1,\cdots,n\} \setminus \bigcup_{k=1}^{m_2} \what J_k= \I,
	\end{align} 
we define the constant matrices
\begin{align}\label{def:hatC_k}
		\what C_{k} :=\sum_{j\in \bigcup_{l=k+1}^{m_2}\what J_l}E_{jj}C, \quad 0\le k\le m_2-1, \qquad \what C_{m_2}:=0_{n\times n}
        %\left( I_n - \sum_{i\in \bigcup_{l=1}^{k}\what J_l}E_{ii}\right) C
	\end{align}

The first transformation is defined by 
    \begin{equation}\label{def:hat-T}
        \what T(z) = \Psi_{\J} (z)\what Q(z) e^{-\Theta(z)\otimes \sigma_3},
    \end{equation}
where
	 \begin{align}\label{def:Theta}
		\Theta(z):=\sum_{j\in \J}\left(\tilde \theta_j(z)-t_j z\right)E_{jj},
	 \end{align}
     with 
     \begin{align}
         t_j := \frac{s_j-s_{\sigma(j)}}{2S}, \quad j \in \J,
     \end{align}
	 and 
	\begin{align}\label{def:what-Q}
		\what Q(z):= \begin{cases}
			\begin{pmatrix}
				I_n & \what C_{k-1} \\
				0 & I_n
			\end{pmatrix}, & \quad z \in \what\Omega_{2k-1} \cup \what\Omega_{2m_2+2k-1},\, 1\leq k \leq m_2, \\
			\begin{pmatrix}
				I_n & 0 \\
				\what C_{k-1} & I_n
			\end{pmatrix}, & \quad z \in \what\Omega_{2k} \cup \what\Omega_{2m_2+2k}, \, 1\leq k \leq m_2,\\
			I_{2n}, & \quad \textrm{elsewhere},
		\end{cases} 
	\end{align}
	where the regions $\what \Omega_j$, $j=1, \dots, 4m_2$, are shown in Figure \ref{fig:what-T}. We have that $\what T$ satisfies the following RH problem.
 %    and
 %    \begin{align}\label{def:hatC_k}
	% 	\what C_{k} :=\sum_{j\in \bigcup_{l=k+1}^{m_2}\what J_l}E_{jj}C, \quad 0\le k\le m_2-1, \qquad \what C_{m_2}:=0
 %        %\left( I_n - \sum_{i\in \bigcup_{l=1}^{k}\what J_l}E_{ii}\right) C
	% \end{align}
 %    with 
	% \begin{align}\label{def:what-J_k}
	% 	\what J_k:=\{j: \sqrt{2r_j}=\widehat r_k\}, \,  1\leq k \leq m_2,  \quad \what J_0 :=\{1,\cdots,n\} \setminus \bigcup_{k=1}^{m_2} \what J_k= \I.
	% \end{align}
\begin{figure}[t]
\begin{center}
\tikzset{every picture/.style={line width=0.75pt}} %set default line width to 0.75pt        

\begin{tikzpicture}[x=0.75pt,y=0.75pt,yscale=-1,xscale=0.9]
        \tikzstyle{every node}=[scale=0.9]

%Straight Lines [id:da8489016544296536] 
\draw    (229.17,163.95) -- (287.84,163.95) ;
\draw [shift={(261.1,163.95)}, rotate = 180] [fill={rgb, 255:red, 0; green, 0; blue, 0 }  ][line width=0.08]  [draw opacity=0] (7.14,-3.43) -- (0,0) -- (7.14,3.43) -- (4.74,0) -- cycle    ;
%Straight Lines [id:da6796983681136035] 
\draw    (275.84,163.95) -- (336.24,163.95) ;
%Straight Lines [id:da613185432630844] 
\draw    (336.24,163.95) -- (384.65,163.95) ;
%Straight Lines [id:da7504332571185158] 
\draw    (381.17,163.95) -- (439.84,163.95) ;
\draw [shift={(413.11,163.95)}, rotate = 180] [fill={rgb, 255:red, 0; green, 0; blue, 0 }  ][line width=0.08]  [draw opacity=0] (7.14,-3.43) -- (0,0) -- (7.14,3.43) -- (4.74,0) -- cycle    ;
%Straight Lines [id:da4143333786505953] 
\draw    (82.19,113.17) -- (135.17,163.95) ;
\draw [shift={(110.56,140.36)}, rotate = 223.79] [fill={rgb, 255:red, 0; green, 0; blue, 0 }  ][line width=0.08]  [draw opacity=0] (7.14,-3.43) -- (0,0) -- (7.14,3.43) -- (4.74,0) -- cycle    ;
%Straight Lines [id:da05382685170363366] 
\draw    (234.86,113.17) -- (287.84,163.95) ;
\draw [shift={(263.23,140.36)}, rotate = 223.79] [fill={rgb, 255:red, 0; green, 0; blue, 0 }  ][line width=0.08]  [draw opacity=0] (7.14,-3.43) -- (0,0) -- (7.14,3.43) -- (4.74,0) -- cycle    ;
%Straight Lines [id:da8321812744651109] 
\draw    (381.17,163.95) -- (434.15,214.74) ;
\draw [shift={(409.54,191.15)}, rotate = 223.79] [fill={rgb, 255:red, 0; green, 0; blue, 0 }  ][line width=0.08]  [draw opacity=0] (7.14,-3.43) -- (0,0) -- (7.14,3.43) -- (4.74,0) -- cycle    ;
%Straight Lines [id:da6699101040648437] 
\draw    (439.84,163.95) -- (492.81,214.74) ;
\draw [shift={(468.2,191.15)}, rotate = 223.79] [fill={rgb, 255:red, 0; green, 0; blue, 0 }  ][line width=0.08]  [draw opacity=0] (7.14,-3.43) -- (0,0) -- (7.14,3.43) -- (4.74,0) -- cycle    ;
%Straight Lines [id:da4246408662572354] 
\draw    (434.67,113.17) -- (381.17,163.95) ;
\draw [shift={(410.89,135.74)}, rotate = 136.49] [fill={rgb, 255:red, 0; green, 0; blue, 0 }  ][line width=0.08]  [draw opacity=0] (7.14,-3.43) -- (0,0) -- (7.14,3.43) -- (4.74,0) -- cycle    ;
%Straight Lines [id:da45611561092034236] 
\draw    (493.33,113.17) -- (439.84,163.95) ;
\draw [shift={(469.56,135.74)}, rotate = 136.49] [fill={rgb, 255:red, 0; green, 0; blue, 0 }  ][line width=0.08]  [draw opacity=0] (7.14,-3.43) -- (0,0) -- (7.14,3.43) -- (4.74,0) -- cycle    ;
%Straight Lines [id:da788272488071564] 
\draw    (135.17,163.95) -- (81.68,214.74) ;
\draw [shift={(111.4,186.52)}, rotate = 136.49] [fill={rgb, 255:red, 0; green, 0; blue, 0 }  ][line width=0.08]  [draw opacity=0] (7.14,-3.43) -- (0,0) -- (7.14,3.43) -- (4.74,0) -- cycle    ;
%Straight Lines [id:da23407888485874495] 
\draw    (287.84,163.95) -- (234.34,214.74) ;
\draw [shift={(264.06,186.52)}, rotate = 136.49] [fill={rgb, 255:red, 0; green, 0; blue, 0 }  ][line width=0.08]  [draw opacity=0] (7.14,-3.43) -- (0,0) -- (7.14,3.43) -- (4.74,0) -- cycle    ;
%Straight Lines [id:da3316896512135067] 
\draw    (535.84,163.95) -- (594.51,163.95) ;
\draw [shift={(567.77,163.95)}, rotate = 180] [fill={rgb, 255:red, 0; green, 0; blue, 0 }  ][line width=0.08]  [draw opacity=0] (7.14,-3.43) -- (0,0) -- (7.14,3.43) -- (4.74,0) -- cycle    ;
%Straight Lines [id:da8239056950451331] 
\draw    (594.51,163.95) -- (647.48,214.74) ;
\draw [shift={(622.87,191.15)}, rotate = 223.79] [fill={rgb, 255:red, 0; green, 0; blue, 0 }  ][line width=0.08]  [draw opacity=0] (7.14,-3.43) -- (0,0) -- (7.14,3.43) -- (4.74,0) -- cycle    ;
%Straight Lines [id:da06292186634830166] 
\draw    (648,113.17) -- (594.51,163.95) ;
\draw [shift={(624.23,135.74)}, rotate = 136.49] [fill={rgb, 255:red, 0; green, 0; blue, 0 }  ][line width=0.08]  [draw opacity=0] (7.14,-3.43) -- (0,0) -- (7.14,3.43) -- (4.74,0) -- cycle    ;
%Straight Lines [id:da7100279824601728] 
\draw    (75.78,163.95) -- (134.45,163.95) ;
\draw [shift={(107.71,163.95)}, rotate = 180] [fill={rgb, 255:red, 0; green, 0; blue, 0 }  ][line width=0.08]  [draw opacity=0] (7.14,-3.43) -- (0,0) -- (7.14,3.43) -- (4.74,0) -- cycle    ;
%Straight Lines [id:da02099420361186699] 
\draw    (22.8,113.17) -- (75.78,163.95) ;
\draw [shift={(51.17,140.36)}, rotate = 223.79] [fill={rgb, 255:red, 0; green, 0; blue, 0 }  ][line width=0.08]  [draw opacity=0] (7.14,-3.43) -- (0,0) -- (7.14,3.43) -- (4.74,0) -- cycle    ;
%Straight Lines [id:da0006363063981379424] 
\draw    (75.78,163.95) -- (22.28,214.74) ;
\draw [shift={(52,186.52)}, rotate = 136.49] [fill={rgb, 255:red, 0; green, 0; blue, 0 }  ][line width=0.08]  [draw opacity=0] (7.14,-3.43) -- (0,0) -- (7.14,3.43) -- (4.74,0) -- cycle    ;
%Straight Lines [id:da21306869660108763] 
\draw  [dash pattern={on 0.84pt off 2.51pt}]  (153.17,163.95) -- (211.84,163.95) ;
%Straight Lines [id:da7270609534987692] 
\draw    (175.86,113.17) -- (228.84,163.95) ;
\draw [shift={(204.23,140.36)}, rotate = 223.79] [fill={rgb, 255:red, 0; green, 0; blue, 0 }  ][line width=0.08]  [draw opacity=0] (7.14,-3.43) -- (0,0) -- (7.14,3.43) -- (4.74,0) -- cycle    ;
%Straight Lines [id:da10193383976883319] 
\draw    (228.84,163.95) -- (175.34,214.74) ;
\draw [shift={(205.06,186.52)}, rotate = 136.49] [fill={rgb, 255:red, 0; green, 0; blue, 0 }  ][line width=0.08]  [draw opacity=0] (7.14,-3.43) -- (0,0) -- (7.14,3.43) -- (4.74,0) -- cycle    ;
%Straight Lines [id:da5569065677997533] 
\draw  [dash pattern={on 0.84pt off 2.51pt}]  (456.17,163.95) -- (514.84,163.95) ;
%Straight Lines [id:da16872551302732608] 
\draw    (534.84,163.95) -- (587.81,214.74) ;
\draw [shift={(563.2,191.15)}, rotate = 223.79] [fill={rgb, 255:red, 0; green, 0; blue, 0 }  ][line width=0.08]  [draw opacity=0] (7.14,-3.43) -- (0,0) -- (7.14,3.43) -- (4.74,0) -- cycle    ;
%Straight Lines [id:da44917248327018044] 
\draw    (588.33,113.17) -- (534.84,163.95) ;
\draw [shift={(564.56,135.74)}, rotate = 136.49] [fill={rgb, 255:red, 0; green, 0; blue, 0 }  ][line width=0.08]  [draw opacity=0] (7.14,-3.43) -- (0,0) -- (7.14,3.43) -- (4.74,0) -- cycle    ;
%Straight Lines [id:da8184641288475145] 
\draw    (135.17,163.95) -- (151.84,163.95) ;
%Straight Lines [id:da969878565615637] 
\draw    (212.17,163.95) -- (228.84,163.95) ;
%Straight Lines [id:da8915694131977353] 
\draw    (439.84,163.95) -- (456.51,163.95) ;
%Straight Lines [id:da1971574127165706] 
\draw    (518.17,163.95) -- (534.84,163.95) ;
%Straight Lines [id:da5240366992467177] 
\draw  [dash pattern={on 4.5pt off 4.5pt}]  (389.74,113.17) -- (336.24,163.95) ;
%Straight Lines [id:da529353822773218] 
\draw  [dash pattern={on 4.5pt off 4.5pt}]  (336.24,163.95) -- (282.75,214.74) ;
%Straight Lines [id:da14886372622578725] 
\draw  [dash pattern={on 4.5pt off 4.5pt}]  (283.27,113.17) -- (336.24,163.95) ;
%Straight Lines [id:da8128449260499224] 
\draw  [dash pattern={on 4.5pt off 4.5pt}]  (336.24,163.95) -- (389.22,214.74) ;

\draw  [dash pattern={on 1.2pt off 3pt}]  (140,135) -- (166,135) ;
\draw  [dash pattern={on 1.2pt off 3pt}]  (140,200) -- (166,200) ;

\draw  [dash pattern={on 1.2pt off 3pt}]  (500,135) -- (526,135) ;
\draw  [dash pattern={on 1.2pt off 3pt}]  (500,200) -- (526,200) ;
% Text Node
\draw (219.96,94.95) node [anchor=north west][inner sep=0.75pt]  [font=\scriptsize] [align=left] {$\displaystyle \what{\Gamma }_{2m_2+1}$};
% Text Node
\draw (64.33,167.65) node [anchor=north west][inner sep=0.75pt]  [font=\scriptsize] [align=left] {$\displaystyle -\what{r}_{m_2}$};
% Text Node
\draw (156.96,93.95) node [anchor=north west][inner sep=0.75pt]  [font=\scriptsize] [align=left] {$\displaystyle \what{\Gamma }_{2m_2+3}$};
% Text Node
\draw (64.96,94.95) node [anchor=north west][inner sep=0.75pt]  [font=\scriptsize] [align=left] {$\displaystyle \what{\Gamma }_{4m_2-3}$};
% Text Node
\draw (7.96,93.95) node [anchor=north west][inner sep=0.75pt]  [font=\scriptsize] [align=left] {$\displaystyle \what{\Gamma }_{4m_2-1}$};
% Text Node
\draw (425.96,96.95) node [anchor=north west][inner sep=0.75pt]  [font=\scriptsize] [align=left] {$\displaystyle \what{\Gamma }_{1}$};
% Text Node
\draw (487.96,95.95) node [anchor=north west][inner sep=0.75pt]  [font=\scriptsize] [align=left] {$\displaystyle \what{\Gamma }_{3}$};
% Text Node
\draw (571.96,93.95) node [anchor=north west][inner sep=0.75pt]  [font=\scriptsize] [align=left] {$\displaystyle \what{\Gamma }_{2m_2-3}$};
% Text Node
\draw (634.96,93.95) node [anchor=north west][inner sep=0.75pt]  [font=\scriptsize] [align=left] {$\displaystyle \what{\Gamma }_{2m_2-1}$};
% Text Node
\draw (636.96,211.95) node [anchor=north west][inner sep=0.75pt]  [font=\scriptsize] [align=left] {$\displaystyle \what{\Gamma }_{2m_2}$};
% Text Node
\draw (572.96,210.95) node [anchor=north west][inner sep=0.75pt]  [font=\scriptsize] [align=left] {$\displaystyle \what{\Gamma }_{2m_2-2}$};
% Text Node
\draw (483.96,212.95) node [anchor=north west][inner sep=0.75pt]  [font=\scriptsize] [align=left] {$\displaystyle \what{\Gamma }_{4}$};
% Text Node
\draw (423.96,212.95) node [anchor=north west][inner sep=0.75pt]  [font=\scriptsize] [align=left] {$\displaystyle \what{\Gamma }_{2}$};
% Text Node
\draw (220.96,210.95) node [anchor=north west][inner sep=0.75pt]  [font=\scriptsize] [align=left] {$\displaystyle \what{\Gamma }_{2m_2+2}$};
% Text Node
\draw (156.96,211.95) node [anchor=north west][inner sep=0.75pt]  [font=\scriptsize] [align=left] {$\displaystyle \what{\Gamma }_{2m_2+4}$};
% Text Node
\draw (64.96,211.95) node [anchor=north west][inner sep=0.75pt]  [font=\scriptsize] [align=left] {$\displaystyle \what{\Gamma }_{4m_2-2}$};
% Text Node
\draw (7.96,213.95) node [anchor=north west][inner sep=0.75pt]  [font=\scriptsize] [align=left] {$\displaystyle \what{\Gamma }_{4m_2}$};
% Text Node
\draw (125.78,166.95) node [anchor=north west][inner sep=0.75pt]  [font=\scriptsize] [align=left] {$\displaystyle -\what{r}_{m_2-1}$};
% Text Node
\draw (216.84,166.95) node [anchor=north west][inner sep=0.75pt]  [font=\scriptsize] [align=left] {$\displaystyle -\what{r}_{2}$};
% Text Node
\draw (280.17,167.95) node [anchor=north west][inner sep=0.75pt]  [font=\scriptsize] [align=left] {$\displaystyle -\what{r}_{1}$};
% Text Node
\draw (378.17,167.95) node [anchor=north west][inner sep=0.75pt]  [font=\scriptsize] [align=left] {$\displaystyle \what{r}_{1}$};
% Text Node
\draw (436.17,167.95) node [anchor=north west][inner sep=0.75pt]  [font=\scriptsize] [align=left] {$\displaystyle \what{r}_{2}$};
% Text Node
\draw (511,165.95) node [anchor=north west][inner sep=0.75pt]  [font=\scriptsize] [align=left] {$\displaystyle \what{r}_{m_2-1}$};
% Text Node
\draw (585,167.65) node [anchor=north west][inner sep=0.75pt]  [font=\scriptsize] [align=left] {$\displaystyle \what{r}_{m_2}$};
% Text Node
\draw (333,168.07) node [anchor=north west][inner sep=0.75pt]  [font=\scriptsize]  {$0$};
% Text Node
\draw (373,132.73) node [anchor=north west][inner sep=0.75pt]  [font=\scriptsize]  {$\what{\Omega }_{1}$};
% Text Node
\draw (430,131.4) node [anchor=north west][inner sep=0.75pt]  [font=\scriptsize]  {$\what{\Omega }_{3}$};
% Text Node
\draw (375.33,183.4) node [anchor=north west][inner sep=0.75pt]  [font=\scriptsize]  {$\what{\Omega }_{2}$};
% Text Node
\draw (428,184.07) node [anchor=north west][inner sep=0.75pt]  [font=\scriptsize]  {$\what{\Omega }_{4}$};
% Text Node
\draw (581.33,182.73) node [anchor=north west][inner sep=0.75pt]  [font=\scriptsize]  {$\what{\Omega }_{2m_2}$};
% Text Node
\draw (570,134.07) node [anchor=north west][inner sep=0.75pt]  [font=\scriptsize]  {$\what{\Omega }_{2m_2-1}$};
% Text Node
\draw (267,128.73) node [anchor=north west][inner sep=0.75pt]  [font=\scriptsize]  {$\what{\Omega }_{2m_2+1}$};
% Text Node
\draw (263,184.73) node [anchor=north west][inner sep=0.75pt]  [font=\scriptsize]  {$\what{\Omega }_{2m_2+2}$};
% Text Node
\draw (217,128.07) node [anchor=north west][inner sep=0.75pt]  [font=\scriptsize]  {$\what{\Omega }_{2m_2+3}$};
% Text Node
\draw (210,184.73) node [anchor=north west][inner sep=0.75pt]  [font=\scriptsize]  {$\what{\Omega }_{2m_2+4}$};
% Text Node
\draw (64,131.4) node [anchor=north west][inner sep=0.75pt]  [font=\scriptsize]  {$\what{\Omega }_{4m_2-1}$};
% Text Node
\draw (67,185.4) node [anchor=north west][inner sep=0.75pt]  [font=\scriptsize]  {$\what{\Omega }_{4m_2}$};

\end{tikzpicture}
\caption{Regions $\what \Omega_i$ and the jump contours $\what \Gamma_i$, $i=1,\ldots,4m_2$, in the RH problem for $\what T$.
The dashed lines represent the contour $ \gamma_+ \cup  \gamma_-$.}
\label{fig:what-T}
\end{center}
\end{figure}

    % with $\Theta(z)$ and $\what Q(z)$ defined in \eqref{def:Theta} and \eqref{def:what-Q}.
    % It is easily seen that from RH problem \ref{rhp:PsiJ} for $\Psi_{\J}$ that $\what T$ satisfies the following RH problem.
    \begin{rhp}\label{rhp:what-T}
		\hfill
		\begin{enumerate}
			
			\item[\rm (a)] $\what T(z)$ is defined and analytic in $\mathbb{C} \setminus \what\Gamma_T$, where 
			\begin{equation}\label{def:what-gammaT}
				\what\Gamma_T:=\cup_{i=1}^{4m_2} \what\Gamma_i \cup[-\what r_{m_2}, \what r_{m_2}];
			\end{equation}
			see Figure \ref{fig:what-T} for an illustration.
			\item[\rm (b)] For $z \in \what \Gamma_T$, we have 
            \begin{align}
                \what T_{+}(z)=\what T_{-}(z) J_{\what T}(z),
            \end{align}
			where
			\begin{equation}\label{def:what-JT}
				J_{\what T}(z):= \begin{cases}
					\begin{pmatrix}
						I_n & \sum_{j\in \what J_{k}} e^{2\tilde \theta_j(z)}E_{jj} C \\
						0 & I_n
					\end{pmatrix}, & \quad z \in \what\Gamma_{2k-1} \cup \what\Gamma_{2m_2+2k-1},  \\
					\begin{pmatrix}
						I_n & 0 \\
						-\sum_{j\in \what J_{k}} e^{-2\tilde \theta_j(z)}E_{jj} C & I_n
					\end{pmatrix}, & \quad z \in \what\Gamma_{2k} \cup \what\Gamma_{2m_2+2k},  \\
					\begin{pmatrix}
						I_n & e^{\Theta (z)}\what C_{k-1}e^{\Theta (z)}\\
						-e^{-\Theta (z)}\what C_{k-1}e^{-\Theta (z)} & I_n-\what C_{k-1}^2
					\end{pmatrix}, & \quad z \in (-\what r_k, -\what r_{k-1}) \cup  (\what r_{k-1}, \what r_k),
				\end{cases} 
			\end{equation}
			with $1\leq k \leq m_2$ and $\what r_0:=0$.
			\item[\rm (c)]As $z \to \infty$ with $z \in \mathbb{C} \setminus  \what\Gamma_T$, we have
			\begin{align}\label{eq:asyhatT}
				\what T(z)&= I_{2n}+ \frac{\Psi_{\mathcal J,1}}{\sqrt{-S}z}+\Boh\left(z^{-2}\right),
			\end{align}
			where $\Psi_{\mathcal J,1}$ is defined in \eqref{J1}. 
		\end{enumerate}
	\end{rhp}
\begin{proof}
   Items (a) and (c) can be confirmed directly by using \eqref{def:hat-T} and RH problem \ref{rhp:PsiJ} for $\Psi_{\J}$. Consequently, our focus is on verifying the jump condition \eqref{def:what-JT}. For $z \in \what\Gamma_{2k-1} \cup \what\Gamma_{2m_2+2k-1}$ and $z \in \what\Gamma_{2k} \cup \what\Gamma_{2m_2+2k}$, the calculations closely resemble those presented in the proof of \eqref{check-JT} and we omit the details here. For $z \in (-\what r_k, -\what r_{k-1})\cup  (\what r_{k-1}, \what r_k)$, we have 
    \begin{align*}
    &\what T_-(z)^{-1} \what T_+(z) = e^{\Theta(z)\otimes \sigma_3}\what Q_-(z)^{-1}\what Q_+(z) e^{-\Theta(z)\otimes \sigma_3}\nonumber\\
    &=e^{\Theta(z)\otimes \sigma_3}\begin{pmatrix}
        I_n & \what C_{k-1}\\
        -\what C_{k-1} & I_n - C_{k-1}^2
    \end{pmatrix} e^{-\Theta(z)\otimes \sigma_3}
     = \begin{pmatrix}
						I_n & e^{\Theta (z)}\what C_{k-1}e^{\Theta (z)}\\
						-e^{-\Theta (z)}\what C_{k-1}e^{-\Theta (z)} & I_n-\what C_{k-1}^2
					\end{pmatrix},
    \end{align*}
    as required. In the last equality, we have made use of the fact that $\what  C_{k-1}^2$ is a diagonal matrix.
\end{proof}

\subsection{Second transformation: $\what T \to \what S$}
Based on the definition of $\Theta(z)$ in \eqref{def:Theta}, it is readily seen that $J_{\what T} (z)$ in \eqref{def:what-JT} tends to $I_{2n}$ exponentially fast as $S \to - \infty$ for $z \in \what \Gamma_T \setminus [-\what r_{m_2}, \what r_{m_2}]$. Moreover, the terms $e^{\Theta (z)}\what C_{k-1}e^{\Theta (z)}$ and $e^{-\Theta (z)}\what C_{k-1}e^{-\Theta (z)}$ in \eqref{def:what-JT} are highly oscillatory for large negative $S$.

    The second transformation then involves the so-called lens opening around the intervals $(-\what r_k, -\what r_{k-1}) \cup  (\what r_{k-1}, \what r_k)$, $1 \le k \le m_2$. The idea is to remove the highly oscillatory terms of $J_{\what T}$ with the cost of creating extra jumps that tend to the identity matrix on some new contours. To proceed, we observe that
    \begin{align}\label{open}
        \begin{pmatrix}
			I_n & e^{\Theta (z)}\what C_{k-1}e^{\Theta (z)}\\
			-e^{-\Theta (z)}\what C_{k-1}e^{-\Theta (z)} & I_n-\what C_{k-1}^2
		\end{pmatrix} = \what J_{k}^{(-)}(z)\begin{pmatrix}
						(I_n - \what C^2_{k-1})^{-1} & 0\\
						0& I_n-\what C_{k-1}^2
					\end{pmatrix}\what J_{k}^{(+)}(z),
                    \end{align}
        where
        \begin{equation}\label{def:Jk-}
            \what J_{k}^{(-)}(z):=\begin{pmatrix}
						I_n & e^{\Theta (z)}\what C_{k-1}(I_n - \what C^2_{k-1})^{-1}e^{\Theta (z)}\\
						0 & I_n
					\end{pmatrix}
        \end{equation}
        and
        \begin{equation}\label{def:Jk+}
            \what J_{k}^{(+)}(z) := \begin{pmatrix}
						I_n & 0\\
						-e^{-\Theta (z)}(I_n - \what C^2_{k-1})^{-1}\what C_{k-1}e^{-\Theta (z)} & I_n
					\end{pmatrix}.
        \end{equation}
        
Let $\Delta_{k,\pm}$ be the lenses on the $\pm$-side of $(-\what r_k, -\what r_{k-1})$ and $(\what r_{k-1}, \what r_k)$, as shown in Figure \ref{fig:S}. The second transformation is defined by
    \begin{equation}\label{def:hat-S}
        \what S(z) = \what T(z)\begin{cases}
             \what J_{k}^{(+)}(z)^{-1}, \quad & z \in \Delta_{k,+},\\ 
             \what J_{k}^{(-)}(z), \quad & z \in \Delta_{k,-},\\
             I_{2n}, \quad & \textrm{elsewhere.}
        \end{cases}
    \end{equation}
It is then easily seen from  RH problem \ref{rhp:what-T}  for $\what T$ and  \eqref{open} that $\what S$ satisfies the following RH problem.
    
    %see Figure \ref{fig:S} for an illustration.
    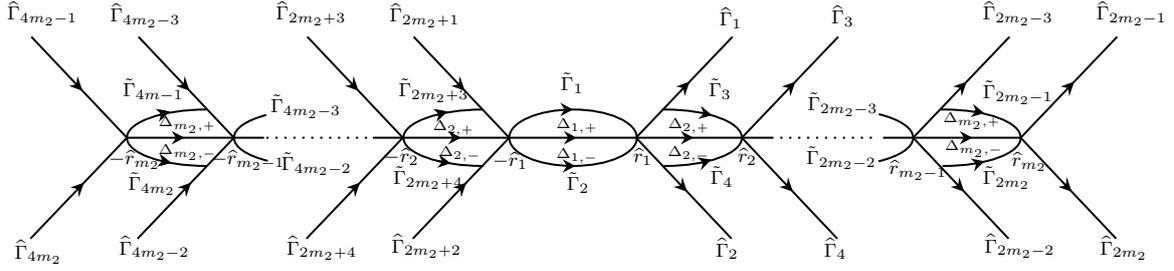
\begin{figure}[t]
\begin{center}
\tikzset{every picture/.style={line width=0.75pt}} %set default line width to 0.75pt        

\begin{tikzpicture}[x=0.75pt,y=0.75pt,yscale=-1,xscale=0.9]
        \tikzstyle{every node}=[scale=0.9]

%Straight Lines [id:da4312890434100992] 
\draw    (252.17,163.95) -- (310.84,163.95) ;
\draw [shift={(284.1,163.95)}, rotate = 180] [fill={rgb, 255:red, 0; green, 0; blue, 0 }  ][line width=0.08]  [draw opacity=0] (7.14,-3.43) -- (0,0) -- (7.14,3.43) -- (4.74,0) -- cycle    ;
%Straight Lines [id:da7300681999452283] 
\draw    (310.84,163.95) -- (346.24,163.95) ;
%Straight Lines [id:da5761025745583493] 
\draw    (347.24,163.95) -- (381.65,163.95) ;
\draw [shift={(348.24,163.95)}, rotate = 180] [fill={rgb, 255:red, 0; green, 0; blue, 0 }  ][line width=0.08]  [draw opacity=0] (7.14,-3.43) -- (0,0) -- (7.14,3.43) -- (4.74,0) -- cycle    ;
%Straight Lines [id:da6608521868343554] 
\draw    (381.17,163.95) -- (439.84,163.95) ;
\draw [shift={(413.11,163.95)}, rotate = 180] [fill={rgb, 255:red, 0; green, 0; blue, 0 }  ][line width=0.08]  [draw opacity=0] (7.14,-3.43) -- (0,0) -- (7.14,3.43) -- (4.74,0) -- cycle    ;
%Straight Lines [id:da8367295312644816] 
\draw    (105.19,113.17) -- (158.17,163.95) ;
\draw [shift={(133.56,140.36)}, rotate = 223.79] [fill={rgb, 255:red, 0; green, 0; blue, 0 }  ][line width=0.08]  [draw opacity=0] (7.14,-3.43) -- (0,0) -- (7.14,3.43) -- (4.74,0) -- cycle    ;
%Straight Lines [id:da13912687562335346] 
\draw    (257.86,113.17) -- (310.84,163.95) ;
\draw [shift={(286.23,140.36)}, rotate = 223.79] [fill={rgb, 255:red, 0; green, 0; blue, 0 }  ][line width=0.08]  [draw opacity=0] (7.14,-3.43) -- (0,0) -- (7.14,3.43) -- (4.74,0) -- cycle    ;
%Straight Lines [id:da2945922152552972] 
\draw    (381.17,163.95) -- (434.15,214.74) ;
\draw [shift={(409.54,191.15)}, rotate = 223.79] [fill={rgb, 255:red, 0; green, 0; blue, 0 }  ][line width=0.08]  [draw opacity=0] (7.14,-3.43) -- (0,0) -- (7.14,3.43) -- (4.74,0) -- cycle    ;
%Straight Lines [id:da43741024153477326] 
\draw    (439.84,163.95) -- (492.81,214.74) ;
\draw [shift={(468.2,191.15)}, rotate = 223.79] [fill={rgb, 255:red, 0; green, 0; blue, 0 }  ][line width=0.08]  [draw opacity=0] (7.14,-3.43) -- (0,0) -- (7.14,3.43) -- (4.74,0) -- cycle    ;
%Straight Lines [id:da1776491781533711] 
\draw    (434.67,113.17) -- (381.17,163.95) ;
\draw [shift={(410.89,135.74)}, rotate = 136.49] [fill={rgb, 255:red, 0; green, 0; blue, 0 }  ][line width=0.08]  [draw opacity=0] (7.14,-3.43) -- (0,0) -- (7.14,3.43) -- (4.74,0) -- cycle    ;
%Straight Lines [id:da9308807641646986] 
\draw    (493.33,113.17) -- (439.84,163.95) ;
\draw [shift={(469.56,135.74)}, rotate = 136.49] [fill={rgb, 255:red, 0; green, 0; blue, 0 }  ][line width=0.08]  [draw opacity=0] (7.14,-3.43) -- (0,0) -- (7.14,3.43) -- (4.74,0) -- cycle    ;
%Straight Lines [id:da06215618431729153] 
\draw    (158.17,163.95) -- (104.68,214.74) ;
\draw [shift={(134.4,186.52)}, rotate = 136.49] [fill={rgb, 255:red, 0; green, 0; blue, 0 }  ][line width=0.08]  [draw opacity=0] (7.14,-3.43) -- (0,0) -- (7.14,3.43) -- (4.74,0) -- cycle    ;
%Straight Lines [id:da5460409230939466] 
\draw    (310.84,163.95) -- (257.34,214.74) ;
\draw [shift={(287.06,186.52)}, rotate = 136.49] [fill={rgb, 255:red, 0; green, 0; blue, 0 }  ][line width=0.08]  [draw opacity=0] (7.14,-3.43) -- (0,0) -- (7.14,3.43) -- (4.74,0) -- cycle    ;
%Shape: Ellipse [id:dp6317521442108273] 
\draw   (310.84,163.95) .. controls (310.84,156.11) and (326.69,149.76) .. (346.24,149.76) .. controls (365.8,149.76) and (381.65,156.11) .. (381.65,163.95) .. controls (381.65,171.79) and (365.8,178.15) .. (346.24,178.15) .. controls (326.69,178.15) and (310.84,171.79) .. (310.84,163.95) -- cycle ;
%Shape: Arc [id:dp38977514637648536] 
\draw  [draw opacity=0] (296.25,178.15) .. controls (296.22,178.15) and (296.2,178.15) .. (296.17,178.15) .. controls (271.87,178.15) and (252.17,171.79) .. (252.17,163.95) .. controls (252.17,156.11) and (271.87,149.76) .. (296.17,149.76) .. controls (296.2,149.76) and (296.22,149.76) .. (296.24,149.76) -- (296.17,163.95) -- cycle ; \draw   (296.25,178.15) .. controls (296.22,178.15) and (296.2,178.15) .. (296.17,178.15) .. controls (271.87,178.15) and (252.17,171.79) .. (252.17,163.95) .. controls (252.17,156.11) and (271.87,149.76) .. (296.17,149.76) .. controls (296.2,149.76) and (296.22,149.76) .. (296.24,149.76) ;  
%Shape: Arc [id:dp43819725503037243] 
\draw  [draw opacity=0] (395.54,149.77) .. controls (396.16,149.76) and (396.78,149.76) .. (397.41,149.76) .. controls (420.84,149.76) and (439.84,156.11) .. (439.84,163.95) .. controls (439.84,171.79) and (420.84,178.15) .. (397.41,178.15) .. controls (397.09,178.15) and (396.78,178.15) .. (396.47,178.15) -- (397.41,163.95) -- cycle ; \draw   (395.54,149.77) .. controls (396.16,149.76) and (396.78,149.76) .. (397.41,149.76) .. controls (420.84,149.76) and (439.84,156.11) .. (439.84,163.95) .. controls (439.84,171.79) and (420.84,178.15) .. (397.41,178.15) .. controls (397.09,178.15) and (396.78,178.15) .. (396.47,178.15) ;  
%Straight Lines [id:da3387319543977434] 
\draw    (346.24,149.76) ;
\draw [shift={(346.24,149.76)}, rotate = 180] [fill={rgb, 255:red, 0; green, 0; blue, 0 }  ][line width=0.08]  [draw opacity=0] (7.14,-3.43) -- (0,0) -- (7.14,3.43) -- (4.74,0) -- cycle    ;
%\draw [shift={(NaN,NaN)}, rotate = 180] [fill={rgb, 255:red, 0; green, 0; blue, 0 }  ][line width=0.08]  [draw opacity=0] (7.14,-3.43) -- (0,0) -- (7.14,3.43) -- (4.74,0) -- cycle    ;
%Straight Lines [id:da4419125471227907] 
\draw    (347.55,178.15) ;
\draw [shift={(347.55,178.15)}, rotate = 180] [fill={rgb, 255:red, 0; green, 0; blue, 0 }  ][line width=0.08]  [draw opacity=0] (7.14,-3.43) -- (0,0) -- (7.14,3.43) -- (4.74,0) -- cycle    ;
%Straight Lines [id:da5475279419137324] 
\draw    (271.39,152.21) -- (273.57,151.6) ;
\draw [shift={(276.46,150.79)}, rotate = 164.36] [fill={rgb, 255:red, 0; green, 0; blue, 0 }  ][line width=0.08]  [draw opacity=0] (5.36,-2.57) -- (0,0) -- (5.36,2.57) -- cycle    ;
\draw [shift={(276.43,150.8)}, rotate = 164.36] [fill={rgb, 255:red, 0; green, 0; blue, 0 }  ][line width=0.08]  [draw opacity=0] (7.14,-3.43) -- (0,0) -- (7.14,3.43) -- (4.74,0) -- cycle    ;
%Straight Lines [id:da5179557118749883] 
\draw    (270.66,175.63) -- (275.44,176.32) ;
\draw [shift={(278.4,176.75)}, rotate = 188.2] [fill={rgb, 255:red, 0; green, 0; blue, 0 }  ][line width=0.08]  [draw opacity=0] (7.14,-3.43) -- (0,0) -- (7.14,3.43) -- (4.74,0) -- cycle    ;
%Straight Lines [id:da8473177529790542] 
\draw    (417.77,151.5) -- (420.4,152.02) ;
\draw [shift={(423.34,152.61)}, rotate = 191.33] [fill={rgb, 255:red, 0; green, 0; blue, 0 }  ][line width=0.08]  [draw opacity=0] (7.14,-3.43) -- (0,0) -- (7.14,3.43) -- (4.74,0) -- cycle    ;
%Straight Lines [id:da7048423743522321] 
\draw    (417.77,176.58) -- (419.99,175.85) ;
\draw [shift={(422.84,174.92)}, rotate = 161.93] [fill={rgb, 255:red, 0; green, 0; blue, 0 }  ][line width=0.08]  [draw opacity=0] (7.14,-3.43) -- (0,0) -- (7.14,3.43) -- (4.74,0) -- cycle    ;
%Straight Lines [id:da9076887308198965] 
\draw    (535.84,163.95) -- (594.51,163.95) ;
\draw [shift={(567.77,163.95)}, rotate = 180] [fill={rgb, 255:red, 0; green, 0; blue, 0 }  ][line width=0.08]  [draw opacity=0] (7.14,-3.43) -- (0,0) -- (7.14,3.43) -- (4.74,0) -- cycle    ;
%Straight Lines [id:da40022367872158604] 
\draw    (594.51,163.95) -- (647.48,214.74) ;
\draw [shift={(622.87,191.15)}, rotate = 223.79] [fill={rgb, 255:red, 0; green, 0; blue, 0 }  ][line width=0.08]  [draw opacity=0] (7.14,-3.43) -- (0,0) -- (7.14,3.43) -- (4.74,0) -- cycle    ;
%Straight Lines [id:da1743803403579539] 
\draw    (648,113.17) -- (594.51,163.95) ;
\draw [shift={(624.23,135.74)}, rotate = 136.49] [fill={rgb, 255:red, 0; green, 0; blue, 0 }  ][line width=0.08]  [draw opacity=0] (7.14,-3.43) -- (0,0) -- (7.14,3.43) -- (4.74,0) -- cycle    ;
%Shape: Arc [id:dp44096405620002155] 
\draw  [draw opacity=0] (550.2,149.77) .. controls (550.82,149.76) and (551.45,149.76) .. (552.07,149.76) .. controls (575.51,149.76) and (594.51,156.11) .. (594.51,163.95) .. controls (594.51,171.79) and (575.51,178.15) .. (552.07,178.15) .. controls (551.76,178.15) and (551.44,178.15) .. (551.12,178.15) -- (552.07,163.95) -- cycle ; \draw   (550.2,149.77) .. controls (550.82,149.76) and (551.45,149.76) .. (552.07,149.76) .. controls (575.51,149.76) and (594.51,156.11) .. (594.51,163.95) .. controls (594.51,171.79) and (575.51,178.15) .. (552.07,178.15) .. controls (551.76,178.15) and (551.44,178.15) .. (551.12,178.15) ;  
%Straight Lines [id:da685443462159252] 
\draw    (98.78,163.95) -- (157.45,163.95) ;
\draw [shift={(130.71,163.95)}, rotate = 180] [fill={rgb, 255:red, 0; green, 0; blue, 0 }  ][line width=0.08]  [draw opacity=0] (7.14,-3.43) -- (0,0) -- (7.14,3.43) -- (4.74,0) -- cycle    ;
%Straight Lines [id:da5494890747110374] 
\draw    (45.8,113.17) -- (98.78,163.95) ;
\draw [shift={(74.17,140.36)}, rotate = 223.79] [fill={rgb, 255:red, 0; green, 0; blue, 0 }  ][line width=0.08]  [draw opacity=0] (7.14,-3.43) -- (0,0) -- (7.14,3.43) -- (4.74,0) -- cycle    ;
%Straight Lines [id:da40863851472686585] 
\draw    (98.78,163.95) -- (45.28,214.74) ;
\draw [shift={(75,186.52)}, rotate = 136.49] [fill={rgb, 255:red, 0; green, 0; blue, 0 }  ][line width=0.08]  [draw opacity=0] (7.14,-3.43) -- (0,0) -- (7.14,3.43) -- (4.74,0) -- cycle    ;
%Shape: Arc [id:dp49866179428623314] 
\draw  [draw opacity=0] (142.86,178.15) .. controls (142.83,178.15) and (142.81,178.15) .. (142.78,178.15) .. controls (118.48,178.15) and (98.78,171.79) .. (98.78,163.95) .. controls (98.78,156.11) and (118.48,149.76) .. (142.78,149.76) .. controls (142.8,149.76) and (142.83,149.76) .. (142.85,149.76) -- (142.78,163.95) -- cycle ; \draw   (142.86,178.15) .. controls (142.83,178.15) and (142.81,178.15) .. (142.78,178.15) .. controls (118.48,178.15) and (98.78,171.79) .. (98.78,163.95) .. controls (98.78,156.11) and (118.48,149.76) .. (142.78,149.76) .. controls (142.8,149.76) and (142.83,149.76) .. (142.85,149.76) ;  
%Straight Lines [id:da8786520251446612] 
\draw  [dash pattern={on 0.84pt off 2.51pt}]  (176.17,163.95) -- (234.84,163.95) ;
%Straight Lines [id:da8972805682566867] 
\draw    (198.86,113.17) -- (251.84,163.95) ;
\draw [shift={(227.23,140.36)}, rotate = 223.79] [fill={rgb, 255:red, 0; green, 0; blue, 0 }  ][line width=0.08]  [draw opacity=0] (7.14,-3.43) -- (0,0) -- (7.14,3.43) -- (4.74,0) -- cycle    ;
%Straight Lines [id:da1687850490340984] 
\draw    (251.84,163.95) -- (198.34,214.74) ;
\draw [shift={(228.06,186.52)}, rotate = 136.49] [fill={rgb, 255:red, 0; green, 0; blue, 0 }  ][line width=0.08]  [draw opacity=0] (7.14,-3.43) -- (0,0) -- (7.14,3.43) -- (4.74,0) -- cycle    ;
%Shape: Arc [id:dp22451878907890033] 
\draw  [draw opacity=0] (176.91,175.58) .. controls (165.58,173.01) and (158.17,168.76) .. (158.17,163.95) .. controls (158.17,159.25) and (165.25,155.09) .. (176.15,152.5) -- (202.17,163.95) -- cycle ; \draw   (176.91,175.58) .. controls (165.58,173.01) and (158.17,168.76) .. (158.17,163.95) .. controls (158.17,159.25) and (165.25,155.09) .. (176.15,152.5) ;  
%Straight Lines [id:da5752035003341397] 
\draw  [dash pattern={on 0.84pt off 2.51pt}]  (456.17,163.95) -- (514.84,163.95) ;
%Straight Lines [id:da060953879377177] 
\draw    (534.84,163.95) -- (587.81,214.74) ;
\draw [shift={(563.2,191.15)}, rotate = 223.79] [fill={rgb, 255:red, 0; green, 0; blue, 0 }  ][line width=0.08]  [draw opacity=0] (7.14,-3.43) -- (0,0) -- (7.14,3.43) -- (4.74,0) -- cycle    ;
%Straight Lines [id:da057160701257438995] 
\draw    (588.33,113.17) -- (534.84,163.95) ;
\draw [shift={(564.56,135.74)}, rotate = 136.49] [fill={rgb, 255:red, 0; green, 0; blue, 0 }  ][line width=0.08]  [draw opacity=0] (7.14,-3.43) -- (0,0) -- (7.14,3.43) -- (4.74,0) -- cycle    ;
%Shape: Arc [id:dp8229549048226312] 
\draw  [draw opacity=0] (515.79,152.1) .. controls (527.27,154.65) and (534.84,159) .. (534.84,163.95) .. controls (534.84,168.92) and (527.21,173.29) .. (515.67,175.83) -- (492.41,163.95) -- cycle ; \draw   (515.79,152.1) .. controls (527.27,154.65) and (534.84,159) .. (534.84,163.95) .. controls (534.84,168.92) and (527.21,173.29) .. (515.67,175.83) ;  
%Straight Lines [id:da3629511417396708] 
\draw    (158.17,163.95) -- (174.84,163.95) ;
%Straight Lines [id:da7525425374044717] 
\draw    (235.17,163.95) -- (251.84,163.95) ;
%Straight Lines [id:da39573607669518673] 
\draw    (439.84,163.95) -- (456.51,163.95) ;
%Straight Lines [id:da8528247797059575] 
\draw    (518.17,163.95) -- (534.84,163.95) ;
%Straight Lines [id:da5803766122887217] 
\draw    (568.77,150.5) -- (571.4,151.02) ;
\draw [shift={(574.34,151.61)}, rotate = 191.33] [fill={rgb, 255:red, 0; green, 0; blue, 0 }  ][line width=0.08]  [draw opacity=0] (7.14,-3.43) -- (0,0) -- (7.14,3.43) -- (4.74,0) -- cycle    ;
%Straight Lines [id:da6447176238335814] 
\draw    (117.77,175.5) -- (120.4,176.02) ;
\draw [shift={(123.34,176.61)}, rotate = 191.33] [fill={rgb, 255:red, 0; green, 0; blue, 0 }  ][line width=0.08]  [draw opacity=0] (5.36,-2.57) -- (0,0) -- (5.36,2.57) -- cycle    ;
\draw [shift={(123.11,176.57)}, rotate = 191.33] [fill={rgb, 255:red, 0; green, 0; blue, 0 }  ][line width=0.08]  [draw opacity=0] (7.14,-3.43) -- (0,0) -- (7.14,3.43) -- (4.74,0) -- cycle    ;
%Straight Lines [id:da29780491417417265] 
\draw    (117.39,152.21) -- (119.57,151.6) ;
\draw [shift={(122.46,150.79)}, rotate = 164.36] [fill={rgb, 255:red, 0; green, 0; blue, 0 }  ][line width=0.08]  [draw opacity=0] (5.36,-2.57) -- (0,0) -- (5.36,2.57) -- cycle    ;
\draw [shift={(122.43,150.8)}, rotate = 164.36] [fill={rgb, 255:red, 0; green, 0; blue, 0 }  ][line width=0.08]  [draw opacity=0] (7.14,-3.43) -- (0,0) -- (7.14,3.43) -- (4.74,0) -- cycle    ;
%Straight Lines [id:da39398200674916917] 
\draw    (569.39,177.21) -- (571.57,176.6) ;
\draw [shift={(574.46,175.79)}, rotate = 164.36] [fill={rgb, 255:red, 0; green, 0; blue, 0 }  ][line width=0.08]  [draw opacity=0] (7.14,-3.43) -- (0,0) -- (7.14,3.43) -- (4.74,0) -- cycle    ;

% Text Node
\draw (242.96,94.95) node [anchor=north west][inner sep=0.75pt]  [font=\scriptsize] [align=left] {$\what{\Gamma }_{2m_2+1}$};
% Text Node
\draw (87.33,169) node [anchor=north west][inner sep=0.75pt]  [font=\scriptsize] [align=left] {$-\what{r}_{m_2}$};
% Text Node
\draw (179.96,93.95) node [anchor=north west][inner sep=0.75pt]  [font=\scriptsize] [align=left] {$ \what{\Gamma }_{2m_2+3}$};
% Text Node
\draw (87.96,94.95) node [anchor=north west][inner sep=0.75pt]  [font=\scriptsize] [align=left] {$ \what{\Gamma }_{4m_2-3}$};
% Text Node
\draw (30.96,93.95) node [anchor=north west][inner sep=0.75pt]  [font=\scriptsize] [align=left] {$ \what{\Gamma }_{4m_2-1}$};
% Text Node
\draw (425.96,96.95) node [anchor=north west][inner sep=0.75pt]  [font=\scriptsize] [align=left] {$\what{\Gamma }_{1}$};
% Text Node
\draw (487.96,95.95) node [anchor=north west][inner sep=0.75pt]  [font=\scriptsize] [align=left] {$ \what{\Gamma }_{3}$};
% Text Node
\draw (571.96,93.95) node [anchor=north west][inner sep=0.75pt]  [font=\scriptsize] [align=left] {$ \what{\Gamma }_{2m_2-3}$};
% Text Node
\draw (634.96,93.95) node [anchor=north west][inner sep=0.75pt]  [font=\scriptsize] [align=left] {$ \what{\Gamma }_{2m_2-1}$};
% Text Node
\draw (636.96,211.95) node [anchor=north west][inner sep=0.75pt]  [font=\scriptsize] [align=left] {$ \what{\Gamma }_{2m_2}$};
% Text Node
\draw (572.96,210.95) node [anchor=north west][inner sep=0.75pt]  [font=\scriptsize] [align=left] {$ \what{\Gamma }_{2m_2-2}$};
% Text Node
\draw (483.96,212.95) node [anchor=north west][inner sep=0.75pt]  [font=\scriptsize] [align=left] {$ \what{\Gamma }_{4}$};
% Text Node
\draw (423.96,212.95) node [anchor=north west][inner sep=0.75pt]  [font=\scriptsize] [align=left] {$ \what{\Gamma }_{2}$};
% Text Node
\draw (243.96,210.95) node [anchor=north west][inner sep=0.75pt]  [font=\scriptsize] [align=left] {$ \what{\Gamma }_{2m_2+2}$};
% Text Node
\draw (185.96,211.95) node [anchor=north west][inner sep=0.75pt]  [font=\scriptsize] [align=left] {$ \what{\Gamma }_{2m_2+4}$};
% Text Node
\draw (92.96,211.95) node [anchor=north west][inner sep=0.75pt]  [font=\scriptsize] [align=left] {$ \what{\Gamma }_{4m_2-2}$};
% Text Node
\draw (34.96,213.95) node [anchor=north west][inner sep=0.75pt]  [font=\scriptsize] [align=left] {$ \what{\Gamma }_{4m_2}$};
% Text Node
\draw (144,169) node [anchor=north west][inner sep=0.75pt]  [font=\scriptsize] [align=left] {$ -\what{r}_{m_2-1}$};
% Text Node
\draw (239.84,168) node [anchor=north west][inner sep=0.75pt]  [font=\scriptsize] [align=left] {$-\what{r}_{2}$};
% Text Node
\draw (300.17,170) node [anchor=north west][inner sep=0.75pt]  [font=\scriptsize] [align=left] {$ -\what{r}_{1}$};
% Text Node
\draw (378.17,167.95) node [anchor=north west][inner sep=0.75pt]  [font=\scriptsize] [align=left] {$ \what{r}_{1}$};
% Text Node
\draw (436.17,167.95) node [anchor=north west][inner sep=0.75pt]  [font=\scriptsize] [align=left] {$ \what{r}_{2}$};
% Text Node
\draw (520,174) node [anchor=north west][inner sep=0.75pt]  [font=\scriptsize] [align=left] {$ \what{r}_{m_2-1}$};
% Text Node
\draw (588,170) node [anchor=north west][inner sep=0.75pt]  [font=\scriptsize] [align=left] {$ \what{r}_{m_2}$};
% Text Node
\draw (95,133) node [anchor=north west][inner sep=0.75pt]  [font=\scriptsize] [align=left] {$ \tilde{\Gamma }_{4m-1}$};
% Text Node
\draw (177,141) node [anchor=north west][inner sep=0.75pt]  [font=\scriptsize] [align=left] {$ \tilde{\Gamma }_{4m_2-3}$};
% Text Node
\draw (248,133) node [anchor=north west][inner sep=0.75pt]  [font=\scriptsize] [align=left] {$ \tilde{\Gamma }_{2m_2+3}$};
% Text Node
\draw (337.96,130) node [anchor=north west][inner sep=0.75pt]  [font=\scriptsize] [align=left] {$ \tilde{\Gamma }_{1}$};
% Text Node
\draw (420,133) node [anchor=north west][inner sep=0.75pt]  [font=\scriptsize] [align=left] {$ \tilde{\Gamma }_{3}$};
% Text Node
\draw (471,140) node [anchor=north west][inner sep=0.75pt]  [font=\scriptsize] [align=left] {$\ \tilde{\Gamma }_{2m_2-3}$};
% Text Node
\draw (572,133.95) node [anchor=north west][inner sep=0.75pt]  [font=\scriptsize] [align=left] {$ \tilde{\Gamma }_{2m_2-1}$};
% Text Node
\draw (572,177) node [anchor=north west][inner sep=0.75pt]  [font=\scriptsize] [align=left] {$ \tilde{\Gamma }_{2m_2}$};
% Text Node
\draw (474,166.95) node [anchor=north west][inner sep=0.75pt]  [font=\scriptsize] [align=left] {$ \tilde{\Gamma }_{2m_2-2}$};
% Text Node
\draw (420.96,176.95) node [anchor=north west][inner sep=0.75pt]  [font=\scriptsize] [align=left] {$ \tilde{\Gamma }_{4}$};
% Text Node
\draw (341.24,180.15) node [anchor=north west][inner sep=0.75pt]  [font=\scriptsize] [align=left] {$\tilde{\Gamma }_{2}$};
% Text Node
\draw (245,177) node [anchor=north west][inner sep=0.75pt]  [font=\scriptsize] [align=left] {$ \tilde{\Gamma }_{2m_2+4}$};
% Text Node
\draw (183,170) node [anchor=north west][inner sep=0.75pt]  [font=\scriptsize] [align=left] {$ \tilde{\Gamma }_{4m_2-2}$};
% Text Node
\draw (98,180) node [anchor=north west][inner sep=0.75pt]  [font=\scriptsize] [align=left] {$ \tilde{\Gamma }_{4m_2}$};
% Text Node
\draw (335.46,152.79) node [anchor=north west][inner sep=0.75pt]  [font=\tiny] [align=left] {$ \Delta _{1,+}$};
% Text Node
\draw (397.54,152.77) node [anchor=north west][inner sep=0.75pt]  [font=\tiny] [align=left] {$\Delta _{2,+}$};
% Text Node
\draw (550,150) node [anchor=north west][inner sep=0.75pt]  [font=\tiny] [align=left] {$ \Delta _{m_2,+}$};
% Text Node
\draw (267,153.79) node [anchor=north west][inner sep=0.75pt]  [font=\tiny] [align=left] {$\Delta _{2,+}$};
% Text Node
\draw (115,152) node [anchor=north west][inner sep=0.75pt]  [font=\tiny] [align=left] {$ \Delta _{m_2,+}$};
% Text Node
\draw (115,165.77) node [anchor=north west][inner sep=0.75pt]  [font=\tiny] [align=left] {$ \Delta _{m_2,-}$};
% Text Node
\draw (551,164) node [anchor=north west][inner sep=0.75pt]  [font=\tiny] [align=left] {$ \Delta _{m_2,-}$};
% Text Node
\draw (335.46,167.79) node [anchor=north west][inner sep=0.75pt]  [font=\tiny] [align=left] {$ \Delta _{1,-}$};
% Text Node
\draw (397.41,166.95) node [anchor=north west][inner sep=0.75pt]  [font=\tiny] [align=left] {$ \Delta _{2,-}$};
% Text Node
\draw (270.5,166.95) node [anchor=north west][inner sep=0.75pt]  [font=\tiny] [align=left] {$ \Delta _{2,-}$};
% Text Node
%\draw (378,147.4) node [anchor=north west][inner sep=0.75pt]  [font=\tiny]  {$\bar \Gamma _{1}$};
% Text Node
%\draw (383,180.4) node [anchor=north west][inner sep=0.75pt]  [font=\tiny]  {$\bar\Gamma _{2}$};
% Text Node
%\draw (524,143.4) node [anchor=north west][inner sep=0.75pt]  [font=\tiny]  {$\bar\Gamma _{2m-3}$};
% Text Node
%\draw (536.84,167.35) node [anchor=north west][inner sep=0.75pt]  [font=\tiny]  {$\bar\Gamma _{2m}$};
% Text Node
%\draw (147,147.4) node [anchor=north west][inner sep=0.75pt]  [font=\tiny]  {$\bar\Gamma _{4m-3}$};
% Text Node
%\draw (144.86,181.55) node [anchor=north west][inner sep=0.75pt]  [font=\tiny]  {$\bar\Gamma _{4m-2}$};
% Text Node
%\draw (299,145.4) node [anchor=north west][inner sep=0.75pt]  [font=\tiny]  {$\bar\Gamma _{2m+1}$};
% Text Node
%\draw (295,186.4) node [anchor=north west][inner sep=0.75pt]  [font=\tiny]  {$\bar\Gamma _{2m+2}$};

\end{tikzpicture}
\caption{Regions $\Delta_{k,\pm}$, $k=1, \ldots,m_2$, and the jump contour of the RH problem for $\what S$. $\what \Gamma_i$, $\tilde \Gamma_i$, $i=1,\ldots,4m_2$, are the solid lines that are distinct from $[-\what r_{m_2},\what r_{m_2}]$. }
\label{fig:S}
\end{center}
\end{figure}

 \begin{rhp}\label{rhp:what-S}
		\hfill
		\begin{enumerate}
			\item[\rm (a)] $\what S(z)$ is defined and analytic in $\mathbb{C} \setminus \what\Gamma_S$, where $\what\Gamma_S$ is shown in Figure \ref{fig:S}. 
            %and we define
   %          \begin{align}
   %              &\bar{\Gamma}_{2k-1}=\partial \Delta_{k+1,+} \cap \what \Gamma_{2k-1}, \quad  \bar{\Gamma}_{2m_2+2k-1} = \partial \Delta_{k+1,+} \cap \what \Gamma_{2m_2+2k-1}
   %              ,\nonumber\\
   %              &\bar \Gamma_{2k} = \partial \Delta_{k+1,-}\cap\what \Gamma_{2k}
   %             , \quad \bar \Gamma_{2m_2+2k} = \partial \Delta_{k+1,-}\cap\what \Gamma_{2m_2+2k},\quad k=1,\cdots,m_2-1.
   %          \end{align}
			\item[\rm (b)]  For $z \in \what \Gamma_S$, we have 
            \begin{align}
                \what S_{+}(z)=\what S_{-}(z) J_{\what S}(z),
            \end{align}
			where
            {\footnotesize{
			\begin{equation}\label{def:what-JS}
				J_{\what S}(z):= \begin{cases}
					\begin{pmatrix}
						I_n & \sum_{j\in \what J_{k}} e^{2\tilde \theta_j(z)}E_{jj} C \\
						0 & I_n
					\end{pmatrix}, & \quad z \in \{\what\Gamma_{2k-1} \cup \what\Gamma_{2m_2+2k-1}\}\setminus\{\bar\Gamma_{2k-1} \cup \bar\Gamma_{2m_2+2k-1}\},  \\
					\begin{pmatrix}
						I_n & 0 \\
						-\sum_{j\in \what J_{k}} e^{-2\tilde \theta_j(z)}E_{jj} C & I_n
					\end{pmatrix}, & \quad z \in \{\what\Gamma_{2k} \cup \what\Gamma_{2m_2+2k}\}\setminus\{\bar\Gamma_{2k} \cup \bar\Gamma_{2m_2+2k}\},  \\
					\begin{pmatrix}
					(I_n - \what C_{k-1}^2)^{-1} & 0\\
					0 & I_n - \what C_{k-1}^2
				\end{pmatrix}, & \quad z \in (-\what r_k, -\what r_{k-1}) \cup  (\what r_{k-1}, \what r_k),\\
                \begin{pmatrix}
						I_n & 0\\
						-e^{-\Theta (z)}(I_n - \what C^2_{k-1})^{-1}\what C_{k-1}e^{-\Theta (z)} & I_n
					\end{pmatrix}, & \quad z \in \tilde\Gamma_{2k-1} \cup \tilde\Gamma_{2m_2+2k-1},  \\
                    \begin{pmatrix}
						I_n & e^{\Theta (z)}\what C_{k-1}(I_n - \what C^2_{k-1})^{-1}e^{\Theta (z)}\\
						0 & I_n
					\end{pmatrix},& \quad z \in \tilde\Gamma_{2k} \cup \tilde\Gamma_{2m_2+2k},  \\
                    \begin{pmatrix}
						I_n & \sum_{j\in \what J_{k}} e^{2\tilde \theta_j(z)}E_{jj} C\\
						-e^{-\Theta (z)}(I_n - \what C^2_{k})^{-1}\what C_{k}e^{-\Theta (z)} & I_n
					\end{pmatrix}, & \quad z \in \bar\Gamma_{2k-1} \cup \bar\Gamma_{2m_2+2k-1},  \\
                    \begin{pmatrix}
						I_n & e^{\Theta (z)}\what C_{k}(I_n - \what C^2_{k})^{-1}e^{\Theta (z)}\\
						-\sum_{j\in \what J_{k}} e^{-2\tilde \theta_j(z)}E_{jj} C & I_n
					\end{pmatrix},& \quad z \in \bar\Gamma_{2k} \cup \bar\Gamma_{2m_2+2k},
				\end{cases} 
			\end{equation}}}
			with $1\leq k \leq m_2$ and $\what r_0:=0$, and where we define
            \begin{align*}
                &\bar{\Gamma}_{2k-1}=\partial \Delta_{k+1,+} \cap \what \Gamma_{2k-1}, \quad  \bar{\Gamma}_{2m_2+2k-1} = \partial \Delta_{k+1,+} \cap \what \Gamma_{2m_2+2k-1}
                ,\nonumber\\
                &\bar \Gamma_{2k} = \partial \Delta_{k+1,-}\cap\what \Gamma_{2k}
               , \quad \bar \Gamma_{2m_2+2k} = \partial \Delta_{k+1,-}\cap\what \Gamma_{2m_2+2k},\quad k=1,\cdots,m_2-1.
            \end{align*}
			\item[\rm (c)]As $z \to \infty$ with $z \in \mathbb{C} \setminus  \what\Gamma_S$, we have
			\begin{align}\label{eq:asyS}
				\what S(z)&= I_{2n}+ \frac{\Psi_{\mathcal J,1}}{\sqrt{-S}z}+\Boh\left(z^{-2}\right),
			\end{align}
			where $\Psi_{\mathcal J,1}$ is defined in \eqref{J1}. 		
		\end{enumerate}
	\end{rhp}
    
    \subsection{Global parametrix}

    As $S \to -\infty$, it is now readily seen that all the jump matrices of $\what S$ tend to the identity matrix exponentially fast except for those along $[-\what r_{m_2}, \what r_{m_2}]$ and we are led to consider the following global parametrix.
    \begin{rhp}\label{rhp:hatN}
		\hfill
		\begin{enumerate}
			\item[\rm (a)] $\what N(z)$ is defined and analytic in $\mathbb{C} \setminus [-\what r_{m_2}, \what r_{m_2}]$.
			\item[\rm (b)]  For $z\in (-\what r_{m_2}, \what r_{m_2})$, we have
            \begin{align}
                \what N_{+}(x)=\what N_{-}(x) J_{\what N}(x),
            \end{align}
			where 
			\begin{align}\label{def:jumphatN}
				J_{\what N}= \begin{pmatrix}
					(I_n - \what C_{k-1}^2)^{-1} & 0\\
					0 & I_n - \what C_{k-1}^2
				\end{pmatrix}, & \quad z \in  (-\what r_{k}, -\what r_{k-1}) \cup  (\what r_{k-1}, \what r_{k}),
			\end{align}
			with $1\leq k \leq m_2$.
			\item[\rm (c)] As $z \to \infty$ with $z \in \mathbb{C} \setminus  [-\what r_{m_2}, \what r_{m_2}]$, we have
			\begin{align}\label{eq:asyhatN}
				\what N(z)= I_{2n}+ \frac{\what N_1}{z}+\Boh\left(z^{-2}\right),
			\end{align}
			where $\what N_1$ is independent of $z$. 
		\end{enumerate}
	\end{rhp}

\begin{lemma}
     A solution of RH problem \ref{rhp:hatN} is given by
    \begin{align}\label{def:hat-N}
         \what N(z) := \begin{pmatrix}
             \sum_{j \in \bigcup_{k=1}^{m_2}\what J_k}\left(\frac{z+\what r_k}{z-\what r_k}\right)^{-\nu_j} E_{jj}+\sum_{i \in \I}E_{ii} & 0\\
             0 & \sum_{j \in \bigcup_{k=1}^{m_2}\what J_k}\left(\frac{z+\what r_k}{z-\what r_k}\right)^{\nu_j} E_{jj}+\sum_{i \in \I}E_{ii}
         \end{pmatrix},
        % \what N(z) := \begin{pmatrix}
        %     \displaystyle \sum_{i=1}^n \left(\frac{z+\sqrt{2r_i}}{z-\sqrt{2r_i}}\right)^{-\nu_j} E_{ii} & 0\\
        %     0 & \displaystyle \sum_{i=1}^n \left(\frac{z+\sqrt{2r_i}}{z-\sqrt{2r_i}}\right)^{\nu_j} E_{ii}
        % \end{pmatrix},
    \end{align}
    where 
    \begin{equation}\label{def:nui}
        \nu_j := -\frac{1}{2 \pi \ii} \ln \left(1-\mu_j \mu_{\sigma(j)}\right)
    \end{equation}
    with $\mu_j=c_{j\sigma(j)}$ being the diagonal entry of the matrix $\Lambda$ in Assumption \ref{assump} and we take the branch cut of $\left(\frac{z+\what r_k}{z-\what r_k}\right)^{\pm\nu_j}$ along $(-\what r_k, \what r_k)$.
\end{lemma}
   
    \begin{proof}
        By re-expressing the (1,1)-block of the jump matrix in \eqref{def:jumphatN} in the  form
\begin{align}
        (I_n - \what C_{k-1}^2)^{-1}=\sum_{j\in \bigcup_{l=k}^{m_2}\what J_l}(1-\mu_j \mu_{\sigma(j)})^{-1}E_{jj} + \sum_{j \in \bigcup_{l=0}^{k-1}\what J_l}E_{jj},
\end{align}
and applying the jump relations
\begin{align}
   \left(\frac{z+\what r_k}{z-\what r_k}\right)^{-\nu_j}_+ &= e^{-2 \nu_j \pi \ii} \left(\frac{z+\what r_k}{z-\what r_k}\right)^{-\nu_j}_-, &&\quad z \in (-\what r_k, \what r_k),\nonumber\\
   \left(\frac{z+\what r_k}{z-\what r_k}\right)^{-\nu_j}_+ &=  \left(\frac{z+\what r_k}{z-\what r_k}\right)^{-\nu_j}_-, &&\quad z \in\C \setminus [-\what r_k, \what r_k],
\end{align}
we have, for $z \in  (-\what r_{k}, -\what r_{k-1}) \cup  (\what r_{k-1}, \what r_{k})$, the (1,1)-block of $\what N(z)$ satisfies 
\begin{align}
    &\sum_{j \in \bigcup_{k=1}^{m_2}\what J_k}\left(\frac{z+\what r_k}{z-\what r_k}\right)^{-\nu_j}_+ E_{jj}+\sum_{i \in \I}E_{ii}\nonumber\\
    &=\sum_{j \in \bigcup_{l=k}^{m_2}\what J_l}e^{-2 \nu_j \pi \ii}\left(\frac{z+\what r_k}{z-\what r_k}\right)^{-\nu_j}_-E_{jj}+\sum_{j \in \bigcup_{l=1}^{k-1}\what J_k}\left(\frac{z+\what r_k}{z-\what r_k}\right)_-^{-\nu_j} E_{jj}+\sum_{i \in \I}E_{ii}\nonumber\\
    &=\left[\sum_{j \in \bigcup_{k=1}^{m_2}\what J_k}\left(\frac{z+\what r_k}{z-\what r_k}\right)^{-\nu_j}_- E_{jj}+\sum_{i \in \I}E_{ii}\right]\left[\sum_{j\in \bigcup_{l=k}^{m_2}\what J_l}(1-\mu_j \mu_{\sigma(j)})^{-1}E_{jj} + \sum_{j \in \bigcup_{l=0}^{k-1}\what J_l}E_{jj}\right]\nonumber\\
     &=\left[\sum_{j \in \bigcup_{k=1}^{m_2}\what J_k}\left(\frac{z+\what r_k}{z-\what r_k}\right)^{-\nu_j}_- E_{jj}+\sum_{i \in \I}E_{ii}\right] \cdot (I_n - \what C_{k-1}^2)^{-1},
\end{align}
as required. The proof of the jump condition for the $(2,2)$-block of $\what N(z)$ follows similarly, we therefore omit the details here.

As $z \to \infty$, a direct verification confirms that $\what N(z)$ defined in \eqref{def:hat-N} satisfies the asymptotic behavior prescribed by \eqref{eq:asyhatN}.
    \end{proof}

   \subsection{Local parametrix near $ \pm \what r_k$}
    Since the convergence of $J_{\what S}$ to the identity matrix is not uniform near $z = \pm \what r_k$, $1 \le k \le m_2$, we have to build local parametrices  $\what P^{(\pm k)}$  around these points. From the inherent symmetry of the jump matrix $J_{\what S}$ in \eqref{def:what-JS}, it follows that, similar to \eqref{rk}, 
    \begin{equation}\label{relation:hatPk-hatP-k}
		\what P^{(k)}(z) =\what \sigma_1 \what P^{(-k)}(-z)\what \sigma_1
	\end{equation}
    with $\what \sigma_1$ given in \eqref{def:hatsigma1}.
%Here, $\what P^{(k)}$ and $\what P^{(-k)}$ denote the local parametrices constructed in the small disks $D(\what r_k, \varepsilon)$ centered at $z=\what r_k$ and $D(-\what r_k, \varepsilon)$ centered at $z=-\what r_k$, respectively. 
Hence, it suffices to construct the local parametrix in the neighborhood of $z=\what r_k$, which reads as follows.
   \begin{rhp}\label{rhp:hat-Pk}
		\hfill
		\begin{enumerate}
			\item[\rm (a)] $\what P^{(k)}(z)$ is defined and analytic in $D(\what{r}_k, \varepsilon)\setminus \what\Gamma_S$, where $\what\Gamma_S$ is shown in Figure \ref{fig:S}.
			
			\item[\rm (b)] For $z \in D(\what{r}_k, \varepsilon) \cap \what\Gamma_S$, we have 
			\begin{equation}
				\what P^{(k)}_+(z)=\what P^{(k)}_-(z)	J_{\what P^{(k)}}(z)
			\end{equation}
			where 
			\begin{equation}\label{def:hat-JP-k}
				J_{\what P^{(k)}}(z):= \begin{cases}
				\begin{pmatrix}
					(I_n -\what  C_{k-1}^2)^{-1} & 0\\
					0 & I_n - \what C_{k-1}^2
				\end{pmatrix}, & \quad z \in D(\what{r}_k, \varepsilon) \cap (\what r_{k-1}, \what r_k),\\
				\begin{pmatrix}
					(I_n - \what C_{k}^2)^{-1} & 0\\
					0 & I_n - \what C_{k}^2
				\end{pmatrix}, & \quad z \in D(\what{r}_k, \varepsilon) \cap (\what r_{k}, \what r_{k+1}),\\
                \begin{pmatrix}
						I_n & 0\\
						-e^{-\Theta (z)}(I_n - \what C^2_{k-1})^{-1}\what C_{k-1}e^{-\Theta (z)} & I_n
					\end{pmatrix}, & \quad z \in D(\what{r}_k, \varepsilon) \cap \tilde\Gamma_{2k-1},  \\
                    \begin{pmatrix}
						I_n & e^{\Theta (z)}\what C_{k-1}(I_n - \what C^2_{k-1})^{-1}e^{\Theta (z)}\\
						0 & I_n
					\end{pmatrix},& \quad z \in D(\what{r}_k, \varepsilon) \cap \tilde\Gamma_{2k},  \\
                    \begin{pmatrix}
						I_n & \sum_{j\in \what J_{k}} e^{2\tilde \theta_j(z)}E_{jj} C\\
						-e^{-\Theta (z)}(I_n - \what C^2_{k})^{-1}\what C_{k}e^{-\Theta (z)} & I_n
					\end{pmatrix}, & \quad z \in D(\what{r}_k, \varepsilon) \cap \bar\Gamma_{2k-1},  \\
                    \begin{pmatrix}
						I_n & e^{\Theta (z)}\what C_{k}(I_n - \what C^2_{k})^{-1}e^{\Theta (z)}\\
						-\sum_{j\in \what J_{k}} e^{-2\tilde \theta_j(z)}E_{jj} C & I_n
					\end{pmatrix},& \quad z \in D(\what{r}_k, \varepsilon) \cap \bar\Gamma_{2k},
				\end{cases} 
			\end{equation}
			with $1\leq k \leq m_2$ and $\what r_{m_2+1}:=\infty$. 
            %We observe that the matrix $\what  P^{(m_2)}$ is analytic on the interval $(\what r_{m_2}, +\infty)$. Therefore, the definition of $\what r_{m_2+1}$ is justified and is introduced solely for the sake of formal consistency.
			
			\item[\rm (c)]As $S \to -\infty$,  we have the matching condition
			\begin{equation}\label{eq:matchhatPk}
				\what  P^{(k)}(z)=\left( I+ \Boh\left((-S)^{-\frac{3}{4}}\right) \right)  \what N(z),\quad z \in \partial D(\what {r}_k, \varepsilon),
			\end{equation}
			where $\what  N(z)$ is given in \eqref{def:hat-N}.
		\end{enumerate}
	\end{rhp}
	
	We can construct $\what  P^{(k)}(z)$ explicitlyby  using the parabolic cylinder parametrix $\Phi^{(\PC)}$ introduced in Appendix \ref{sec:PC}. To proceed, we define 
	\begin{align}\label{def:hat-fk}
	\what f^{(k)}(z) : = -2(-S)^{-\frac{3}{4}}\left(-\tilde \theta_j(z) + \tilde \theta_j(\what r_k)\right)^{\frac 12}, \quad j \in \what J_k,
	\end{align}
	where $\tilde \theta_j(z)$ is defined in \eqref{def:tilde-theta} and the branch cut of $(\cdot)^{\frac 12}$ is chosen such that $\arg (z - \what r_k) \in (-\pi, \pi)$. It is easily seen that
	\begin{equation}\label{def:hat-fk-rk}
	\what f^{(k)}(z)  = e^{\frac{3 \pi \ii}{4}} (2 \what r_k)^{\frac 12} (z - \what r_k) + \Boh \left((z - \what r_k)^2\right), \quad \textrm{as $ z \to \what r_k$}.
	\end{equation}
	
Let $\Phi^{(\PC)}(z; \nu_j)$ be the parabolic cylinder parametrix with $\nu_j$ defined in \eqref{def:nui}, we now define
	\begin{align}\label{def:hat-Pk}
	\what  P^{(k)}(z) = \what E^{(k)}(z) \what \Phi^{(k)}(z) \what U^{(k)}(z) e^{\left(-\sum_{j \in \bigcup_{l=k}^{m_2} \what J_l} \tilde \theta_j(z)E_{jj} \right)\otimes \sigma_3},
	\end{align}
	where
	{\footnotesize
	\begin{align}\label{def:hatE}
	&\quad\what E^{(k)}(z) := \what N(z) \nonumber\\
	&\times\begin{cases}
	\begin{pmatrix}
	\sum_{j \in \bigcup_{l=0}^{k-1} \what J_l} E_{jj} + \sum_{j \in \bigcup_{l=k+1}^{m_2} \what J_l} E_{jj}& -\sum_{j \in \what J_k}h_{1j} (-S)^{-\frac{3\nu_j}{4}}\what f^{(k)}(z)^{- \nu_j} e^{\tilde \theta_j(\what r_k)}\mu_{\sigma(j)}^{-1} E_{j\sigma(j)} \\
	\sum_{j \in \what J_k} (-S)^{\frac{3\nu_j}{4}}\what f^{(k)}(z)^{\nu_j} e^{-\tilde \theta_j(\what r_k)}E_{jj}  & \sum_{j \in \bigcup_{l=0}^{k-1} \what J_l} E_{jj} + \sum_{j \in \bigcup_{l=k+1}^{m_2} \what J_l} E_{jj}
	\end{pmatrix}, &  \Im z >0,\\
	\begin{pmatrix}
	\sum_{j \in \bigcup_{l=0}^{k-1} \what J_l} E_{jj} + \sum_{j \in \bigcup_{l=k+1}^{m_2} \what J_l} E_{jj}(I_n - C^2)^{-1} & -\sum_{j \in \what J_k}h_{1j}(-S)^{-\frac{3\nu_j}{4}} \what f^{(k)}(z)^{- \nu_j} e^{\tilde \theta_j(\what r_k)}\mu_{\sigma(j)}^{-1} E_{j\sigma(j)}\\
	\sum_{j \in \what J_k} (-S)^{\frac{3\nu_j}{4}}\what f^{(k)}(z)^{\nu_j} e^{-\tilde \theta_j(\what r_k)}E_{jj}  & \sum_{j \in \bigcup_{l=0}^{k-1} \what J_l} E_{jj} + \sum_{j \in \bigcup_{l=k+1}^{m_2} \what J_l} E_{jj}(I_n - C^2)
	\end{pmatrix}, & \Im z <0,
	\end{cases}
	\end{align}}
\begin{align}
\what \Phi^{(k)}(z)&:= \begin{pmatrix}
\sum_{j \in \bigcup_{l=0}^{k-1} \what J_l} E_{jj} & 0\\
0 & \sum_{j \in \bigcup_{l=0}^{k-1} \what J_l} E_{jj} 
\end{pmatrix}+ \begin{pmatrix}
\sum_{j \in \bigcup_{l=k+1}^{m_2} \what J_l} E_{jj} e^{\tilde \theta_j(z)} &0\\
0 &  \sum_{j \in \bigcup_{l=k+1}^{m_2} \what J_l} E_{jj} e^{-\tilde \theta_j(z)} 
\end{pmatrix}\nonumber\\
&\quad+
 \begin{pmatrix}
\sum_{j \in \what J_k} E_{jj} \Phi^{(\PC)}_{11}\left((-S)^{\frac 34} \what f^{(k)}(z);\nu_j\right)  &  \sum_{j \in \what J_k} E_{jj} \Phi^{(\PC)}_{12}\left((-S)^{\frac 34} \what f^{(k)}(z);\nu_j\right)\\
 \sum_{j \in \what J_k} E_{jj} \Phi^{(\PC)}_{21}\left((-S)^{\frac 34} \what f^{(k)}(z);\nu_j\right) & \sum_{j \in \what J_k} E_{jj} \Phi^{(\PC)}_{22}\left((-S)^{\frac 34} \what f^{(k)}(z);\nu_j\right)
\end{pmatrix},
\end{align}
% with $\mu_j=c_{j\sigma(j)}$, 
and
\begin{align}\label{hat-Uk}
\what U^{(k)}(z) &:= \begin{pmatrix}
\sum_{j \in \bigcup_{l=0}^{k-1} \what J_l} E_{jj} & 0\\
0 & \sum_{j \in \bigcup_{l=0}^{k-1} \what J_l} E_{jj} 
\end{pmatrix}+ \begin{pmatrix}
0& \sum_{j \in\what J_k} E_{jj} \\
-  \sum_{j \in\what J_k} E_{jj} h_{1j}^{-1} C & 0
\end{pmatrix}\nonumber\\
&\quad+\begin{cases}
 \begin{pmatrix}
\sum_{j \in \bigcup_{l=k+1}^{m_2}\what J_l} E_{jj} & 0\\
0 & \sum_{j \in \bigcup_{l=k+1}^{m_2} \what J_l} E_{jj}
\end{pmatrix}, \quad & z \in \Delta_{k+1,+} \cup \Delta_{k,+},\\
 \begin{pmatrix}
\sum_{j \in \bigcup_{l=k+1}^{m_2}\what J_l} E_{jj} & 0\\
-\sum_{j \in \bigcup_{l=k+1}^{m_2}\what J_l} E_{jj}  C(I_n-C^2)^{-1} & \sum_{j\in \bigcup_{l=k+1}^{m_2}\what J_l} E_{jj}
\end{pmatrix}, \quad & z \in \what \Omega_{2k-1} \setminus \Delta_{k,+},\\
 \begin{pmatrix}
\sum_{j \in \bigcup_{l=k+1}^{m_2}\what J_l} E_{jj}(I_n-C^2) & 0\\
0 & \sum_{j \in \bigcup_{l=k+1}^{m_2}\what J_l} E_{jj}(I_n-C^2)^{-1}
\end{pmatrix}, \quad & z \in \Delta_{k+1,-} \cup \Delta_{k,-},\\
 \begin{pmatrix}
\sum_{j \in \bigcup_{l=k+1}^{m_2}\what J_l} E_{jj}(I_n-C^2) & -\sum_{j\in \bigcup_{l=k+1}^{m_2}\what J_l} E_{jj}  C\\
0 & \sum_{j\in \bigcup_{l=k+1}^{m_2}\what J_l} E_{jj}(I_n-C^2)^{-1}
\end{pmatrix}, \quad & z \in \what \Omega_{2k} \setminus \Delta_{k,-},
\end{cases}
\end{align}
where 
\begin{equation}\label{def:h1i}
h_{1j} := \frac{2 \pi}{\Gamma(-\nu_j)} e^{\pi \ii \nu_j}
\end{equation}
and $\nu_j$ is defined in \eqref{def:nui}.

For later use, we note that 
\begin{align}\label{hat-Uk-1}
    &\quad\what U^{(k)}(z)^{-1} \nonumber\\
    &= \begin{pmatrix}
\sum_{j \in \bigcup_{l=0}^{k-1} \what J_l} E_{jj} & 0\\
0 & \sum_{j \in \bigcup_{l=0}^{k-1} \what J_l} E_{jj} 
\end{pmatrix}+ \begin{pmatrix}
0&  -  \sum_{j \in\what J_k} E_{j\sigma(j)} h_{1j} \mu_{\sigma(j)}^{-1}\\
\sum_{j \in\what J_k} E_{jj} & 0
\end{pmatrix}\nonumber\\
&\quad+\begin{cases}
 \begin{pmatrix}
\sum_{j \in \bigcup_{l=k+1}^{m_2}\what J_l} E_{jj} & 0\\
0 & \sum_{j \in \bigcup_{l=k+1}^{m_2} \what J_l} E_{jj}
\end{pmatrix}, \quad & z \in \Delta_{k+1,+} \cup \Delta_{k,+},\\
 \begin{pmatrix}
\sum_{j \in \bigcup_{l=k+1}^{m_2}\what J_l} E_{jj} & 0\\
\sum_{j \in \bigcup_{l=k+1}^{m_2}\what J_l} E_{jj}  C(I_n-C^2)^{-1} & \sum_{j\in \bigcup_{l=k+1}^{m_2}\what J_l} E_{jj}
\end{pmatrix}, \quad & z \in \what \Omega_{2k-1} \setminus \Delta_{k,+},\\
 \begin{pmatrix}
\sum_{j \in \bigcup_{l=k+1}^{m_2}\what J_l} E_{jj}(I_n-C^2)^{-1} & 0\\
0 & \sum_{j \in \bigcup_{l=k+1}^{m_2}\what J_l} E_{jj}(I_n-C^2)
\end{pmatrix}, \quad & z \in \Delta_{k+1,-} \cup \Delta_{k,-},\\
 \begin{pmatrix}
\sum_{j \in \bigcup_{l=k+1}^{m_2}\what J_l} E_{jj}(I_n-C^2)^{-1} & \sum_{j\in \bigcup_{l=k+1}^{m_2}\what J_l} E_{jj}  C\\
0 & \sum_{j\in \bigcup_{l=k+1}^{m_2}\what J_l} E_{jj}(I_n-C^2)
\end{pmatrix}, \quad & z \in \what \Omega_{2k} \setminus \Delta_{k,-}.
\end{cases}
\end{align}
\begin{remark}
To explain the case that $\mu_{\sigma(j)}=0$ in \eqref{def:hatE} and \eqref{hat-Uk-1}, we see from \eqref{def:h1i} that
    \begin{align}\label{hij=0}
        h_{1j}\mu_{\sigma(j)}^{-1} = \frac{2 \pi e^{\pi \ii \nu_j}}{\Gamma(-\nu_j) \mu_{\sigma(j)}}.
        \end{align}
        Substituting the definition of $\nu_j$ in \eqref{def:nui} into the above equation, we have  
        %and applying the Taylor expansion, we obtain after simplification
    \begin{align}\label{hij=0-1}
    h_{1j}\mu_{\sigma(j)}^{-1}\rvert_{\mu_{\sigma(j)}=0}=\lim_{\mu_{\sigma(j)}\to 0}\frac{-2 \pi \nu_je^{\pi \ii \nu_j}}{\Gamma(1-\nu_j) \mu_{\sigma(j)}}=\lim_{\mu_{\sigma(j)}\to 0} \frac{-\ii \ln(1-\mu_j\mu_{\sigma(j)})}{\Gamma(1)\mu_{\sigma(j)}}=\ii \mu_j.
    \end{align}
    Therefore, the functions $\what E^{(k)}(z)$ in \eqref{def:hatE} and $\what U^{(k)}(z)^{-1}$ in \eqref{hat-Uk-1} are well-defined for $\mu_{\sigma(j)}=0$.
\end{remark}

\begin{proposition}
The function $\what P^{(k)}(z)$ defined in \eqref{def:hat-Pk} solves RH problem \ref{rhp:hat-Pk}.
\end{proposition}

\begin{proof}
We first show the prefactor $\what E^{(k)}(z)$ is analytic in $D(\what r_k, \varepsilon)$. From the definition of $\what E^{(k)}(z)$ in \eqref{def:hatE}, the only possible jump is on $(\what r_k-\varepsilon, \what r_k+\varepsilon)$. It follows from \eqref{def:jumphatN} and \eqref{def:hat-fk}   that, for $x \in (\what r_k-\varepsilon, \what r_k)$, 
\begin{align}
    & \what E^{(k)}_-(x)^{-1}\what E^{(k)}_+(x) \nonumber\\
    &= \begin{pmatrix}
	\sum_{j \in \bigcup_{l=0}^{k-1} \what J_l} E_{jj} + \sum_{j \in \bigcup_{l=k+1}^{m_2} \what J_l} E_{jj}(I_n - C^2) & \sum_{j \in \what J_k} (-S)^{-\frac{3\nu_j}{4}}\what f^{(k)}_-(x)^{-\nu_j} e^{\tilde \theta_j(\what r_k)}E_{jj}\\
	-\sum_{j \in \what J_k}h_{1j}^{-1} (-S)^{\frac{3\nu_j}{4}}\what f^{(k)}_-(x)^{ \nu_j} e^{-\tilde \theta_j(\what r_k)}E_{jj} C  & \sum_{j \in \bigcup_{l=0}^{k-1} \what J_l} E_{jj} + \sum_{j \in \bigcup_{l=k+1}^{m_2} \what J_l} E_{jj}(I_n - C^2)^{-1}
	\end{pmatrix}\nonumber\\
    &\quad \times \what N_-(x)^{-1} \what N_+(x)\nonumber\\
    &\quad \times\begin{pmatrix}
	\sum_{j \in \bigcup_{l=0}^{k-1} \what J_l} E_{jj} + \sum_{j \in \bigcup_{l=k+1}^{m_2} \what J_l} E_{jj}& -\sum_{j \in \what J_k}h_{1j} (-S)^{-\frac{3\nu_j}{4}}\what f^{(k)}_+(x)^{- \nu_j} e^{\tilde \theta_j(\what r_k)}\mu_{\sigma(j)}^{-1} E_{j\sigma(j)}\\
	\sum_{j \in \what J_k} (-S)^{\frac{3\nu_j}{4}}\what f^{(k)}_+(x)^{\nu_j} e^{-\tilde \theta_j(\what r_k)}E_{jj}  & \sum_{j \in \bigcup_{l=0}^{k-1} \what J_l} E_{jj} + \sum_{j \in \bigcup_{l=k+1}^{m_2} \what J_l} E_{jj}
	\end{pmatrix}\nonumber\\
    &= \begin{pmatrix}
	\sum_{j \in \bigcup_{l=0}^{k-1} \what J_l} E_{jj} + \sum_{j \in \bigcup_{l=k+1}^{m_2} \what J_l} E_{jj}(I_n - C^2) & \sum_{j \in \what J_k} (-S)^{-\frac{3\nu_j}{4}}\what f^{(k)}_-(x)^{-\nu_j} e^{\tilde \theta_j(\what r_k)}E_{jj}\\
	-\sum_{j \in \what J_k}h_{1j}^{-1} (-S)^{\frac{3\nu_j}{4}}\what f^{(k)}_-(x)^{ \nu_j} e^{-\tilde \theta_j(\what r_k)}E_{jj} C  & \sum_{j \in \bigcup_{l=0}^{k-1} \what J_l} E_{jj} + \sum_{j \in \bigcup_{l=k+1}^{m_2} \what J_l} E_{jj}(I_n - C^2)^{-1}
	\end{pmatrix}\nonumber\\
    &\quad \times \begin{pmatrix}
					(I_n - \what C_{k-1}^2)^{-1} & 0\\
					0 & I_n - \what C_{k-1}^2
				\end{pmatrix}\nonumber\\
              &\quad \times  \begin{pmatrix}
	\sum_{j \in \bigcup_{l=0}^{k-1} \what J_l} E_{jj} + \sum_{j \in \bigcup_{l=k+1}^{m_2} \what J_l} E_{jj}& -\sum_{j \in \what J_k}h_{1j}(-S)^{-\frac{3\nu_j}{4}} \what f^{(k)}_+(x)^{- \nu_j} e^{\tilde \theta_j(\what r_k)}\mu_{\sigma(j)}^{-1} E_{j\sigma(j)}\\
	\sum_{j \in \what J_k} (-S)^{\frac{3\nu_j}{4}}\what f^{(k)}_+(x)^{\nu_j} e^{-\tilde \theta_j(\what r_k)}E_{jj}  & \sum_{j \in \bigcup_{l=0}^{k-1} \what J_l} E_{jj} + \sum_{j \in \bigcup_{l=k+1}^{m_2} \what J_l} E_{jj}
	\end{pmatrix}\nonumber\\
    &= \begin{pmatrix}
	\sum_{j \in \bigcup_{l=0}^{k-1} \what J_l} E_{jj} + \sum_{j \in \bigcup_{l=k+1}^{m_2} \what J_l} E_{jj}(I_n - C^2) & \sum_{j \in \what J_k} (-S)^{-\frac{3\nu_j}{4}}\what f^{(k)}_-(x)^{-\nu_j} e^{\tilde \theta_j(\what r_k)}E_{jj}\\
	-\sum_{j \in \what J_k}h_{1j}^{-1}(-S)^{\frac{3\nu_j}{4}} \what f^{(k)}_-(x)^{ \nu_j} e^{-\tilde \theta_j(\what r_k)}E_{jj} C  & \sum_{j \in \bigcup_{l=0}^{k-1} \what J_l} E_{jj} + \sum_{j \in \bigcup_{l=k+1}^{m_2} \what J_l} E_{jj}(I_n - C^2)^{-1}
	\end{pmatrix}\nonumber\\
    &\quad \times \begin{pmatrix}
					\sum_{j\in \bigcup_{l=k}^{m_2}\what J_l}E_{jj}(I_n-C^2)^{-1} + \sum_{j \in \bigcup_{l=0}^{k-1}\what J_l}E_{jj} & 0\\
					0 & \sum_{j\in \bigcup_{l=k}^{m_2}\what J_l}E_{jj}(I_n-C^2) + \sum_{j \in \bigcup_{l=0}^{k-1}\what J_l}E_{jj}
				\end{pmatrix}\nonumber\\
              &\quad \times  \begin{pmatrix}
	\sum_{j \in \bigcup_{l=0}^{k-1} \what J_l} E_{jj} + \sum_{j \in \bigcup_{l=k+1}^{m_2} \what J_l} E_{jj}& -\sum_{j \in \what J_k}h_{1j} (-S)^{-\frac{3\nu_j}{4}}\what f^{(k)}_+(x)^{- \nu_j} e^{\tilde \theta_j(\what r_k)}\mu_{\sigma(j)}^{-1} E_{j\sigma(j)}\\
	\sum_{j \in \what J_k} (-S)^{\frac{3\nu_j}{4}}\what f^{(k)}_+(x)^{\nu_j} e^{-\tilde \theta_j(\what r_k)}E_{jj}  & \sum_{j \in \bigcup_{l=0}^{k-1} \what J_l} E_{jj} + \sum_{j \in \bigcup_{l=k+1}^{m_2} \what J_l} E_{jj}
	\end{pmatrix}\nonumber\\
    &=\begin{pmatrix}
        \sum_{j \in \bigcup_{l=0}^{k-1} \what J_l} E_{jj} & 0\\
        0& \sum_{j \in \bigcup_{l=0}^{k-1} \what J_l} E_{jj}
    \end{pmatrix}+\begin{pmatrix}
        \sum_{j \in \bigcup_{l=k+1}^{m_2} \what J_l} E_{jj} & 0\\
        0& \sum_{j \in \bigcup_{l=k+1}^{m_2} \what J_l} E_{jj}
    \end{pmatrix}\nonumber\\
    &\quad+\begin{pmatrix}
        \sum_{j\in\what J_k}\left(\frac{\what f^{(k)}_+(x)}{\what f^{(k)}_-(x)}\right)^{\nu_j}(1-\mu_j\mu_{\sigma(j)})E_{jj} & 0\\
        0& \sum_{j\in\what J_k}\left(\frac{\what f^{(k)}_+(x)}{\what f^{(k)}_-(x)}\right)^{-\nu_j}(1-\mu_j\mu_{\sigma(j)})^{-1}E_{jj}
    \end{pmatrix}.
\end{align}
Applying the fact that $\what f^{(k)}_+(x)^{\nu_j} = e^{2 \nu_j \pi \ii}\what f^{(k)}_-(x)^{\nu_j}$ and the definition of $\nu_j$ in \eqref{def:nui}, we conclude that
\begin{align}
    \what E^{(k)}_-(x)^{-1}\what E^{(k)}_+(x)=I_{2n},\qquad x \in (\what r_k-\varepsilon, \what r_k). 
\end{align}
Similarly, for  $x \in (\what r_k, \what r_k+\varepsilon)$, one has 
\begin{align}
    & \what E^{(k)}_-(x)^{-1}\what E^{(k)}_+(x) \nonumber\\
    &= \begin{pmatrix}
	\sum_{j \in \bigcup_{l=0}^{k-1} \what J_l} E_{jj} + \sum_{j \in \bigcup_{l=k+1}^{m_2} \what J_l} E_{jj}(I_n - C^2) & \sum_{j \in \what J_k} (-S)^{-\frac{3\nu_j}{4}}\what f^{(k)}(x)^{-\nu_j} e^{\tilde \theta_j(\what r_k)}E_{jj}\\
	-\sum_{j \in \what J_k}h_{1j}^{-1} (-S)^{\frac{3\nu_j}{4}}\what f^{(k)}(x)^{ \nu_j} e^{-\tilde \theta_j(\what r_k)}E_{jj} C  & \sum_{j \in \bigcup_{l=0}^{k-1} \what J_l} E_{jj} + \sum_{j \in \bigcup_{l=k+1}^{m_2} \what J_l} E_{jj}(I_n - C^2)^{-1}
	\end{pmatrix}\nonumber\\
    &\quad \times \what N_-(x)^{-1} \what N_+(x)\nonumber\\
    &\quad \times\begin{pmatrix}
	\sum_{j \in \bigcup_{l=0}^{k-1} \what J_l} E_{jj} + \sum_{j \in \bigcup_{l=k+1}^{m_2} \what J_l} E_{jj}& -\sum_{j \in \what J_k}h_{1j} (-S)^{-\frac{3\nu_j}{4}}\what f^{(k)}(x)^{- \nu_j} e^{\tilde \theta_j(\what r_k)}\mu_{\sigma(j)}^{-1} E_{j\sigma(j)}\\
	\sum_{j \in \what J_k} (-S)^{\frac{3\nu_j}{4}}\what f^{(k)}(x)^{\nu_j} e^{-\tilde \theta_j(\what r_k)}E_{jj}  & \sum_{j \in \bigcup_{l=0}^{k-1} \what J_l} E_{jj} + \sum_{j \in \bigcup_{l=k+1}^{m_2} \what J_l} E_{jj}
	\end{pmatrix}\nonumber\\
    &= \begin{pmatrix}
	\sum_{j \in \bigcup_{l=0}^{k-1} \what J_l} E_{jj} + \sum_{j \in \bigcup_{l=k+1}^{m_2} \what J_l} E_{jj}(I_n - C^2) & \sum_{j \in \what J_k} (-S)^{-\frac{3\nu_j}{4}}\what f^{(k)}(x)^{-\nu_j} e^{\tilde \theta_j(\what r_k)}E_{jj}\\
	-\sum_{j \in \what J_k}h_{1j}^{-1} (-S)^{\frac{3\nu_j}{4}}\what f^{(k)}(x)^{ \nu_j} e^{-\tilde \theta_j(\what r_k)}E_{jj} C  & \sum_{j \in \bigcup_{l=0}^{k-1} \what J_l} E_{jj} + \sum_{j \in \bigcup_{l=k+1}^{m_2} \what J_l} E_{jj}(I_n - C^2)^{-1}
	\end{pmatrix}\nonumber\\
    &\quad \times \begin{pmatrix}
					(I_n - \what C_{k}^2)^{-1} & 0\\
					0 & I_n - \what C_{k}^2
				\end{pmatrix}\nonumber\\
                &\quad \times \begin{pmatrix}
	\sum_{j \in \bigcup_{l=0}^{k-1} \what J_l} E_{jj} + \sum_{j \in \bigcup_{l=k+1}^{m_2} \what J_l} E_{jj}& -\sum_{j \in \what J_k}h_{1j} (-S)^{-\frac{3\nu_j}{4}}\what f^{(k)}(x)^{- \nu_j} e^{\tilde \theta_j(\what r_k)}\mu_{\sigma(j)}^{-1} E_{j\sigma(j)}\\
	\sum_{j \in \what J_k} (-S)^{\frac{3\nu_j}{4}}\what f^{(k)}(x)^{\nu_j} e^{-\tilde \theta_j(\what r_k)}E_{jj}  & \sum_{j \in \bigcup_{l=0}^{k-1} \what J_l} E_{jj} + \sum_{j \in \bigcup_{l=k+1}^{m_2} \what J_l} E_{jj}
	\end{pmatrix}\nonumber\\
     &= \begin{pmatrix}
	\sum_{j \in \bigcup_{l=0}^{k-1} \what J_l} E_{jj} + \sum_{j \in \bigcup_{l=k+1}^{m_2} \what J_l} E_{jj}(I_n - C^2) & \sum_{j \in \what J_k} (-S)^{-\frac{3\nu_j}{4}}\what f^{(k)}(x)^{-\nu_j} e^{\tilde \theta_j(\what r_k)}E_{jj}\\
	-\sum_{j \in \what J_k}h_{1j}^{-1}(-S)^{\frac{3\nu_j}{4}} \what f^{(k)}(x)^{ \nu_j} e^{-\tilde \theta_j(\what r_k)}E_{jj} C  & \sum_{j \in \bigcup_{l=0}^{k-1} \what J_l} E_{jj} + \sum_{j \in \bigcup_{l=k+1}^{m_2} \what J_l} E_{jj}(I_n - C^2)^{-1}
	\end{pmatrix}\nonumber\\
    &\quad \times \begin{pmatrix}
					\sum_{j\in \bigcup_{l=k+1}^{m_2}\what J_l}E_{jj}(I_n-C^2)^{-1} + \sum_{j \in \bigcup_{l=0}^{k}\what J_l}E_{jj} & 0\\
					0 & \sum_{j\in \bigcup_{l=k+1}^{m_2}\what J_l}E_{jj}(I_n-C^2) + \sum_{j \in \bigcup_{l=0}^{k}\what J_l}E_{jj}
				\end{pmatrix}\nonumber\\
              &\quad \times  \begin{pmatrix}
	\sum_{j \in \bigcup_{l=0}^{k-1} \what J_l} E_{jj} + \sum_{j \in \bigcup_{l=k+1}^{m_2} \what J_l} E_{jj}& -\sum_{j \in \what J_k}h_{1j} (-S)^{-\frac{3\nu_j}{4}}\what f^{(k)}_+(x)^{- \nu_j} e^{\tilde \theta_j(\what r_k)}\mu_{\sigma(j)}^{-1} E_{j\sigma(j)}\\
	\sum_{j \in \what J_k} (-S)^{\frac{3\nu_j}{4}}\what f^{(k)}_+(x)^{\nu_j} e^{-\tilde \theta_j(\what r_k)}E_{jj}  & \sum_{j \in \bigcup_{l=0}^{k-1} \what J_l} E_{jj} + \sum_{j \in \bigcup_{l=k+1}^{m_2} \what J_l} E_{jj}
	\end{pmatrix}\nonumber\\
    &=\begin{pmatrix}
        \sum_{j \in \bigcup_{l=0}^{k-1} \what J_l} E_{jj} & 0\\
        0& \sum_{j \in \bigcup_{l=0}^{k-1} \what J_l} E_{jj}
    \end{pmatrix}+\begin{pmatrix}
        \sum_{j \in \bigcup_{l=k+1}^{m_2} \what J_l} E_{jj} & 0\\
        0& \sum_{j \in \bigcup_{l=k+1}^{m_2} \what J_l} E_{jj}
    \end{pmatrix}\nonumber\\
    &\quad+\begin{pmatrix}
        \sum_{j\in\what J_k}\left(\frac{\what f^{(k)}(x)}{\what f^{(k)}(x)}\right)^{\nu_j}E_{jj} & 0\\
        0& \sum_{j\in\what J_k}\left(\frac{\what f^{(k)}(x)}{\what f^{(k)}(x)}\right)^{-\nu_j}E_{jj}
    \end{pmatrix}=I_{2n}.
\end{align}
Moreover, from \eqref{def:hat-N} and \eqref{def:hat-fk-rk}, we observe 
\begin{align}
     \what E^{(k)}(\what r_k) &=\begin{pmatrix}
         \sum_{i\in\I} E_{ii}&0\\
         0&\sum_{i\in\I} E_{ii}
     \end{pmatrix}+\begin{pmatrix}
	\sum_{j \in \bigcup_{l=1}^{k-1} \what J_l} \left(\frac{\what r_k + \what r_l}{\what r_k - \what r_l}\right)^{-\nu_{j}}E_{jj} & 0\\
	0  & \sum_{j \in \bigcup_{l=1}^{k-1} \what J_l} \left(\frac{\what r_k + \what r_l}{\what r_k - \what r_l}\right)^{\nu_{j}}E_{jj} 
	\end{pmatrix}\nonumber\\
    &\quad+\begin{pmatrix}
	 \sum_{j \in \bigcup_{l=k+1}^{m_2} \what J_l} e^{\pi \ii \nu_j}\left(\frac{\what r_k + \what r_l}{\what r_l - \what r_k}\right)^{-\nu_{j}}E_{jj}&0\\
	0  &  \sum_{j \in \bigcup_{l=k+1}^{m_2} \what J_l} e^{-\pi \ii \nu_j}\left(\frac{\what r_k + \what r_l}{\what r_l - \what r_k}\right)^{\nu_{j}}E_{jj}
	\end{pmatrix}\nonumber\\
    &\quad+\begin{pmatrix}
	0& -\sum_{j \in \what J_k}h_{1j} (2\what r_k)^{-\frac{3 \nu_j}{2}} e^{\tilde \theta_j(\what r_k)-\frac{3 \pi \ii \nu_j}{4}}\mu_{\sigma(j)}^{-1} E_{j\sigma(j)}\\
	\sum_{j \in \what J_k} (2 \what r_k)^{\frac{3 \nu_j}{2}} e^{-\tilde \theta_j(\what r_k)+\frac{3 \pi \ii \nu_j}{4}}E_{jj}  & 0
	\end{pmatrix}.
\end{align}
% {\tiny
% \begin{align}
%     \what E^{(k)}(\what r_k) = \begin{pmatrix}
% 	\sum_{j \in \bigcup_{l=1}^{k-1} \what J_l} \left(\frac{\what r_k + \what r_l}{\what r_k - \what r_l}\right)^{-\nu_{j}}E_{ii} + \sum_{j \in \bigcup_{l=k+1}^{m_2} \what J_l} e^{\pi \ii \nu_j}\left(\frac{\what r_k + \what r_l}{\what r_l - \what r_k}\right)^{-\nu_{j}}E_{ii}& -\sum_{j \in \what J_k}h_{1j} (2\what r_k)^{-\frac{3 \nu_j}{2}} e^{\tilde \theta_j(\what r_k)-\frac{3 \pi \ii \nu_j}{4}}E_{ii} C^{-1}\\
% 	\sum_{j \in \what J_k} (2 \what r_k)^{\frac{3 \nu_j}{2}} e^{-\tilde \theta_j(\what r_k)+\frac{3 \pi \ii \nu_j}{4}}E_{ii}  & \sum_{j \in \bigcup_{l=1}^{k-1} \what J_l} \left(\frac{\what r_k + \what r_l}{\what r_k - \what r_l}\right)^{\nu_{j}}E_{ii} + \sum_{j \in \bigcup_{l=k+1}^{m_2} \what J_l} e^{-\pi \ii \nu_j}\left(\frac{\what r_k + \what r_l}{\what r_l - \what r_k}\right)^{\nu_{j}}E_{ii}
% 	\end{pmatrix}
% \end{align}}
Thus, $\what E^{(k)}(z)$ is indeed analytic near $\what r_k$.
It's then straightforward to verify the jump condition \eqref{def:hat-JP-k} by using the analyticity of $\what E^{(k)}(z)$, \eqref{jump:PC}, \eqref{hat-Uk} and \eqref{hat-Uk-1}.

For the matching condition \eqref{eq:matchhatPk}, applying \eqref{def:hat-Pk} along with the asymptotic behavior of the parabolic cylinder parametrix $\Phi^{(\PC)}(z)$ at infinity in \eqref{infty:PC} gives, for $\Im z > 0$,
\begin{align}
    &\what P^{(k)}(z)\what N(z)^{-1} \nonumber\\
    &= \what N(z)\begin{pmatrix}
        I_n-\sum_{j\in\what J_k}E_{jj} & -\sum_{j\in\what J_k}h_{1j}\mu_{\sigma(j)}^{-1}E_{j\sigma(j)}\\
        \sum_{j\in\what J_k}E_{jj} & I_n-\sum_{j\in\what J_k}E_{jj}
    \end{pmatrix} \Bigg[I_{2n}\nonumber\\
   &\quad +\frac{1}{(-S)^{\frac{3}{4}}f(z)}\begin{pmatrix}
        0 & \sum_{j \in \what J_k} \nu_j(-S)^{\frac{3\nu_j}{2}} \what f^{(k)}(z)^{2 \nu_j} e^{-2 \tilde \theta_j(\what r_k)}E_{jj}\\
        \sum_{j \in \what J_k}(-S)^{-\frac{3\nu_j}{2}} \what f^{(k)}(z)^{-2 \nu_j} e^{2 \tilde \theta_j(\what r_k)}E_{jj}&0
    \end{pmatrix}\nonumber\\
   &\quad +\Boh\left((-S)^{-\frac 32}\right)\Bigg]\begin{pmatrix}
        I_n-\sum_{j\in\what J_k}E_{jj} & \sum_{j\in\what J_k}E_{jj}\\
        -\sum_{j\in\what J_k}h_{1j}^{-1}\mu_{j}E_{\sigma(j)j} & I_n-\sum_{j\in\what J_k}E_{jj}\end{pmatrix}\what N(z)^{-1}, \quad S\to-\infty. 
\end{align}
By substituting explicit formula of $\what N$ given in \eqref{def:hat-N}, we obtain after simplification that
\begin{align}\label{match2}
\what P^{(k)}(z) \what N(z)^{-1}= I_{2n} + \frac{\what J_1^{(k)}(z)}{(-S)^{\frac 34}} + \Boh\left((-S)^{-\frac 32}\right), \qquad S \to -\infty,
\end{align}
where %uniformly for $z \in \partial D(\what r_k, \varepsilon)$, where
\begin{align}\label{hat-Jk}
    \what J_1^{(k)}(z) := \begin{pmatrix}
        0 & \left[\what J_1^{(k)}(z)\right]_{12}\\
        \left[\what J_1^{(k)}(z)\right]_{21} & 0
    \end{pmatrix}
\end{align}
with 
\begin{align}
    \left[\what J_1^{(k)}(z)\right]_{12} & = -\sum_{j \in \what J_k} h_{1j}(-S)^{-\frac{3\nu_j}{2}}  \what f^{(k)}(z)^{-2 \nu_j-1} e^{2 \tilde \theta_j(\what r_k)} \left(\frac{z-\what r_k}{z + \what r_k}\right)^{2 \nu_j}\mu_{\sigma(j)}^{-1} E_{j \sigma(j)},\label{hat-Jk12}\\
    \left[\what J_1^{(k)}(z)\right]_{21}  & = -\sum_{j \in \what J_k} \frac{\nu_j}{h_{1j}}(-S)^{\frac{3\nu_j}{2}} \what f^{(k)}(z)^{2 \nu_j-1} e^{-2 \tilde \theta_j(\what r_k)}  \left(\frac{z-\what r_k}{z + \what r_k}\right)^{-2 \nu_j}E_{jj}C.\label{hat-Jk21}
\end{align}
Similarly, for  $\Im z<0$, by applying the same methodology, we can derive the identical formula as presented in \eqref{match2}. This result suggests that the matching condition \eqref{eq:matchhatPk} holds uniformly for $z \in \partial D(\what r_k, \varepsilon)$.
% {\tiny
% \begin{align}\label{hat-Jk}
% \what J_1^{(k)}(z)=-\frac{1}{\what f^{(k)}(z)} \begin{pmatrix}
% 0 & \sum_{j \in \what J_k} h_{1j}(-S)^{-\frac{3\nu_j}{2}}  \what f^{(k)}(z)^{-2 \nu_j} e^{2 \tilde \theta_j(\what r_k)} \left(\frac{z-\what r_k}{z + \what r_k}\right)^{2 \nu_j}E_{ii}C^{-1}\\
% \sum_{j \in \what J_k} \frac{\nu_j}{h_{1j}}(-S)^{\frac{3\nu_j}{2}} \what f^{(k)}(z)^{2 \nu_j} e^{-2 \tilde \theta_j(\what r_k)}  \left(\frac{z-\what r_k}{z + \what r_k}\right)^{-2 \nu_j}E_{ii}C & 0
% \end{pmatrix}.
% \end{align}}
\end{proof}

 %    \subsection{Local parametrix near $-\what r_k$}
 %    Similarly, we can use the following symmetric relationship to construct the local parametrix near $-\what r_k$:
	% \begin{equation}\label{relation:hatPk-hatP-k}
	% 	\what P^{(k)}(z) = \what{\sigma}_1 \what P^{(-k)}(-z)\what{\sigma}_1,
	% \end{equation}
	% where $\what{\sigma}_1$ is defined in \eqref{def:sigmak}.
    
    \subsection{Final transformation}
    
    The final transformation is defined by 
    \begin{align}\label{def:hatR}
    \what R(z) = \begin{cases}
    \what S(z) \what P^{( k)}(z)^{-1}, \quad & z \in D(\what r_k, \varepsilon),\\
    \what S(z) \what P^{(-k)}(z)^{-1}, \quad & z \in D(-\what r_k, \varepsilon),\\
    \what S(z) \what N(z)^{-1}, \quad & \textrm{elsewhere.}
    \end{cases}
    \end{align}
From the RH problems for $\what S$, $\what P^{(\pm k)}$ and $\what N$, it can be verified  that $\pm \what{r}_k$ are actually removable singularities of $\what R$. It is then easily seen that $\what R$ satisfies the following RH problem.
    \begin{rhp}\label{rhp:hatR}
    \hfill
		\begin{itemize}
			\item [\rm{(a)}] $\what R(z)$ is defined and analytic in $\mathbb{C} \setminus \what \Gamma_{R}$, where
			\begin{equation}
				\what \Gamma_{R}:=\what \Gamma_S \cup \bigcup_{k=1}^{m_2}\left( \partial D(\what{r}_k,\varepsilon) \cup \partial D(-\what {r}_k,\varepsilon)\right) \setminus \{ \bigcup_{k=1}^{m_2} \left( D(\what{r}_k,\varepsilon) \cup D(-\what{r}_k,\varepsilon)\right) \cup  (-\what r_{m_2}, \what r_{m_2})\}.
			\end{equation}
			%and the orientation of $\partial D(\pm \what{r}_k,\varepsilon)$ are counterclockwise.
			\item [\rm{(b)}] For $z \in \what\Gamma_{R}$, we have
			\begin{equation}\label{eq:hatRjump}
				\what R_+(z) = \what R_-(z) J_{\what R} (z),
			\end{equation}
			where
			\begin{equation}\label{def:hatJR}
				J_{\what R}(z) := \begin{cases}
					\what P^{(\pm k)}(z) \what N(z)^{-1}, & \quad z \in \partial D(\pm\what{r}_k, \varepsilon), \quad k=1,\ldots,m_2, \\
					% \what P^{(-k)}(z) \what N(z)^{-1}, & \quad z \in \partial D(-\what{r}_k, \varepsilon),\\
					\what N_{-}(z) J_{\what S}(z) \what N_{+}(z)^{-1}, & \quad z \in \what\Gamma_{R} \setminus \{\bigcup_{k=1}^{m_2} \left( D(\what{r}_k,\varepsilon) \cup D(-\what{r}_k,\varepsilon)\right)\},
				\end{cases}
			\end{equation}
			with $J_{\what S}(z)$ defined in \eqref{def:what-JS}.
			\item [\rm{(c)}] As $z \to \infty$, we have
			\begin{equation}\label{eq:asyhatR}
				\what R(z) = I_{2n} + \frac{\what R_1}{z} + \Boh (z^{-2}),
			\end{equation}
			where $\what R_1$ is independent of $z$.
		\end{itemize}
    \end{rhp}
    From \eqref{def:what-JS}, \eqref{relation:hatPk-hatP-k} and \eqref{eq:matchhatPk}, we conclude that, as $S \to -\infty$,
    \begin{align}\label{eq:hatJR}
    J_{\what R}(z) =\begin{cases}
    I_{2n} + \Boh \left(e^{-\what c(-S)^{\frac 32}}\right), \quad & z \in \what \Gamma_R \setminus  \{\bigcup_{k=1}^{m_2} \left(\partial D(\what{r}_k, \varepsilon) \cup \partial D(-\what{r}_k, \varepsilon)\right)\},\\
    I_{2n} + \frac{\what J_1^{(k)}(z)}{(-S)^{\frac 34}}+\Boh\left((-S)^{-\frac 32}\right), \quad & z \in \partial D(\what{r}_k, \varepsilon), \\
    I_{2n} + \frac{\what \sigma_1 \what J_1^{(k)}(-z)\what \sigma_1}{(-S)^{\frac 34}}+\Boh\left((-S)^{-\frac 32}\right), \quad & z \in \partial D(-\what{r}_k, \varepsilon),
    \end{cases}
    \end{align}
    for some $\what c > 0$, where $\what \sigma_1$ and $\what J_1^{(k)}(z)$  are defined in \eqref{def:hatsigma1} and  \eqref{hat-Jk}. By a standard argument \cite{Deift1999, Deift1993}, we conclude that
    \begin{align}\label{eq:hatR}
    \what R(z) = I_{2n} + \frac{\what R^{(1)}(z)}{(-S)^{\frac 34}}+\Boh\left((-S)^{-\frac 32}\right), \qquad S \to -\infty,
    \end{align}
    uniformly for $z \in \mathbb{C} \setminus \what \Gamma_R$.
    
    Moreover, inserting \eqref{eq:hatJR} and \eqref{eq:hatR} into the jump condition \eqref{eq:hatRjump} for $\what R$ yields
    \begin{align}
    \what R_+^{(1)}(z) =  \what R_-^{(1)}(z) + 
    \begin{cases}
        J_1^{(k)}(z),  \quad & z \in \partial D(\what{r}_k, \varepsilon),
        \\
        \what \sigma_1 \what J_1^{(k)}(-z)\what \sigma_1, \quad & z \in \partial D(-\what{r}_k, \varepsilon),
    \end{cases}
        \end{align}
    for $k=1,\ldots,m_2$. This, together with the fact that $\what R^{(1)}(z) = \Boh(z^{-1})$ as $z \to \infty$, implies that
    \begin{align}
     \what R^{(1)}(z) =\sum_{k=1}^{m_2}\left( \frac{1}{2 \pi \ii} \int_{\partial D(\what{r}_k, \varepsilon)} \frac{\what J_1^{(k)}(\zeta)}{\zeta-z}d\zeta + \frac{1}{2 \pi \ii} \int_{\partial D(-\what{r}_k, \varepsilon)} \frac{\what \sigma_1 \what J_1^{(k)}(-\zeta)\what \sigma_1}{\zeta-z}d\zeta\right).
    \end{align}
    Based on the definition of $\what J_1^{(k)}$ in \eqref{hat-Jk}, and applying the residue theorem to the equation above, we obtain 
    \begin{align}\label{eq:hatR1}
    \what R^{(1)}(z) = \sum_{k=1}^{m_2}\left(\frac{A^{(k)}}{z-\what r_k} - \frac{\what \sigma_1 A^{(k)} \what \sigma_1}{z + \what r_k}\right), \quad z \in \mathbb{C}\setminus \{\partial D(\what{r}_k, \varepsilon) \cup \partial D(-\what{r}_k, \varepsilon)\},
    \end{align}
    where 
    \begin{align}\label{Ak}
        A^{(k)}:= \Res_{z = \what r_k} \what J_1^{(k)}(z) = \begin{pmatrix}
            0 & \left[A^{(k)}\right]_{12}\\
            \left[A^{(k)}\right]_{21} & 0
        \end{pmatrix}
    \end{align}
    with 
    \begin{align}
        \left[A^{(k)}\right]_{12} &:= -e^{-\frac{3 \pi \ii}{4}} (2 \what r_k)^{-\frac 12}\sum_{j \in \what J_k} h_{1j} (-S)^{-\frac{3\nu_j}{2}} e^{-\frac{3 \pi \ii \nu_j}{2}} (2 \what r_k)^{-3\nu_j} e^{2 \tilde \theta_j(\what r_k)} \mu_{\sigma(j)}^{-1} E_{j \sigma(j)},\\
         \left[A^{(k)}\right]_{21} &:= -e^{-\frac{3 \pi \ii}{4}} (2 \what r_k)^{-\frac 12}\sum_{j \in \what J_k} \frac{\nu_j}{h_{1j}}  (-S)^{\frac{3\nu_j}{2}} e^{\frac{3 \pi \ii \nu_j}{2}} (2 \what r_k)^{3\nu_j}  e^{-2 \tilde \theta_j(\what r_k)}  E_{jj}C.
    \end{align}
%     {\tiny
%         \begin{align}\label{Ak}
%     A^{(k)}&= \Res_{z = \what r_k} \what J_1^{(k)}(z) \nonumber\\
%     &= -e^{-\frac{3 \pi \ii}{4}} (2 \what r_k)^{-\frac 12}\begin{pmatrix}
% 0 & \sum_{j \in \what J_k} h_{1j} (-S)^{-\frac{3\nu_j}{2}} e^{-\frac{3 \pi \ii \nu_j}{2}} (2 \what r_k)^{-3\nu_j} e^{2 \tilde \theta_j(\what r_k)} \mu_{\sigma(j)}^{-1} E_{j \sigma(j)}\\
% \sum_{j \in \what J_k} \frac{\nu_j}{h_{1j}}  (-S)^{\frac{3\nu_j}{2}} e^{\frac{3 \pi \ii \nu_j}{2}} (2 \what r_k)^{3\nu_j}  e^{-2 \tilde \theta_j(\what r_k)}  E_{ii}C & 0
% \end{pmatrix}.
%     \end{align}}

\section{Proof of Theorem \ref{-infty asympototics results}}\label{proof}
As outcomes of the analysis performed in the previous two sections, we are ready to prove our main result. Recall the relationship between $\beta_1$ and RH problem \ref{rhp:pole-free} given in \eqref{def:beta1}, together with the definitions  of $\Psi_{\mathcal I,1}$ and $\Psi_{\mathcal J,1}$ in \eqref{I1} and \eqref{J1}, we have
     \begin{align}\label{beta1}
         \beta_1 = -\ii \lim_{\lambda\to \infty} \lambda [\Xi]_{12}(\lambda;\vec {s}) = -\ii [\Xi_1]_{12} =-\ii [\Psi_{\mathcal{I},1}]_{12}-\ii [\Psi_{\mathcal{J},1}]_{12}.
     \end{align}
  We next derive asymptotics of $\Psi_{\mathcal I,1}$ and $\Psi_{\mathcal J,1}$ as $S \to -\infty$, respectively. 
    % As we mentioned before, the specific choice of $C=\Lambda P$ facilitates the proof of Theorem \ref{-infty asympototics results} in the cases of $C^2=I_n$ and $\det (I_n -C^2) \neq 0$ independently.
	\paragraph{Asymptotics of $\Psi_{\mathcal I,1}$}
    %\label{sec:5.1}
	 Following the sequence of transformations $\Psi_{\I} \to T \to R$ described in \eqref{def:PsitoT} and \eqref{def:R}, we obtain
\begin{align}
	[\Psi_{\mathcal{I},1}]_{12} = \sqrt{-S} \left[ R_1+ N_1\right]_{12}.
	\end{align}
        From the explicit formula of $N(z)$ in \eqref{def:N} and  the definition of $\gamma_i(z)$ in \eqref{def:gamma}, one has
    \begin{align}
        \left[N_1\right]_{12}=\sum_{i \in \I} \ii \frac{\sqrt{r_i}}{2}E_{ii} C
        %\ii \operatorname{diag}\left(\frac{\sqrt{r_1}}{2},\cdots,\frac{\sqrt{r_n}}{2}\right)C
    \end{align}
    with $r_i(z)$ defined in \eqref{def:ri}.
    By comparing the asymptotic behaviors of $R(z)$ in \eqref{eq:asyR} and \eqref{def:RR}, we have
	\begin{align}
	\sqrt{-S}R_1=\Boh\left((-S)^{-1}\right).
	\end{align}
    Thus, it follows that 
    \begin{align}\label{psi-12}
        -\ii [\Psi_{\mathcal{I},1}]_{12}&=\sqrt{-S}\sum_{i \in \I} \frac{\sqrt{r_i}}{2}E_{ii} C+\Boh\left((-S)^{-1}\right)\nonumber\\
        &=\sum_{i \in \I}\sqrt{\frac{-s_i-s_{\sigma(i)}}{2}}E_{ii} C+\Boh\left((-S)^{-1}\right).
        %\operatorname{diag}\left(\sqrt{\frac{-s_1-s_{\sigma(1)}}{2}},\cdots,\sqrt{\frac{-s_n-s_{\sigma(n)}}{2}} \right)C+\Boh\left((-S)^{-1}\right).
    \end{align}
    %  Since
    % \begin{align}
    %     C=\sum_{i=1}^n c_{i\sigma(i)}E_{i\sigma(i)},
    % \end{align}
    % we recover the first formula in \eqref{-infty asympototics}.
	\paragraph{Asymptotics of $\Psi_{\mathcal J,1}$}
    Tracing back the transformations $\Psi_{\mathcal{J}} \to \what T \to \what S \to \what R$ defined in \eqref{def:hat-T}, \eqref{def:hat-S} and \eqref{def:hatR}, we have
	\begin{align}
	 [\Psi_{\mathcal J,1}]_{12}=\sqrt{-S} \left[\what R_1+\what N_1\right]_{12}.
	\end{align}
    Using the explicit expression of $\what N(z)$ in \eqref{def:hat-N}, one can show that
    \begin{align}
        \left[\what N_1\right]_{12}=0.
    \end{align}
    By comparing the asymptotic behaviors of $\what R(z)$ in \eqref{eq:asyhatR} and \eqref{eq:hatR}, and subsequently applying \eqref{eq:hatR1}, we arrive at
	\begin{align}
	[\Psi_{\mathcal J,1}]_{12} = (-S)^{-\frac 14} \sum_{k=1}^{m_2}\left[A^{(k)} - \begin{pmatrix}
		    0 & I_n\\
                I_n & 0
		\end{pmatrix} A^{(k)} \begin{pmatrix}
		    0 & I_n\\
                I_n & 0
		\end{pmatrix}\right]_{12} + \Boh\left((-S)^{-1}\right).
	\end{align}
	Upon insertion of the expression for $A^{(k)}$ defined in \eqref{Ak} into the above equation, it follows that
	\begin{align}\label{eq:beta1-2-2}
	[\Psi_{\mathcal J,1}]_{12} &= - (-S)^{-\frac 14} e^{-\frac{3 \pi \ii}{4}} \sum_{k=1}^{m_2} (2 \what r_k)^{-\frac 12}\left[\sum_{j \in \what J_k} \left(h_{1j} (-S)^{-\frac{3\nu_j}{2}} e^{-\frac{3 \pi \ii \nu_j}{2}} (2 \what r_k)^{-3\nu_j} e^{2 \tilde \theta_j(\what r_k)} \mu_{\sigma(j)}^{-1} E_{j \sigma(j)}\right.\right.\nonumber\\
	& \left.\left.\quad- \frac{\nu_j}{h_{1j}}  (-S)^{\frac{3\nu_j}{2}} e^{\frac{3 \pi \ii \nu_j}{2}} (2 \what r_k)^{3\nu_j}  e^{-2 \tilde \theta_j(\what r_k)}  E_{jj}C\right)\right]+ \Boh\left((-S)^{-1}\right)\nonumber\\
	& = \ii (-S)^{-\frac 14} e^{-\frac{\pi \ii}{4}} \sum_{j \in \bigcup_{k=1}^{m_2}\what J_k} \left[ (2 \what r_k)^{-\frac 12}\left(h_{1j} (-S)^{-\frac{3\nu_j}{2}} e^{-\frac{3 \pi \ii \nu_j}{2}} (2 \what r_k)^{-3\nu_j} e^{2 \tilde \theta_j(\what r_k)}\mu_{\sigma(j)}^{-1} E_{j \sigma(j)}\right.\right.\nonumber\\
	& \left.\left.\quad- \frac{\nu_j}{h_{1j}}  (-S)^{\frac{3\nu_j}{2}} e^{\frac{3 \pi \ii \nu_j}{2}} (2 \what r_k)^{3\nu_j}  e^{-2 \tilde \theta_j(\what r_k)}  E_{jj}C\right)\right]+ \Boh\left((-S)^{-1}\right),
	\end{align}
where $\widetilde \theta_j$	is defined in \eqref{def:tilde-theta}. Using the definition of $h_{1j}$ in \eqref{def:h1i}, we could rewrite \eqref{eq:beta1-2-2} as 
	% $\nu_j$, $\tilde \theta_j(z)$ in \eqref{def:h1i}, \eqref{def:nui} and \eqref{def:tilde-theta}, we can simplify \eqref{eq:beta1-2-2} as
	\begin{align}\label{eq:beta-1-2-31}
	[\Psi_{\mathcal J,1}]_{12}& = \ii (-S)^{-\frac 14}e^{-\frac{\pi \ii}{4}} \sum_{j \in \bigcup_{k=1}^{m_2}\what J_k} \left[ (2 \what r_k)^{-\frac 12}\left((-S)^{-\frac{3\nu_j}{2}} e^{-\frac{\pi \ii \nu_j}{2}} (2 \what r_k)^{-3\nu_j} e^{2 \tilde \theta_j(\what r_k)} \frac{\sqrt{2 \pi}}{\Gamma(-\nu_j)}\mu_{\sigma(j)}^{-1} E_{j \sigma(j)}\right.\right.\nonumber\\
	& \left.\left.\quad-  (-S)^{\frac{3\nu_j}{2}} e^{\frac{\pi \ii \nu_j}{2}} (2 \what r_k)^{3\nu_j}  e^{-2 \tilde \theta_j(\what r_k)} \frac{\nu_j \Gamma(-\nu_j)}{\sqrt{2 \pi}} E_{jj}C\right)\right]+ \Boh\left((-S)^{-1}\right).
    \end{align}
   On account of the singularity of Gamma function at 0, we define two sets to distinguish between the cases $\nu_j=0$ and $\nu_j \neq 0$:
    \begin{align}
        \J_1 = \{j: c_{j\sigma(j)}c_{\sigma(j)j}=0\}, \qquad \J_2 = \J\setminus \J_1.
    \end{align}
 Thus, equation \eqref{eq:beta-1-2-31} can be further written as
    \begin{align}\label{eq:beta-1-2-3}
	&[\Psi_{\mathcal J,1}]_{12} = \ii (-S)^{-\frac 14} \nonumber\\
    &\quad \times \left\{\sum_{j \in \bigcup_{k=1}^{m_2}\what J_k \bigcap \J_1} \left[ (2 \what r_k)^{-\frac 12}\left((-S)^{-\frac{3\nu_j}{2}} e^{-\frac{\pi \ii \nu_j}{2}} (2 \what r_k)^{-3\nu_j} e^{2 \tilde \theta_j(\what r_k)}e^{-\frac{\pi \ii}{4}} \frac{\sqrt{2 \pi}}{\Gamma(-\nu_j)}\mu_{\sigma(j)}^{-1} E_{j \sigma(j)}\right.\right.\right.\nonumber\\
	& \left.\left.\quad-  (-S)^{\frac{3\nu_j}{2}} e^{\frac{\pi \ii \nu_j}{2}} (2 \what r_k)^{3\nu_j}  e^{-2 \tilde \theta_j(\what r_k)} e^{-\frac{\pi \ii}{4}}\frac{\nu_j \Gamma(-\nu_j)}{\sqrt{2 \pi}} E_{jj}C\right)\right]\nonumber\\
    &\quad + \sum_{j \in \bigcup_{k=1}^{m_2}\what J_k \bigcap \J_2} \left[ (2 \what r_k)^{-\frac 12}\left((-S)^{-\frac{3\nu_j}{2}} e^{-\frac{\pi \ii \nu_j}{2}} (2 \what r_k)^{-3\nu_j} e^{2 \tilde \theta_j(\what r_k)}e^{-\frac{\pi \ii}{4}} \frac{\sqrt{2 \pi}}{\Gamma(-\nu_j)}\mu_{\sigma(j)}^{-1} E_{j \sigma(j)}\right.\right.\nonumber\\
	& \left.\left.\left.\quad-  (-S)^{\frac{3\nu_j}{2}} e^{\frac{\pi \ii \nu_j}{2}} (2 \what r_k)^{3\nu_j}  e^{-2 \tilde \theta_j(\what r_k)}e^{-\frac{\pi \ii}{4}} \frac{\nu_j \Gamma(-\nu_j)}{\sqrt{2 \pi}} E_{jj}C\right)\right]\right\}+ \Boh\left((-S)^{-1}\right).
    \end{align}
    
If $\nu_j=0$, that is $\mu_j\mu_{\sigma(j)}=0$. Recall \eqref{hij=0}, \eqref{hij=0-1} and the property of Gamma function, we have
    \begin{align}\label{6.13}
        \frac{1}{\Gamma(-\nu_j)\mu_{\sigma(j)}}E_{j\sigma(j)} =\frac{1}{\Gamma(-\nu_j)\mu_j\mu_{\sigma(j)}}E_{jj} C=\frac{\ii}{2 \pi}E_{jj} C
    \end{align}
    and 
    \begin{align}\label{6.14}
        \nu_j\Gamma(-\nu_j)=-\Gamma(1-\nu_j)=-1.
    \end{align}
   Substituting \eqref{6.13} and \eqref{6.14} into the first summation in \eqref{eq:beta-1-2-3}, we have
   \begin{align}\label{Psi-J1}
	& \sum_{j \in \bigcup_{k=1}^{m_2}\what J_k\bigcap\J_1}(2 \what r_k)^{-\frac 12}e^{-\frac{\pi \ii}{4}} \left((-S)^{-\frac{3\nu_j}{2}} e^{-\frac{\pi \ii \nu_j}{2}} (2 \what r_k)^{-3\nu_j} e^{2 \tilde \theta_j(\what r_k)} \frac{\sqrt{2 \pi}}{\Gamma(-\nu_j)}\mu_{\sigma(j)}^{-1} E_{j \sigma(j)}\right.\nonumber\\
	& \left.\quad-  (-S)^{\frac{3\nu_j}{2}} e^{\frac{\pi \ii \nu_j}{2}} (2 \what r_k)^{3\nu_j}  e^{-2 \tilde \theta_j(\what r_k)} \frac{\nu_j \Gamma(-\nu_j)}{\sqrt{2 \pi}} E_{jj}C\right)\nonumber\\
     &=\sum_{j \in \bigcup_{k=1}^{m_2}\what J_k\bigcap\J_1}(2 \what r_k)^{-\frac 12}e^{-\frac{\pi \ii}{4}} \left(e^{2 \tilde \theta_j(\what r_k)}\frac{\ii}{\sqrt{2 \pi}}E_{jj} C+  e^{-2 \tilde \theta_j(\what r_k)} \frac{1}{\sqrt{2 \pi}} E_{jj}C\right)\nonumber\\
    %& = e^{2 \tilde \theta_j(\what r_k)}\frac{e^{\frac{\pi \ii}{4}}}{\sqrt{2 \pi}}E_{jj} C+   e^{-2 \tilde \theta_j(\what r_k)} \frac{e^{-\frac{\pi \ii}{4}}}{\sqrt{2 \pi}} E_{jj}C\nonumber\\
     & =   \sum_{j \in \bigcup_{k=1}^{m_2}\what J_k\bigcap\J_1}(2 \what r_k)^{-\frac 12}\frac{1}{\sqrt{2 \pi}}\left(e^{2 \tilde \theta_j(\what r_k) +\frac{\pi\ii}{4}} +e^{-2 \tilde \theta_j(\what r_k) -\frac{\pi\ii}{4}}  \right)E_{jj}C\nonumber\\
     &=\sum_{j \in \bigcup_{k=1}^{m_2}\what J_k\bigcap\J_1}(2 \what r_k)^{-\frac 12}\sqrt{\frac{2}{\pi}}\cos \left(2 \ii \tilde \theta_j\left(\what r_k\right)- \frac{\pi}{4} \right) E_{jj} C.
 %     \left((-S)^{-\frac{3\nu_j}{2}} e^{-\frac{\pi \ii \nu_j}{2}} (2 \what r_k)^{-3\nu_j} e^{2 \tilde \theta_j(\what r_k)}e^{\frac{\pi \ii}{4}}E_{jj} C\right.\right.\nonumber\\
	% & \left.\left.\quad+  (-S)^{\frac{3\nu_j}{2}} e^{\frac{\pi \ii \nu_j}{2}} (2 \what r_k)^{3\nu_j}  e^{-2 \tilde \theta_j(\what r_k)} e^{-\frac{\pi \ii}{4}} E_{jj}C\right)\right]+ \Boh\left((-S)^{-1}\right)\nonumber\\
    \end{align}
 If  $\nu_j \neq 0$,  we apply the reflection formula of the Gamma function to deduce that \begin{align}\label{reflection}
	 \Gamma(\nu_j)\Gamma(-\nu_j) = -\frac{ \pi}{\nu_j\sin(\nu_j\pi)}=-\frac{2\pi \ii e^{-\pi\ii\nu_j}}{\nu_j(1-e^{-2 \pi\ii\nu_j})}.
	\end{align}
	Note that 
    \begin{align}
       \Gamma(-\nu_j)=\sqrt{\Gamma(-\nu_j)\Gamma(-\nu_j)} =\sqrt{\frac{\Gamma(-\nu_j)}{\Gamma(\nu_j)}}\sqrt{\frac{-2\pi\ii}{\nu_j(1-e^{-2\pi\ii\nu_j})}}e^{-\frac{\pi\ii\nu_j}{2}}
    \end{align}
   and $\nu_j =-\frac{1}{2 \pi \ii}\ln (1- \mu_j \mu_{\sigma(j)})$, we have
	
 %    And when $\nu_j =-\frac{1}{2 \pi \ii}\ln (1- \mu_j \mu_{\sigma(j)}) \ii \mathbb{R}$, we rewrite \eqref{eq:beta-1-2-3} as follows by using \eqref{reflection}
    \begin{align}\label{6.18}
    &\sum_{j \in \bigcup_{k=1}^{m_2}\what J_k \bigcap\J_2}(2 \what r_k)^{-\frac 12}\left[e^{-\nu_j\ln{\left((-S)^{\frac 32} (2\what r_k)^3\right)}+2 \tilde \theta_j(\what r_k)-\frac{\pi\ii}{4}}\right. \nonumber\\
    &\times \sqrt{\frac{\Gamma(\nu_j)}{\Gamma(-\nu_j)}}\sqrt{\ii \nu_j(1-e^{-2 \pi \ii \nu_j})}\mu_{\sigma(j)}^{-1} E_{j \sigma(j)}\nonumber\\
	& \left.+e^{\nu_j\ln{\left((-S)^{\frac 32} (2\what r_k)^3\right)}-2 \tilde \theta_j(\what r_k) +\frac{\pi\ii}{4}} \sqrt{\frac{\ii \nu_j }{1-e^{-2\pi\ii\nu_j}}}\sqrt{\frac{\Gamma(-\nu_j)}{\Gamma(\nu_j)}}E_{jj}C\right]\nonumber\\
    &=\sum_{j \in \bigcup_{k=1}^{m_2}\what J_k\bigcap\J_2}(2 \what r_k)^{-\frac 12}\sqrt{-\frac{\ln(1-\mu_j\mu_{\sigma(j)})}{2 \pi \mu_j\mu_{\sigma(j)}}}\left[e^{-\nu_j\ln{\left((-S)^{\frac 32} (2\what r_k)^3\right)}+2 \tilde \theta_j(\what r_k) -\frac{\pi\ii}{4}+\frac{1}{2}\ln\frac{\Gamma(\nu_j)}{\Gamma(-\nu_j)}}E_{jj}C\right.\nonumber\\
	&\quad \left.+e^{\nu_j\ln{\left((-S)^{\frac 32} (2\what r_k)^3\right)}-2 \tilde \theta_j(\what r_k) +\frac{\pi\ii}{4}-\frac{1}{2}\ln\frac{\Gamma(\nu_j)}{\Gamma(-\nu_j)}}  E_{jj}C\right].
	\end{align}
    The last identity follows directly from $\mu_j = c_{j\sigma(j)}$ and \eqref{expansion}. Inserting $\what r_k$ in \eqref{def:what-J_k} into the above equation and using \eqref{def:rj}, we simplify \eqref{6.18} as 
    \begin{align}\label{eq:beta-1-2-5}
	\sum_{j \in \bigcup_{k=1}^{m_2}\what J_k\bigcap\J_2}\what r_k^{-\frac 12} \sqrt{-\frac{\ln(1-\mu_j\mu_{\sigma(j)})}{\pi \mu_j\mu_{\sigma(j)}}} \cos \psi(s_j, s_{\sigma(j)}) E_{jj} C 
	\end{align}
where
	\begin{multline}\label{def-psi}
	 \psi(s_j, s_{\sigma(j)}) = 2 \ii \tilde \theta_j\left(\sqrt{\frac{s_j+s_{\sigma(j)}}{S}}\right)  +\frac{3}{4\pi}\ln (1-\mu_j\mu_{\sigma(j)}) \ln (-4(s_j+s_{\sigma(j)}))
     \\+\frac{\ii}{2}\ln\frac{\Gamma(\nu_j)}{\Gamma(-\nu_j)} + \frac{\pi}{4}.
	\end{multline}
 %    In particular, when $\nu_j =-\frac{1}{2 \pi \ii}\ln (1- \mu_j \mu_{\sigma(j)})$ is purely imaginary and nonzero, we utilize the properties of the Gamma function for purely imaginary arguments
 %     \begin{align}
 %       |\Gamma(\nu_j)|^2=|\Gamma(-\nu_j)|^2 = \Gamma(\nu_j)\Gamma(-\nu_j),
 %     \end{align}
 %     which implies 
 %     \begin{align}
 %         \ln\frac{\Gamma(\nu_j)}{\Gamma(-\nu_j)} = \ln \frac{|\Gamma(\nu_j)|e^{\ii \arg \Gamma(\nu_j)}}{|\Gamma(-\nu_j)|e^{\ii \arg \Gamma(-\nu_j)}}=\ln e^{-2 \ii \arg \Gamma(-\nu_j)} = -2 \ii \arg \Gamma(-\nu_j).
 %     \end{align}
 %     Substituting this into  \eqref{def-psi} yields:
 %      \begin{align}
	%  \psi(s_j, s_{\sigma(j)}) &= 2 \ii \tilde \theta_j\left(\sqrt{\frac{s_j+s_{\sigma(j)}}{S}}\right)  +\frac{3}{4\pi}\ln (1-\mu_j\mu_{\sigma(j)}) \ln (-4(s_j+s_{\sigma(j)}))\nonumber\\
 %     &\quad+ \arg \Gamma \left(\frac{\ln (1- \mu_j\mu_{\sigma(j)})}{2 \pi \ii}\right) + \frac{\pi}{4}.
	% \end{align}

By substituting \eqref{Psi-J1} and \eqref{eq:beta-1-2-5} into \eqref{eq:beta-1-2-3}, we obtain
   \begin{align}
       [\Psi_{\mathcal J,1}]_{12} &= \ii (-S)^{-\frac 14}\left[\sum_{j \in \bigcup_{k=1}^{m_2}\what J_k\bigcap\J_1}(2 \what r_k)^{-\frac 12}\sqrt{\frac{2}{\pi}}\cos \left(2 \ii \tilde \theta_j\left(\what r_k\right)- \frac{\pi}{4} \right) E_{jj} C\right.\nonumber\\
       &\quad\left.+\sum_{j \in \bigcup_{k=1}^{m_2}\what J_k\bigcap\J_2} \what r_k^{-\frac 12} \sqrt{-\frac{\ln(1-\mu_j\mu_{\sigma(j)})}{\pi \mu_j\mu_{\sigma(j)}}} \cos \psi(s_j, s_{\sigma(j)}) E_{jj} C\right]+\Boh\left((-S)^{-1}\right).
   \end{align}
   A further appeal to 
   \eqref{def:rj} and \eqref{def:what-J_k} yields 
   \begin{align}\label{eq:beta-1-2-6}
       -\ii [\Psi_{\mathcal J,1}]_{12} &= \sum_{j \in \J_1}(-s_j-s_{\sigma(j)})^{-\frac 14}\sqrt{\frac{1}{\pi}}\cos \left(2 \ii \tilde \theta_j\left(\sqrt{\frac{s_j+s_{\sigma(j)}}{S}}\right)- \frac{\pi}{4} \right) E_{jj} C\nonumber\\
       &\quad+\sum_{j \in \J_2}(-s_j-s_{\sigma(j)})^{-\frac 14}\sqrt{-\frac{\ln(1-\mu_j\mu_{\sigma(j)})}{\pi \mu_j\mu_{\sigma(j)}}} \cos \psi(s_j, s_{\sigma(j)}) E_{jj} C+\Boh\left((-S)^{-1}\right).
   \end{align}

    Finally, inserting asymptotics of $[\Psi_{\I,1}]_{12}$ and $[\Psi_{\J,1}]_{12}$ in \eqref{psi-12} and \eqref{eq:beta-1-2-6} into \eqref{beta1}, and applying the expansion in \eqref{expansion}, we obtain the desired asymptotic form \eqref{-infty asympototics}.
    This finishes the proof of Theorem \ref{-infty asympototics results}.
    \qed

\begin{remark}
If $\nu_j =-\frac{1}{2 \pi \ii}\ln (1- \mu_j \mu_{\sigma(j)})$ is purely imaginary and nonzero, we note that
     \begin{align}
       |\Gamma(\nu_j)|^2=|\Gamma(-\nu_j)|^2 = \Gamma(\nu_j)\Gamma(-\nu_j),
     \end{align}
     which implies 
     \begin{align}
         \ln\frac{\Gamma(\nu_j)}{\Gamma(-\nu_j)} = \ln \frac{|\Gamma(\nu_j)|e^{\ii \arg \Gamma(\nu_j)}}{|\Gamma(-\nu_j)|e^{\ii \arg \Gamma(-\nu_j)}}=\ln e^{-2 \ii \arg \Gamma(-\nu_j)} = -2 \ii \arg \Gamma(-\nu_j).
     \end{align}
\end{remark}
\section*{Acknowledgements}
We are grateful to Shuai-Xia Xu for helpful discussions and comments. Luming Yao was partially supported by National Natural Science Foundation of China under grant number 12401316 and Scientific Foundation for Youth Scholars of Shenzhen University under grant number 868--000001032818. Lun Zhang was partially supported by National Natural Science Foundation of China under grant numbers 12271105, 11822104, and ``Shuguang Program'' supported by Shanghai Education Development Foundation and Shanghai Municipal Education Commission.

    \appendix
	
	\section{The Airy parametrix}\label{airy}
	The Airy parametrix $\Phi^{({\Ai})}$ is the unique solution of the following RH problem.
	\begin{rhp}
		\hfill
		\begin{itemize}
			\item[\rm(a)] $\Phi^{(\mathrm{Ai})}(z)$ is analytic in $\mathbb{C} \setminus \{\cup_{j=1}^4 \Sigma_j \cup \{0\}\}$, where the contours $\Sigma_j$, $j=1, 2, 3, 4$, are indicated in Figure \ref{fig:Airy}.
			\item[\rm(b)] For $z \in \cup_{j=1}^3 \Sigma_j$, we have
			\begin{align}\label{jump:Airy}
				\Phi^{({\Ai})}_+(z)=\Phi^{({\Ai})}_-(z)\begin{cases}
					\begin{pmatrix}
						1 & 1\\
						0 & 1
					\end{pmatrix}, & \qquad z \in \Sigma_1,\\
					\begin{pmatrix}
						1 & 0\\
						1& 1
					\end{pmatrix}, & \qquad z \in \Sigma_2 \cup \Sigma_4,\\
					\begin{pmatrix}
						0 & 1\\
						-1 & 0
					\end{pmatrix}, & \qquad z \in \Sigma_3.
				\end{cases}
			\end{align}
			\item[\rm(c)] As $z \to \infty$, we have
			\begin{align}\label{infty:Ai}
				\Phi^{({\Ai})}(z) = \frac{1}{\sqrt{2}} \begin{pmatrix}
					z^{-\frac 14} & 0\\
					0 & z^{\frac 14}
				\end{pmatrix} \begin{pmatrix}
					1 & \ii\\
					\ii & 1
				\end{pmatrix}\left(I + \Boh(z^{-\frac 32})\right)e^{-\frac 23 z^{3/2} \sigma_3}.
			\end{align}
			\item[\rm(d)] $\Phi^{({\Ai})}(z)$ is bounded near the origin.
		\end{itemize}
	\end{rhp}
	
	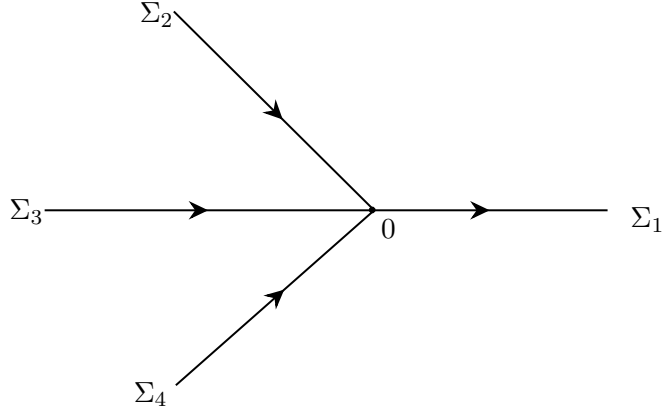
\begin{figure}[t]
		\begin{center}
			
			\tikzset{every picture/.style={line width=0.75pt}} %set default line width to 0.75pt
			
			\begin{tikzpicture}[x=0.75pt,y=0.75pt,yscale=-1,xscale=1]
				
				%Straight Lines [id:da1902587191166466]
				\draw    (52,120) -- (203,120) -- (333,120) ;
				\draw [shift={(132.5,120)}, rotate = 180] [fill={rgb, 255:red, 0; green, 0; blue, 0 }] (7.14,-3.43) -- (0,0) -- (7.14,3.43) -- (4.74,0) -- cycle ;  
				\draw [shift={(273,120)}, rotate = 180] [fill={rgb, 255:red, 0; green, 0; blue, 0 }  ](7.14,-3.43) -- (0,0) -- (7.14,3.43) -- (4.74,0) -- cycle    ;
				%Straight Lines [id:da5833272932983059]
				\draw    (116.5,20) -- (216.5,120) ;
				\draw [shift={(170.04,73.54)}, rotate = 225] [fill={rgb, 255:red, 0; green, 0; blue, 0 }  ](7.14,-3.43) -- (0,0) -- (7.14,3.43) -- (4.74,0) -- cycle    ;
				%Straight Lines [id:da4356022020361009]
				\draw    (117.5,208) -- (216.5,120) ;
				\draw [shift={(170.74,160.68)}, rotate = 138.37] [fill={rgb, 255:red, 0; green, 0; blue, 0 }  ](7.14,-3.43) -- (0,0) -- (7.14,3.43) -- (4.74,0) -- cycle    ;
				%Shape: Circle [id:dp8172020689471965]
				\draw  [fill={rgb, 255:red, 0; green, 0; blue, 0 }  ,fill opacity=1 ] (216.67,119.83) .. controls (216.67,119.19) and (216.15,118.66) .. (215.5,118.66) .. controls (214.85,118.66) and (214.33,119.19) .. (214.33,119.83) .. controls (214.33,120.48) and (214.85,121) .. (215.5,121) .. controls (216.15,121) and (216.67,120.48) .. (216.67,119.83) -- cycle ;
				
				% Text Node
				\draw (218.5,123) node [anchor=north west][inner sep=0.75pt]   [align=left] {$0$};
				% Text Node
				\draw (343,117) node [anchor=north west][inner sep=0.75pt]   [align=left] {$\Sigma_1$};
				% Text Node
				\draw (98,14) node [anchor=north west][inner sep=0.75pt]   [align=left] {$\Sigma_2$};
				% Text Node
				\draw (33,113) node [anchor=north west][inner sep=0.75pt]   [align=left] {$\Sigma_3$};
				% Text Node
				\draw (95,205) node [anchor=north west][inner sep=0.75pt]   [align=left] {$\Sigma_4$};

			\end{tikzpicture}
			\caption{The jump contours of the RH problem for $\Phi^{({\Ai})}$.}
			\label{fig:Airy}
		\end{center}
	\end{figure}
	
	Denote $\omega:= e^{ 2 \pi \ii/3}$, the unique solution is given by (cf. \cite{Deift1999})
	\begin{align}
		\Phi^{({\Ai})}(z) = \sqrt{2 \pi} \begin{cases}
			\begin{pmatrix}
				\Ai (z) & - \omega^2 \Ai (\omega^2 z)\\
				-\ii \Ai'(z) & \ii \omega \Ai' (\omega^2 z)
			\end{pmatrix}, & \arg z \in \left(0, \frac{3 \pi}{4} \right),\\
			\begin{pmatrix}
				-\omega \Ai (\omega z) & - \omega^2 \Ai (\omega^2 z)\\
				\ii \omega^2 \Ai'(\omega z) & \ii \omega \Ai' (\omega^2 z)
			\end{pmatrix}, & \arg z \in \left(\frac{3 \pi}{4}, \pi \right),\\
			\begin{pmatrix}
				-\omega^2 \Ai (\omega^2 z) & \omega \Ai (\omega z)\\
				\ii \omega \Ai'(\omega^2 z) & -\ii \omega^2 \Ai' (\omega z)
			\end{pmatrix}, & \arg z \in \left(-\pi, -\frac{3 \pi}{4} \right),\\
			\begin{pmatrix}
				\Ai (z) &  \omega \Ai (\omega z)\\
				-\ii \Ai'(z) & -\ii \omega^2 \Ai' (\omega z)
			\end{pmatrix}, & \arg z \in \left(-\frac{3 \pi}{4}, 0\right).
		\end{cases}
	\end{align}

    \section{The parabolic cylinder parametrix}\label{sec:PC}
    The parabolic cylinder parametrix $\Phi^{(\PC)}(z) = \Phi^{(\PC)}(z; \nu)$ with $\nu$ being a real or complex parameter is a solution of the following RH problem.

    \begin{rhp}\label{rhp:PC}
		\hfill
		\begin{itemize}
			\item[\rm(a)] $\Phi^{(\mathrm{PC})}(z)$ is analytic in $\mathbb{C} \setminus \{\cup_{j=1}^5 \what\Sigma_j \}$, where the contours \begin{equation}
			    \what \Sigma_j = \{ z \in \mathbb{C}|\arg z = \frac{(j-1)\pi}{2}\}, \quad j = 1, \dots 4, \quad \what \Sigma_5 = \{z \in \mathbb{C}|\arg z = -\frac{\pi}{4}\}
			\end{equation}
            are indicated in Figure \ref{fig:PC}.
			\item[\rm(b)] For $z \in \cup_{j=1}^5 \what\Sigma_j$, we have
			\begin{align}\label{jump:PC}
				\Phi^{({\PC})}_+(z)=\Phi^{({\PC})}_-(z)\begin{cases}
					H_0, & \qquad z \in \what\Sigma_1,\\
					H_1, & \qquad z \in \what\Sigma_2 ,\\
					H_2, & \qquad z \in \what\Sigma_3,\\
                    H_3, & \qquad z \in \what\Sigma_4,\\
                    e^{2 \pi \ii \nu \sigma_3}, & \qquad z \in \what\Sigma_5,\\
				\end{cases}
			\end{align}
            where
            \begin{align}
                H_0 = \begin{pmatrix}
                    1 & 0\\
                    h_0 & 1
                \end{pmatrix}, \quad H_1 = \begin{pmatrix}
                    1 & h_1\\
                    0 & 1
                \end{pmatrix}, \quad H_{k+2} = e^{\pi \ii (\nu + \frac 12)\sigma_3} H_k e^{-\pi \ii (\nu + \frac 12)\sigma_3}, \quad k = 0,1
            \end{align}
            with 
            \begin{align}
                h_0 = -\ii \frac{\sqrt{2 \pi}}{\Gamma(\nu +1)}, \quad h_1 = \frac{\sqrt{2 \pi}}{\Gamma(-\nu)}e^{\pi \ii \nu}, \quad 1 + h_0 h_1 = e^{2 \pi \ii \nu}.
            \end{align}
			\item[\rm(c)] As $z \to \infty$, we have
			\begin{align}\label{infty:PC}
				\Phi^{({\PC})}(z) = \left(I + \begin{pmatrix}
				    0 & \nu\\
                    1 & 0
				\end{pmatrix}\frac{1}{z}+\Boh(z^{-2})\right)e^{\left(\frac{z^2}{4}-\nu \ln z\right) \sigma_3}.
			\end{align}
			%\item[\rm(d)] $\Phi^{({\PC})}(z)$ is bounded near the origin.
		\end{itemize}
	\end{rhp}

     \begin{figure}[t]
		\begin{center}

\tikzset{every picture/.style={line width=0.75pt}} %set default line width to 0.75pt        

\begin{tikzpicture}[x=0.75pt,y=0.75pt,yscale=-1,xscale=1]

%Straight Lines [id:da2696453496951594] 
\draw    (202,124) -- (418,124) ;
\draw [shift={(421,124)}, rotate = 180] [fill={rgb, 255:red, 0; green, 0; blue, 0 }  ](7.14,-3.43) -- (0,0) -- (7.14,3.43) -- (4.74,0) -- cycle    ;
\draw [shift={(199,124)}, rotate = 0] [fill={rgb, 255:red, 0; green, 0; blue, 0 }  ](7.14,-3.43) -- (0,0) -- (7.14,3.43) -- (4.74,0) -- cycle    ;
%Straight Lines [id:da03380555803889718] 
\draw    (309.03,15.5) -- (309.98,121.99) -- (310.97,232.5) ;
\draw [shift={(311,235.5)}, rotate = 269.49] [fill={rgb, 255:red, 0; green, 0; blue, 0 }  ](7.14,-3.43) -- (0,0) -- (7.14,3.43) -- (4.74,0) -- cycle    ;
\draw [shift={(309,12.5)}, rotate = 89.49] [fill={rgb, 255:red, 0; green, 0; blue, 0 }  ](7.14,-3.43) -- (0,0) -- (7.14,3.43) -- (4.74,0) -- cycle    ;
%Straight Lines [id:da7411063680859744] 
\draw    (310,124) -- (389.89,204.87) ;
\draw [shift={(392,207)}, rotate = 225.35] [fill={rgb, 255:red, 0; green, 0; blue, 0 }  ](7.14,-3.43) -- (0,0) -- (7.14,3.43) -- (4.74,0) -- cycle    ;

% Text Node
\draw (297,125) node [anchor=north west][inner sep=0.75pt]   [align=left] {0};
% Text Node
\draw (400,97) node [anchor=north west][inner sep=0.75pt]   [align=left] {$\what\Sigma_1$};
% Text Node
\draw (280,5) node [anchor=north west][inner sep=0.75pt]   [align=left] {$\what\Sigma_2$};
% Text Node
\draw (181,97) node [anchor=north west][inner sep=0.75pt]   [align=left] {$\what\Sigma_3$};
% Text Node
\draw (280,218) node [anchor=north west][inner sep=0.75pt]   [align=left] {$\what\Sigma_4$};
% Text Node
\draw (395,180) node [anchor=north west][inner sep=0.75pt]   [align=left] {$\what\Sigma_5$};

\end{tikzpicture}
			\caption{The jump contours of the RH problem for $\Phi^{({\PC})}$.}
			\label{fig:PC}
		\end{center}
	\end{figure}
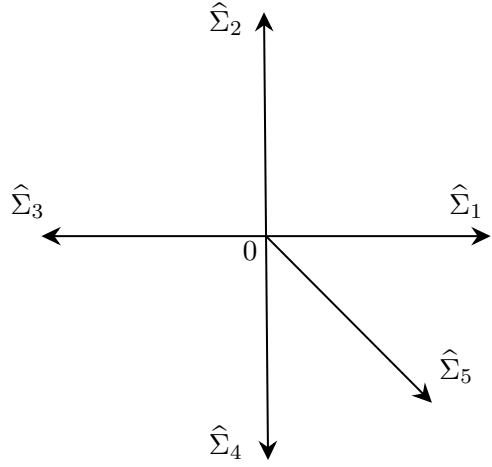
According to \cite[Section 9.4]{Fokas2006}, the solution to above RH problem can be built explicitly in terms of the parabolic cylinder functions (cf. \cite[Chapter 12]{DLMF}) $D_{\nu}$ and $D_{-\nu-1}$. More precisely, we have
\begin{align}\label{PC1}
    \Phi^{({\PC})}(z) = \begin{pmatrix}
        \frac{z}{2} & 1\\
        1 & 0
    \end{pmatrix}\begin{pmatrix}
        D_{-\nu-1}(\ii z)& D_{\nu}(z)\\
        D'_{-\nu-1}(\ii z)& D'_{\nu}(z)
    \end{pmatrix}\begin{pmatrix}
        e^{\frac{\pi \ii}{2}(\nu +1)} & 0\\
        0 & 1
    \end{pmatrix}, \quad \arg z \in \left(-\frac{\pi}{4}, 0\right).
\end{align}
    For $z$ in other sectors of the complex plane, $\Phi^{({\PC})}$ is determined by the jump relation \eqref{jump:PC} and \eqref{PC1}.


\begin{thebibliography}{10} 
        \bibitem{AS1}
        M. J. Ablowitz and H. Segur, Asymptotic solutions of the Korteweg-deVries equation, {\ Stud. Appl. Math.} {57} (1976/77), 13--44.

        \bibitem{AS3}
        M. J. Ablowitz and H. Segur, Exact linearization of a Painlev\'{e} transcendent, { Phys. Rev. Lett.} {38} (1977), 1103--1106.
        
        % \bibitem{BBDI}
        %  J. Baik, R. Buckingham, J. DiFranco and  A.R. Its, Total integrals of global solutions to Painlev\'{e}
        %  II, {\it Nonlinearity}, \textbf{22} (2009), 1021--1061.
    
		\bibitem{MM12}
		M. Bertola and M. Cafasso, Fredholm determinants and pole-free solutions to the noncommutative Painlev\'e II equation, Comm. Math. Phys. {309} (2012),  793--833.
        
        \bibitem{BCR2018}
        M. Bertola, M. Cafasso and V. Rubtsov, Noncommutative painlev\'{e} equations and systems of Calogero type, Comm. Math. Phys. {363} (2018), 503--530.
        
        \bibitem{BCI2016}
        A. Bogatskiy, T. Claey and A. Its, Hankel determinant and orthogonal polynomials for a Gaussian weight with a discontinuity at the edge, Comm. Math. Phys. 347 (2016), 127--162.

        \bibitem{Boh1}
         O. Bohigas, J. X. De Carvalho and M. P. Pato, 
         Deformations of the Tracy-Widom distribution,  Phys. Rev. E 79 (2009), 03117.

         \bibitem{Boh04}
         O. Bohigas and M. P. Pato, 
         Missing levels in correlated spectra, Phys. Lett. B 595 (2004), 171--176.


          \bibitem{Boh06}
          O. Bohigas and M. P. Pato, 
          Randomly incomplete spectra and intermediate statistics,
          Phys. Rev. E  74 (2006), 036212.

         
        \bibitem{Monvel}
        A. Boutet de Monvel, A. Its and  D. Shepelsky, Painlev\'{e}-type asymptotics for the Camassa-Holm  equation,
        SIAM J.  Math.  Anal.  42  (2010), 1854--1873.

        
	\bibitem{Deift1999}
		P. Deift, Orthogonal Polynomials and Random Matrices: A Riemann--Hilbert Approach, Courant Lecture Notes, vol. 3, New York University, 1999.
        
        \bibitem{Deift1993}
		P. Deift and X. Zhou, A steepest descent method for oscillatory Riemann--Hilbert problems. Asymptotics
		for the MKdV equation, Ann. Math. (2) 137 (1993), 295--368.

        \bibitem{DZ95}
        P. Deift and X. Zhou, Asymptotics for the Painlev\'{e} II equation, { Comm. Pure Appl. Math.} {48} (1995), 277--337.

        \bibitem{DXZ25}
        J. Du, S.-X. Xu and Y.-Q. Zhao, Asymptotics of Fredholm determinant solutions of the noncommutative Painlev\'{e} II equation, to appear in Anal. Appl., doi:10.1142/S0219530525500277.

        \bibitem{Duits} M. Duits,
        Painlev\'{e} kernels in Hermitian matrix models,
        Constr. Approx. 39 (2014), 173--196.

        
        \bibitem{FN1980}
         H. Flaschka and A.~C. Newell, Monodromy-and spectrum-preserving deformations I, Comm. Math. Phys. 76 (1980), 65--116.
        
         \bibitem{Fokas2006}
         A. S. Fokas, A. R. Its, A. A. Kapaev, and V. Y. Novokshenov, Painlev\'e Transcendents. The Riemann--Hilbert Approach, Math. Surveys Monographs 128, AMS, Providence, RI, 2006.

        \bibitem{FW2015}
        P.~J. Forrester and N.~S. Witte, Painlev\'e II in random matrix theory and related fields, Constr. Approx. 41 (2015), 589--613.
        

       \bibitem{GGRW}
        I. Gelfand, S. Gelfand, V. Retakh and R.~L. Wilson, 
        Quasideterminants, Adv. Math. 193 (2005), 56--141.

        \bibitem{HM1980}
		S.~P. Hastings and J.~B. McLeod,
		A boundary value problem associated with the second Painlev\'{e} transcendent 
            and the Korteweg-de Vries equation,
		Arch. Ration. Mech. Anal. 73 (1980), 31--51.
        
        \bibitem{Ince:book} 
        E. L. Ince, Ordinary Differential Equations, Dover Publications, New York, 1944.

        \bibitem{IIKS}
        A.~R. Its, A.~G. Izergin, V.~E. Korepin and N.~A. Slavnov, 
        Differential equations for quantum correlation functions, Internat. J. Modern Phys. B 4 (1990), 1003--1037.
		
		\bibitem{IP2020}
         A. R. Its and A. Prokhorov,  On $\beta=6$ Tracy-Widom distribution and the second Calogero-Painlev\'{e} system, preprint arXiv:2010.06733.

        \bibitem{DLMF}
        F. W. J. Olver, A. B. Olde Daalhuis, D. W. Lozier, B. I. Schneider, R. F. Boisvert, C. W. Clark, B. R. Miller, B. V. Saunders, H. S. Cohl, and M. A. McClain, eds., NIST Digital Library of Mathematical Functions, Release 1.2.4, http://dlmf.nist.gov/, (2025).

        

         \bibitem{RR2010}
         V. Retakh and V. Rubtsov, Noncommutative Toda chains, Hankel quasideterminants and Painlev\'{e} II equation, J. Phys. A: Math. Theor. 43 (2010), 505204. 
         
         \bibitem{Rum1} 
         I. Rumanov, Classical integrability for beta-ensembles and general Fokker-Planck equations, J. Math. Phys., 
         {56} (2015), 013508, 16pp.% arXiv:1306.2117.


         \bibitem{Rum2} I. Rumanov, 
         Painlev\'e Representation of $\text{Tracy-Widom}_{\beta}$  distribution for $\beta$  = 6, Comm. Math. Phys. 342 (2016), 843--868.

        \bibitem{SA}
		 H. Segur and M. J. Ablowitz, Asymptotic solutions of nonlinear evolution equations and a Painlev\'{e} transcedent, Physica D 3 (1981), 165--184.

        %  \bibitem{SUL}
        % B. I. Suleimanov, The second Painlev\'{e} equation in a problem on nonlinear effects near caustics, Zap. Nauchn. Sem. Leningrad. Otdel. Mat. Inst. Steklov. (LOMI) 187 (1991), 110--128; English transl., J. Math. Sci. (N.Y.) 73 (1995), no. 4, 482--493.
        %Suleimanov, I. The second Painlev\'{e} equation in a problem concerning nonlinear effects near caustics. J Math Sci 73, 482–493 (1995).

        \bibitem{TW1994}
        C. Tracy and H. Widom, Level spacing distributions and the Airy kernel, Comm. Math. Phys. 159 (1994), 151--174.

        \bibitem{XYZ}
        T.-Y. Xu, Y.-L. Yang and L. Zhang,
        Transient asymptotics of the modified Camassa-Holm equation,
        J. Lond. Math. Soc. 110 (2024), e12967.
	\end{thebibliography}
\end{document}